\documentclass[letterpaper]{article}
\usepackage{preamble}

\usepackage{graphicx}    \usepackage{caption} 

\usepackage{algorithm}
\usepackage[noend]{algorithmic}

\usepackage{amsmath,amsthm,amsfonts,amssymb}
\usepackage{adjustbox}
\usepackage{thmtools} 
\usepackage{subcaption}
\usepackage{authblk}
\usepackage{thm-restate}
\usepackage{footnote}
\makesavenoteenv{tabular, center}
\usepackage{mathtools}
\usepackage[capitalize]{cleveref}
\usepackage{xcolor}
\usepackage{todonotes}
\usepackage{threeparttable}
\usepackage{dsfont}
\usepackage{microtype}
\usepackage{diagbox}  \usepackage{tikz}
\usetikzlibrary{decorations.pathreplacing}
\usetikzlibrary{arrows.meta,decorations.pathmorphing}
\usepackage{enumitem}
\setlist{noitemsep}
\setlist[description]{leftmargin=\parindent,labelindent=0pt,topsep=2pt}
\setlist[itemize]{leftmargin=\parindent,labelindent=0pt,topsep=2pt}
\usepackage{newfloat}
\usepackage{listings}
\DeclareCaptionStyle{ruled}{labelfont=normalfont,labelsep=colon,strut=off} 
\floatstyle{ruled}
\newfloat{listing}{tb}{lst}{}
\floatname{listing}{Listing}

\usepackage{booktabs}

\newcommand{\PAV}{\operatorname{PAV}}
\newcommand{\norm}[1]{\lVert#1\rVert}
\newcommand{\lift}[1]{\operatorname{lift}(#1)}

\title{Approval-Based Apportionment: \\ Like Portioning, Approximately like Committee Voting}
\author{Paul Gölz}
\author{Hannane Yaghoubizade}
\affil{Cornell University, School of Operations Research and Information Engineering}
\date{}

\begin{document}
\maketitle

\begin{abstract}
We study approval-based apportionment, a variant of committee elections in which candidates (``parties'') can be selected several times.
We show that the proportionality axioms EJR, EJR+, and FJR coincide, and so do PJR, PJR+, and FPJR; that Lindahl priceability, an axiom implying core stability, is equivalent to a notion of approximate optimality with respect to the proportional approval voting (PAV) score; and that locally PAV-optimal committees are priceable.

Approval-based apportionment (where a candidate receives an integer number of seats) lies between committee elections (zero or one seat) and portioning (a fractional number of seats).
We formally connect portioning and apportionment by giving a construction that lifts axioms from apportionment to portioning and preserves implications between them.
Several of our new implications between apportionment axioms are natural from a portioning perspective, leading us to believe that apportionment sits closer to portioning than to committee elections.
None of them holds in committee elections, but several extend approximately, which makes apportionment a fruitful setting for conjecturing approximate relationships in approval-based committee elections.
\end{abstract}

\section{Introduction}
\definecolor{committeecolor}{HTML}{009E73}
\definecolor{portioningcolor}{HTML}{E69F00}      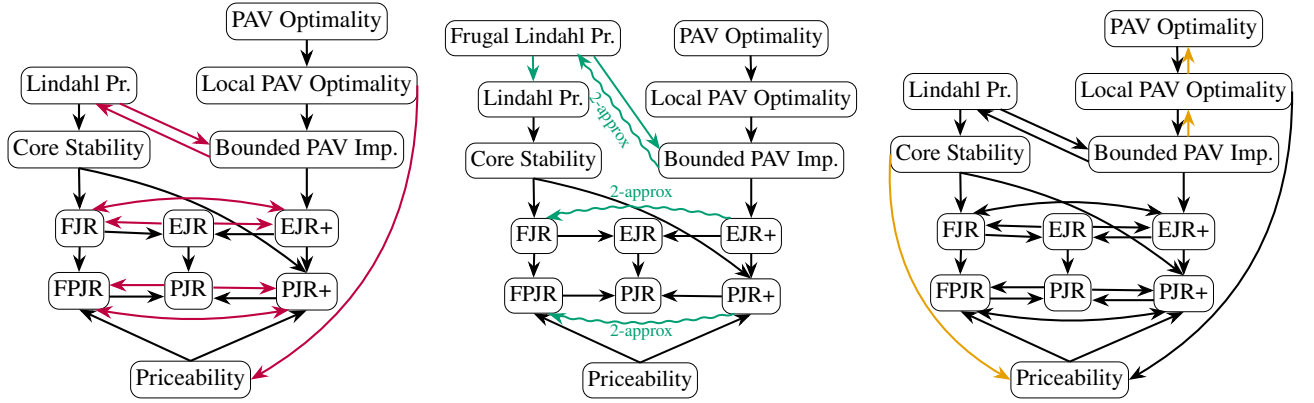
\begin{figure*}[t]
    \centering

    \begin{subfigure}[t]{0.32\textwidth}
        \centering
        \begin{adjustbox}{max width=\linewidth,center}
            \begin{tikzpicture}[
     x=0.9cm,
    y=1.1cm,
    box/.style={
        draw,
        rounded corners,
        very thin,
        node font=\footnotesize,
        minimum height=14pt,
        inner xsep=2pt,
        inner ysep=1pt,
        align=center
    }
]
  \node[box] (node1) at (-3.57,3.96) {Lindahl Pr.\vphantom{Ty}};
  \node[box] (node2) at (-3.57,3.18) {Core Stability\vphantom{Ty}};
  \node[box] (node5) at (-3.56,2.22) {FJR\vphantom{Ty}};
  \node[box] (node12) at (-1.94,2.22) {EJR\vphantom{Ty}};
  \node[box] (node6) at (-0.18,2.22) {EJR+\vphantom{Ty}};

  \draw[arrows=-Stealth, thick] (node1.south) -- (node2.north);
  \draw[arrows=-Stealth, thick] (-3.19,2.18) -- (-2.32,2.15);
  \draw[arrows=-Stealth, thick] (-0.67,2.15) -- (-1.56,2.15);

  \node[box] (node7) at (-3.57,1.46) {FPJR\vphantom{Ty}};
  \node[box] (node11) at (-1.94,1.46) {PJR\vphantom{Ty}};
  \node[box] (node8) at (-0.18,1.45) {PJR+\vphantom{Ty}};
  \node[box] (node10) at (-1.91,0.38) {Priceability\vphantom{Ty}};

  \draw[arrows=-Stealth, thick] (node2.south) -- (node5.north);
  \draw[arrows=-Stealth, thick] (node5.south) -- (node7.north);
  \draw[decorate, thick, arrows=-Stealth, bend left=17]
      (node2.south) to[bend left=13] (node8.north);
  \draw[arrows=-Stealth, thick] (node12.south) -- (node11.north);
  \draw[arrows=-Stealth, thick] (node10.north) -- (node8.south);
  \draw[arrows=-Stealth, thick] (node10.north) -- (node7.south);
  \draw[arrows=-Stealth, thick] (node6.south) -- (node8.north);

  \node[box] (node3) at (-0.18,3.18) {Bounded PAV Imp.\vphantom{Ty}};
  \node[box] (node4) at (-0.18,3.96) {Local PAV Optimality\vphantom{Ty}};

  \draw[arrows=-Stealth, thick, line cap=projecting]
      (node4.south) -- (node3.north);
  \draw[thick, line cap=projecting, arrows=-Stealth, purple]
      (-1.56,2.28) -- (-0.68,2.27);
  \draw[arrows=-Stealth, thick]
      (-0.63,1.37) -- (-1.57,1.37);
  \draw[thick, line cap=projecting, arrows=-Stealth, purple]
      (-1.57,1.5) -- (-0.63,1.49);

  \draw[arrows=-Stealth, thick, line cap=projecting]
      (node3.south) -- (node6.north);

  \draw[thick, line cap=projecting, purple, arrows=Stealth-]
      (-3.19,2.3) -- (-2.34,2.28);

  \draw[arrows=-Stealth, thick]
      (-3.12,1.39) -- (-2.3,1.39);
  \draw[thick, line cap=projecting, purple, arrows=Stealth-]
      (-3.12,1.53) -- (-2.3,1.53);

  \draw[decorate, thick, arrows=Stealth-Stealth, purple, bend right=15]
      (-0.5,2.45) to[bend right=11] (-3.38,2.45);

  \draw[decorate, thick, arrows=Stealth-Stealth, purple, bend left=15]
      (-0.44,1.23) to[bend left=9] (-3.33,1.23);

  \draw[thick, purple, arrows=-Stealth]
      (-2.96,3.74) -- (-1.62,3.24);
  \draw[thick, purple, arrows=Stealth-]
      (-3.33,3.74) -- (-1.62,3.09);
  \draw[arrows=-Stealth, thick, purple, bend left=50]
      (node4.east) to[bend left=30] (node10.east);

  \node[box] (node9) at (-0.18,4.74) {PAV Optimality\vphantom{Ty}};

  \draw[arrows=-Stealth, thick, line cap=projecting]
      (node9.south) -- (node4.north);

\end{tikzpicture}
        \end{adjustbox}
        \caption{Implications in \textbf{apportionment}. Black arrows are inherited from the committee election setting, and the red arrows are established in this work.\footnotemark}
        \label{fig:abctoapp}
    \end{subfigure}
    \hfill
    \begin{subfigure}[t]{0.32\textwidth}
        \centering
        \begin{adjustbox}{max width=\linewidth,center}
            \begin{tikzpicture}[
     x=0.9cm,
    y=1.1cm,
    box/.style={
        draw,
        rounded corners,
        very thin,
        node font=\footnotesize,
        minimum height=14pt,
        inner xsep=2pt,
        inner ysep=1pt,
        align=center
    }
]
  \node[box] (node1) at (-3.57,3.96) {Lindahl Pr.\vphantom{Ty}};
  \node[box] (node13) at (-3.58,4.74) {Frugal Lindahl Pr.\vphantom{Ty}};
  \node[box] (node2) at (-3.57,3.18) {Core Stability\vphantom{Ty}};
  \node[draw=none, node font=\scriptsize, text=committeecolor, fill opacity=1, draw opacity=1] at (-1.89,2.77) {2-approx};
  \node[box] (node5) at (-3.56,2.22) {FJR\vphantom{Ty}};
  \node[box] (node12) at (-1.94,2.22) {EJR\vphantom{Ty}};
  \node[box] (node6) at (-0.18,2.22) {EJR+\vphantom{Ty}};

  \draw[arrows=-Stealth, thick] (node1.south) -- (node2.north);
  \draw[arrows=-Stealth, thick] (node5.east) -- (node12.west);
  \draw[arrows=-Stealth, thick] (node6.west) -- (node12.east);

  \node[box] (node7) at (-3.57,1.46) {FPJR\vphantom{Ty}};
  \node[box] (node11) at (-1.94,1.46) {PJR\vphantom{Ty}};
  \node[box] (node8) at (-0.18,1.45) {PJR+\vphantom{Ty}};
  \node[box] (node10) at (-1.91,0.38) {Priceability\vphantom{Ty}};

  \draw[arrows=-Stealth, thick] (node2.south) -- (node5.north);
  \draw[arrows=-Stealth, thick] (node5.south) -- (node7.north);
  \draw[decorate, thick, arrows=-Stealth, bend left=17]
      (node2.south) to[bend left=13] (node8.north);
  \draw[arrows=-Stealth, thick] (node12.south) -- (node11.north);
  \draw[arrows=-Stealth, thick] (node10.north) -- (node8.south);
  \draw[arrows=-Stealth, thick] (node10.north) -- (node7.south);
  \draw[arrows=-Stealth, thick] (node6.south) -- (node8.north);

  \node[box] (node3) at (-0.18,3.18) {Bounded PAV Imp.\vphantom{Ty}};
  \node[box] (node4) at (-0.18,3.96) {Local PAV Optimality\vphantom{Ty}};

  \draw[arrows=-Stealth, thick, line cap=projecting]
      (-0.18,3.74) -- (-0.18,3.4);
  \draw[arrows=-Stealth, thick]
      (node8.west) -- (node11.east);

  \draw[arrows=-Stealth, thick, line cap=projecting]
      (node3.south) -- (node6.north);

  \draw[arrows=-Stealth, thick]
      (node7.east) -- (node11.west);

  \draw[decorate, thick, arrows=-Stealth, bend right=15, committeecolor, decorate, decoration={snake}, /pgf/decoration/segment length=7pt, /pgf/decoration/amplitude=0.35pt]
      (-0.5,2.45) to[bend right=11] (-3.38,2.45);

  \draw[decorate, thick, arrows=-Stealth, bend left=15, committeecolor, decorate, decoration={snake}, /pgf/decoration/segment length=7.0pt, /pgf/decoration/amplitude=0.35pt]
      (-0.44,1.23) to[bend left=9] (-3.33,1.23);

  \draw[thick, arrows=-Stealth, committeecolor]
      (-2.63,4.51) -- (-1.60,3.32);
  \draw[thick, arrows=Stealth-, committeecolor, decorate, decoration={snake}, /pgf/decoration/segment length=6.6pt, /pgf/decoration/amplitude=0.4pt]
      (-2.89,4.55) -- (-1.62,3.09);

  \node[box] (node9) at (-0.18,4.74) {PAV Optimality\vphantom{Ty}};

  \draw[arrows=-Stealth, thick, line cap=projecting]
      (-0.19,4.52) -- (-0.18,4.18);

  \draw[arrows=-Stealth, thick, committeecolor] (node13.south) -- (node1);

  \node[draw=none, node font=\scriptsize, text=committeecolor, fill opacity=1, draw opacity=1, rotate=-56] at (-2.39,3.69) {2-approx};

  \node[draw=none, node font=\scriptsize, text=committeecolor, fill opacity=1, draw opacity=1] at (-1.89,1) {2-approx};
\end{tikzpicture}
        \end{adjustbox}
        \caption{Implications in \textbf{committee elections}, including some implying multiplicatively relaxed properties. Black arrows were previously known; green arrows are new in this work.}
        \label{fig:diag-party}
    \end{subfigure}
    \hfill
    \begin{subfigure}[t]{0.32\textwidth}
        \centering
        \begin{adjustbox}{max width=\linewidth,center}
            \begin{tikzpicture}[
     x=0.9cm,
    y=1.1cm,
    box/.style={
        draw,
        rounded corners,
        very thin,
        node font=\footnotesize,
        minimum height=14pt,
        inner xsep=2pt,
        inner ysep=1pt,
        align=center
    }
]
  \node[box] (node1) at (-3.57,3.96) {Lindahl Pr.\vphantom{Ty}};
  \node[box] (node2) at (-3.57,3.18) {Core Stability\vphantom{Ty}};
  \node[box] (node5) at (-3.56,2.22) {FJR\vphantom{Ty}};
  \node[box] (node12) at (-1.94,2.22) {EJR\vphantom{Ty}};
  \node[box] (node6) at (-0.18,2.22) {EJR+\vphantom{Ty}};

  \draw[arrows=-Stealth, thick] (node1.south) -- (node2.north);
  \draw[arrows=-Stealth, thick] (-3.19,2.18) -- (-2.32,2.15);
  \draw[arrows=-Stealth, thick] (-0.67,2.15) -- (-1.56,2.15);

  \node[box] (node7) at (-3.57,1.46) {FPJR\vphantom{Ty}};
  \node[box] (node11) at (-1.94,1.46) {PJR\vphantom{Ty}};
  \node[box] (node8) at (-0.18,1.45) {PJR+\vphantom{Ty}};
  \node[box] (node10) at (-1.91,0.38) {Priceability\vphantom{Ty}};

  \draw[arrows=-Stealth, thick] (node2.south) -- (node5.north);
  \draw[arrows=-Stealth, thick] (node5.south) -- (node7.north);
  \draw[decorate, thick, arrows=-Stealth, bend left=17]
      (node2.south) to[bend left=13] (node8.north);
  \draw[arrows=-Stealth, thick] (node12.south) -- (node11.north);
  \draw[arrows=-Stealth, thick] (node10.north) -- (node8.south);
  \draw[arrows=-Stealth, thick] (node10.north) -- (node7.south);
  \draw[arrows=-Stealth, thick] (node6.south) -- (node8.north);

  \node[box] (node3) at (-0.18,3.18) {Bounded PAV Imp.\vphantom{Ty}};
  \node[box] (node4) at (-0.18,3.96) {Local PAV Optimality\vphantom{Ty}};

  \draw[arrows=-Stealth, thick, line cap=projecting]
      (-0.27,3.74) -- (-0.27,3.4);
  \draw[thick, line cap=projecting, arrows=-Stealth]
      (-1.56,2.28) -- (-0.68,2.27);
  \draw[arrows=-Stealth, thick]
      (-0.63,1.37) -- (-1.57,1.37);
  \draw[thick, line cap=projecting, arrows=-Stealth]
      (-1.57,1.5) -- (-0.63,1.49);

  \draw[arrows=-Stealth, thick, line cap=projecting]
      (node3.south) -- (node6.north);

  \draw[thick, line cap=projecting, arrows=Stealth-]
      (-3.19,2.3) -- (-2.34,2.28);

  \draw[arrows=-Stealth, thick]
      (-3.12,1.39) -- (-2.3,1.39);
  \draw[thick, line cap=projecting, arrows=Stealth-]
      (-3.12,1.53) -- (-2.3,1.53);

  \draw[decorate, thick, arrows=Stealth-Stealth, bend right=15]
      (-0.5,2.45) to[bend right=11] (-3.38,2.45);

  \draw[decorate, thick, arrows=Stealth-Stealth, bend left=15]
      (-0.44,1.23) to[bend left=9] (-3.33,1.23);

  \draw[thick, arrows=-Stealth]
      (-2.96,3.74) -- (-1.62,3.24);
  \draw[thick, arrows=Stealth-]
      (-3.33,3.74) -- (-1.62,3.09);
  \draw[arrows=-Stealth, thick, bend left=50]
      (node4.east) to[bend left=30] (node10.east);

  \node[box] (node9) at (-0.18,4.74) {PAV Optimality\vphantom{Ty}};

  \draw[arrows=-Stealth, thick, line cap=projecting]
      (-0.28,4.52) -- (-0.27,4.18);

  \draw[arrows=-Stealth, thick, bend left=50, portioningcolor]
        (node2.west) to[bend right=32] (node10.west);

  \draw[arrows=-Stealth, thick, portioningcolor] (-0.11,3.4) -- (-0.11,3.74);

  \draw[arrows=-Stealth, thick, portioningcolor] (-0.11,4.18) -- (-0.11,4.52);
\end{tikzpicture}
        \end{adjustbox}
        \caption{Implications between \textbf{portioning} analogs of apportionment properties, formally defined in \cref{sec:portioning:lift}. Black arrows also hold in apportionment; yellow ones only in portioning.}
        \label{fig:portioning}
    \end{subfigure}

    \caption{Implication relations between properties in approval-based apportionment, committee elections, and portioning.}
    \label{fig:implication-relations}
\end{figure*}

The work by \textcite{JR} on proportionality in approval-based committee elections has proved deeply influential on the subsequent decade of social choice research.
In these committee elections, each voter $1 \leq i \leq n$ approves of a subset $A_i$ of the candidates $C$; the task is to choose a committee of $k$ candidates; and the challenge is to capture and attain the ideal of proportionality in this setting.
The work following up on \textcite{JR} has led to an explosion of proportionality axioms.
To name just a few, their axioms of \emph{justified representation (JR)} and \emph{extended justified representation (EJR)} have since been joined by variants abbreviated \emph{PJR}~\cite{PJR}, \emph{FJR}~\cite{FJR}, \emph{PJR+} and \emph{EJR+}~\cite{EJR+}, \emph{FPJR}~\cite{FPJR}, \emph{BJR}~\cite{FGP+26}, and Droop-quota versions of most of these~\cite{CE26}; alongside notions of core stability~\cite{JR}, priceability~\cite{limitsofwelf}, and others.

Faced with such an embarrassment of axiomatic riches, it becomes more urgent to learn about the relationship between axioms.
By this, we mean going beyond just their implication hierarchy, which is well understood.
For example, FJR implies EJR, which in turn implies PJR. \emph{Should we think of EJR as conceptually closer to FJR or to PJR?}
To take another example, a committee $W$ maximizing the so-called PAV score (to be formally defined later) attains strong proportionality axioms like EJR+~\cite{EJR+}.
\emph{Is the PAV score fundamentally related to proportionality axioms in this setting, or just one of many ways to attain them?}

Our approach to get at such (a priori, ill-defined) structure resembles the use of \emph{model organisms} in biology, where observations in a simpler organism (say, a mouse) yield conjectured patterns that can then be tested in a more complicated organism (say, a human).
In the same way, we study a related social choice setting with slightly simpler mathematical structure, in which more theorems (say, implications between axioms) hold true.
We then go through theorems that hold in the simpler setting but not in full committee elections, and ask if they extend in some approximate form.

Our ``model organism'' for committee elections is \emph{approval-based apportionment}~\cite{paulapproval}.
The only difference is that the same candidate may be placed several times in the committee; whereas committee elections allocate zero or one seat to each candidate (and allocate $k$ seats in total), apportionment can give any natural number of seats to each candidate (under the same constraint).
Alternatively to this presentation as a relaxation of committee elections, apportionment can also be seen as a subdomain consisting only of committee elections in which sufficiently many copies of each candidate are available.
Through this embedding, \textcite{paulapproval} formalize how to lift axioms and voting rules from committee elections to apportionment.
Properties that hold for all committee-election instances must hold for all apportionment instances, so the latter setting lets us prove more theorems.
In this sense, apportionment is a ``simpler organism'', which we study as a proxy for committee elections.

Whereas this introduction focuses on the benefits of studying apportionment for understanding committee elections, the apportionment setting is well motivated in its own right.
If the candidates are political parties, apportionment maps voters' approval preferences over them to an allocation of seats in a legislature to the parties~\cite{paulapproval}, which generalizes classic apportionment~\cite{BY01} in which each voter can only support a single party.
Another application is representing opinions with a small number of AI-generated textual statements, where \textcite{FGP+26} argue that ``representing multiple segments of the population by identical statements might sometimes be\footnotetext[1]{It was previously known that bounded PAV improvement implies core stability~\cite{paulapproval}, which we strengthen by showing that it even implies Lindahl priceability.} appropriate.''\footnote{Whereas the model of \textcite{FGP+26} assumes cardinal preferences, an earlier preprint uses approval preferences.}

\subsection{Our Results and Techniques}
We begin by showing three new results for the apportionment setting in \cref{sec:apportionment}. These new implications are represented as red arrows in \cref{fig:abctoapp}, in addition to implications inherited from the committee elections setting (black arrows):
\begin{itemize}
\item The proportionality axioms EJR, EJR+, and FJR coincide, and so do PJR, PJR+, and FPJR. 
Answering the question from our introduction, this is a formal way in which EJR is closer to FJR than to PJR.
\item The axiom \emph{Lindahl priceability}~\cite{munagala2022auditing} is equivalent to what we call \emph{bounded PAV improvement}, which says that the PAV score of a committee $W$ cannot increase by $n/k$ or more when adding one more candidate. Bounded PAV improvement is the key property in proofs that the PAV voting rule (and its local-search variants) satisfy EJR and EJR+~\cite{JR,EJR+}. This result shows an intrinsic link between the PAV objective and strong notions of proportionality.
\item Any committee maximizing (or locally maximizing) the PAV score is priceable. This is surprising because Pareto-optimal, ``welfarist'' rules like PAV cannot guarantee priceability in committee elections~\cite{limitsofwelf}.
\end{itemize}
This yields three implications that ``almost'' hold in general committee elections, and only fail to go through because desirable candidates cannot be chosen several times.

To better understand the similarity of implication relationships in apportionment to those in other settings, we consider one more setting in \cref{sec:portioning}, \emph{approval-based portioning}~\cite{BMS05}.
In this setting, the number of ``seats'' allocated to each candidate can be fractional rather than integer. (Their sum is normalized to $1$ rather than $k$.)
Whereas \textcite{paulapproval} discuss approval-based apportionment as lying between committee elections and classic (single-party choice) apportionment, we believe that portioning is the more natural second endpoint.

Formalizing a relationship used by \textcite{Peters25} to show that PAV converges to the maximum Nash welfare portioning, we define a way to lift predicates over committees (axioms or voting rules) from apportionment to portioning, broadly speaking by considering the limit as the committee size $k$ goes to infinity.
We show that axioms and voting rules from committee elections map to natural analogs in portioning.\footnote{E.g., the \emph{method of equal shares}~\cite{limitsofwelf} lifts to \emph{majoritarian portioning}~\cite{Majoritarian-portioning}.}
Any theorem in apportionment stating that there always exists a committee with a property $X$ or that a committee property $X$ implies a property $Y$ extends to the lifted versions of these properties in portioning, so portioning is, in turn, a ``simpler'' setting than apportionment.
In our impression, the network of axiom implications of apportionment seems closer to that of portioning than to that of committee elections.
This suggests that the committee's indivisibility matters less than whether candidates may be selected several times.

In \cref{sec:abc}, we return to committee elections to show that, though the implications we found in apportionment do not hold in committee elections, several of them \emph{approximately} hold in this more general setting (\cref{fig:diag-party}).
First, we find that EJR+ implies 2-approximate FJR and PJR+ implies 2-approximate FPJR, a question that we would not have thought to ask without the result in apportionment.
Second, bounded PAV improvement is closely linked to a natural strengthening of Lindahl priceability (\emph{frugal Lindahl priceability}, equivalent in apportionment), which implies bounded PAV improvement, and whose 2-approximate notion is implied by bounded PAV improvement.

\subsection{Related Work}
Our work owes much to the literature on committee elections.
We reference relevant works throughout and refer the reader to \textcite{LS23} for a comprehensive treatment.

Using approval ballots for political apportionment has been proposed many times, including by \textcite{BKP19} and \textcite{Majoritarian-portioning}.
\textcite{BLS18} observe that (single-party choice) apportionment is a subdomain of approval-based committee elections.

\textcite{paulapproval} systematically study the approval-based apportionment setting.
They formalize how to lift axioms and voting rules from committee elections to apportionment, similar to the lifting we give from apportionment to portioning.
They also show that PAV-optimal committees satisfy \emph{core stability}, whereas the satisfiability of core stability in general committee elections is a major open problem.
While they discuss portioning as a source of possible apportionment methods (by rounding the portioning), they do not develop the connection between the models.
\textcite{DDE+23} show that strategyproofness and weak forms of proportionality are incompatible in apportionment, and characterize the Chamberlin--Courant rule using a weak form of strategyproofness.
They also discuss apportionment as an intermediate model between committee elections and portioning in terms of their output types.
As mentioned, \textcite{Peters25} identifies the connection between PAV and maximum Nash welfare in portioning and inspires our lifting between these settings.

Other papers touch on the apportionment setting as a special case of their model: \textcite{BBG22} discuss how their \emph{excess method} rule coincides with sequential PAV; \textcite{CGP24} discuss apportionment as a special case of their sequential-decision model; and \textcite{FGP+26} show that the BJR axiom is incomparable with other strengthenings of JR in apportionment.

Portioning (or \emph{fair mixing}) was introduced by \textcite{BMS05}, motivated by time sharing.
Since the outcome space is convex, microeconomic tools like the Lindahl equilibrium readily apply to this setting~\cite{FGM16}.
We refer the reader to \textcite{ST26} for an overview of this literature.

\section{Apportionment}
\label{sec:apportionment}
\subsection{Preliminaries}\label{sec:app:pre}
An apportionment instance is a tuple $(N, P, A, k)$ consisting of a set of voters $N = \{1, \dots, n\}$, a finite set of parties $P$, a tuple $A=(A_i)_{i\in N}$ where $A_i \subseteq P$ is the approval set of $i$, and the committee size $k \in \mathbb{N}$.
Set $q(S) \coloneqq \lfloor|S| \cdot k /n\rfloor$ for the (Hare) \emph{quota} of a set of voters $S\subseteq N$.
By default, $i$ ranges over $N$ and $j$ over $P$.
For ease of exposition, we assume that every voter $i$ approves at least one party.
A \emph{committee} is a multiset $W: P \to \mathbb{N}$ of size $|W| = \sum_{j} W(j) = k$.
As usual, $W(j)$ denotes the number of copies of $j$ in $W$, and $+$ and $-$ denote multiset addition and subtraction.
An \emph{apportionment method} is a function mapping each apportionment instance to a committee.
We write $u_i(W) \coloneqq \sum_{j\in A_i} W(j)$ for voter $i$'s \emph{utility} for committee $W$, which simply counts the number of seats filled by approved parties.

We now define several proportionality axioms.
Each of them is lifted from its committee-elections definition to the apportionment setting using the embedding of \textcite{paulapproval}.\footnote{Whereas the $k$ copies per party in their embedding are ``sufficiently many'' for their considered axioms, we use $k+1$ copies.}
Some lifted definitions include non-trivial simplifications, so we formally derive them in \cref{app:appor:def}.
This appendix also includes definitions for additional axioms, which we defer due to lack of space.

Fix an apportionment instance and a committee $W$. Then:
\begin{description}
    \item[EJR/EJR+.] $W$ satisfies \emph{extended justified representation (plus)}\footnote{These definitions collapse, see the proof of \cref{thm:efjr+}.}
    if, for every nonempty $S \subseteq N$ with $\bigcap_{i \in S} A_i \neq \emptyset$, there exists a voter $i \in S$ with $u_i(W) \geq q(S)$.
    \item[FJR.] $W$ satisfies \emph{full justified representation} if, for every nonempty $S \subseteq N$, every multiset $T$ of parties of size $|T| \leq q(S)$, and every $\beta \in \mathbb{N}$ such that $u_i(T) \geq \beta$ for all $i \in S$, there is an $i \in S$ such that $u_i(W) \geq \beta$.
    \item[Lindahl Priceability.] $W$ is \emph{Lindahl priceable} if there exist nonnegative personalized prices $(p_{ij})_{i,j}$ such that $\sum_{i} p_{ij} \leq n/k$ for all $j$ and such that, for any voter $i$ and multiset $T$ of parties with $\sum_{j} T(j) \, p_{ij} \leq 1$, it holds that $u_i(T) \leq u_i(W)$. (Conceptually, $i$ cannot afford any outcome they prefer over $W$ within a budget of $1$.)
    \item[Priceability.] $W$ is \emph{priceable} if there exist nonnegative payments $(p_{ij})_{i,j}$ and a per-unit price $p \geq 0$ such that $p_{ij}=0$ whenever $j \notin A_i$, $\sum_{j} p_{ij} \leq 1$ for all $i$, $\sum_{i} p_{ij} = W(j) \cdot p$ for all $j$, and $\sum_{i : j \in A_i} (1 -\sum_{j' \in P} p_{ij'}) \leq p$ for all $j$. (The unspent budget of $j$'s supporters is at most $p$.)
\end{description}

\noindent An apportionment method satisfies an axiom if it only produces committees satisfying that axiom.

The \emph{proportional approval voting (PAV)} rule~\cite{Thiele95} and its objective play a prominent role in our results.
The \emph{PAV score} of a committee $W$ is $\PAV{}(W) \coloneqq \sum_{i\in N} H_{u_i(W)}$ where $H_t \coloneqq \sum_{\ell=1}^t 1/\ell$.
\(W\) is \emph{PAV optimal} if it maximizes the PAV score among all committees and is \emph{locally PAV optimal} if no swap of one selected party for another strictly increases its PAV score. 
For each $j$, we denote the marginal PAV gain from adding one copy of $j$ to $W$ (leading to a multiset of size $k+1$) by
\(
\Delta_W^+(j)\coloneqq \PAV{}(W\!+\!\{j\})-\PAV{}(W)
=\sum_{i:\,j\in A_i} 1/(u_i(W)\!+\!1).
\)
We say that a committee $W$ satisfies \textbf{bounded PAV improvement} if $\Delta_W^+(j) < n/k$ for all $j$.
While it may be unintuitive to speak of the increase of the PAV score when adding a $(k+1)$th seat, this property (and its committee-election analog, see \cref{sec:abc:prelims}) plays a key role, in the sense that it is implied by PAV optimality\footnote{It is even implied by local PAV optimality, and the approximate local PAV optimality reached by the polynomial-time local-search variant of PAV~\cite{AEH+18}.} and is sufficient to prove proportionality axioms like EJR+.

\subsection{New Implications Between Properties}
In this section, we prove three sets of new implication relations between committee properties in the apportionment setting (red arrows in \cref{fig:abctoapp}).
Recall that, due to the embedding of apportionment into committee elections, all implications between axioms from committee elections hold in apportionment as well (black arrows in the figure).
We verify in \cref{app:appor:omit} that no other pair of axioms in the figure implies each other, so the picture of pairwise relationships of the displayed properties is complete.

Our point is not that any of the three results is hard to prove.
Instead, these implications contradicted our expectations (formed in committee elections) and thus pushed us to reevaluate the structure of approval-based apportionment.

\medskip\noindent\textbf{FJR$\Leftrightarrow$EJR$\Leftrightarrow$EJR+ and FPJR$\Leftrightarrow$PJR$\Leftrightarrow$PJR+.}
We begin with a simple observation that several of the strengthenings of justified representation, though distinct in committee elections, collapse in apportionment.
\begin{theorem}\label{thm:efjr+}
   In the apportionment setting, FJR, EJR, and EJR+ coincide, as do FPJR, PJR, and PJR+.
\end{theorem}
\begin{proof}
    The collapse between EJR+ and EJR is so immediate that we even gave identical definitions for the apportionment setting.
    In committee elections, the two axioms differ in which groups $S$ of voters are ``cohesive'' enough to deserve representation guarantees.
    EJR demands a certain number of candidates, all of whom every $i \in S$ approves, whereas EJR+ demands a single candidate, universally approved in $S$, that must not be in $W$.
    Since apportionment corresponds to the case of sufficiently many copies of each candidate, any single universally-approved candidate can be replicated into several candidates, or into one that is not yet selected, which is why cohesiveness simplifies as in our definition.
    
    Since FJR already implies EJR in committee elections, it suffices to prove that EJR implies FJR. 
    Let the committee $W$ satisfy EJR.
Consider any \(S\subseteq N\) and \(T: P \to \mathbb{N}\) such that $|T|\leq q(S) \leq |S|\cdot k/n$ and $u_i(T) \geq \beta$ for all $i \in S$.
By averaging over $T$, some party $j \in T$ is approved by at least
$\beta \cdot |S|/|T| \geq \beta\cdot n/k$
voters in \(S\). By applying the EJR guarantee of \(W\) to this subset of $S$, one of its members must satisfy \(u_i(W)\geq\beta\), which establishes FJR. The equivalence of FPJR, PJR, and PJR+ follows from an analogous argument given in \cref{app:appor:omit}.
\end{proof}

To build intuition for the apportionment setting, we show on the following example how the argument above fails to prove that, say, EJR+ implies FJR in general committee elections.
\begin{example}
    Let $n=6$ and $k=3$, let the set of candidates/parties be $\{a, b, c, x, y\}$, and the approval sets be $A_1=A_2=A_3 =\{a, b\}$ and $A_4=A_5=A_6 =\{a, c\}$.
\end{example}
Consider the (terrible) committee $W=\{a, x, y\}$.
In both apportionment and committee elections, this violates FJR, as witnessed by the grand coalition $S=N$ with the proposal $T = \{a, b, c\}$, in which every $i$ approves $\beta=2$ alternatives, while they only approve one in $W$.
By applying the averaging argument, we get that $a$ is approved by at least $2 \, n / k = 4$ (in fact, six) voters.
In apportionment, this shows that $W$ violates EJR+: since such popular parties are available, a utility of 1 for all voters is not enough.
In committee elections, by contrast, the popular candidate $a$ is already in $W$ and cannot be added again.
Hence, the violation of FJR does not lead to a violation of EJR+ there.

We will show in \cref{sec:abc:approx} that the FJR violation of an EJR+ committee in committee elections cannot get worse than the example above:
If the $u_i(T)$ were strictly more than twice $\max_{i \in S} u_i(W)$, then a variant of this proof goes through.

\medskip\noindent\textbf{Bounded PAV improvement$\Leftrightarrow$Lindahl priceability.}
Our next result shows a surprisingly tight relation between the PAV objective and Lindahl priceability.
Whereas PAV started out as the only voting rule known to satisfy strong proportionality axioms like EJR~\cite{JR}, it now has serious competitors like the method of equal shares~\cite{limitsofwelf}.
The fact that Lindahl priceability is not only satisfied by PAV, but moreover equivalent to a notion of approximate PAV optimality, 
reaffirms that the PAV objective is fundamental to proportional representation.

\begin{theorem} \label{thm:PAV-LP}
In the apportionment setting, $W$ is Lindahl priceable if and only if it satisfies bounded PAV improvement.
\end{theorem}
\begin{proof}
Suppose first that \(W\) is Lindahl priceable and let
\((p_{ij})_{i,j}\) be a corresponding price system.
For every voter \(i\) and party \(j \in A_i\), a multiset $T$ with \(u_i(W)+1\) copies of \(j\) would give \(i\) strictly
greater utility than \(W\). It must therefore be unaffordable, so
$
(u_i(W)+1)\cdot p_{ij}>1
$
and
$
p_{ij}>1/({u_i(W)+1}).
$
Consequently, for every party \(j\),
\[
\Delta_W^+(j)
=
\sum_{i:j\in A_i}\frac{1}{u_i(W)+1}
<
\sum_{i:j\in A_i}p_{ij}
\leq
\sum_{i\in N}p_{ij}
\leq
\frac{n}{k}.
\]
(When nobody approves $j$, the first inequality is not strict, but $\Delta_W^+(j)=0$.) \medskip

\noindent Now, suppose that $W$ is not Lindahl priceable, from which we will deduce that $\Delta_W^+(j) \geq  n/k$ for some $j$.
Despite our assumption, let us attempt to construct prices that are almost Lindahl prices.
For $\epsilon > 0$ small enough, set
\[ p_{ij} \coloneqq \begin{cases} 
\tfrac{1}{u_i(W) + 1} + \epsilon & \text{if $j \in A_i$} \\
0 & \text{otherwise}
\end{cases}\]
for all $i$ and $j$.
We chose these prices so that a voter $i$ cannot afford a multiset $T$ they prefer over $W$ within a budget of $1$.
Indeed, that would require at least $u_i(W) + 1$ approved alternatives, for a total price of at least
\[(u_i(W) + 1) \cdot (\tfrac{1}{u_i(W) + 1} + \epsilon) = 1 + (u_i(W) + 1) \, \epsilon > 1.\]
Since, by assumption, $W$ is not Lindahl priceable, the price system above must fail the first condition, that is, for some $j$,
\[   n/k < \sum_{i} p_{ij} = \sum_{i : j \in A_i} (\tfrac{1}{u_i(W) + 1} + \epsilon) \leq \Delta_W^+(j) + n \, \epsilon. \]
Since, for some $j$, this holds for arbitrarily small $\epsilon>0$, it follows that $\Delta_W^+(j) \geq  n/k$.
\end{proof}
We will explain in \cref{sec:abc:approx} how to adapt this proof into one for an approximate relationship between bounded PAV improvement and a variant of Lindahl priceability.

\medskip\noindent\textbf{PAV Is Priceable.}
Our third result shows that (locally) PAV optimal committees are priceable.
This is counter-intuitive because priceability was defined as part of a case against ``welfarist'' rules~\cite{limitsofwelf}, which, like PAV, maximize a function of voter utilities.
Although that paper shows that no Pareto-optimal welfarist rule is priceable, these desiderata turn out to be compatible in apportionment.

\begin{theorem}
    In the apportionment setting, local PAV optimality implies priceability.
    \label{thm:pavpriceable}
\end{theorem}
\begin{proof}
    Let $W$ be locally PAV optimal.
    It would be natural to show priceability by setting payments according to the Lindahl prices in the previous proof: $p'_{ij} \coloneqq W(j)/(u_i(W) + 1)$ if $j \in A_i$ and $p_{ij}'=0$ otherwise.
    For these payments, each voter spends $u_i(W)/(u_i(W) + 1) \leq 1$, each party receives $W(j) \cdot \sum_{i : j \in A_i} 1 / (u_i(W) + 1) = W(j) \cdot \Delta_W^+(j)$, and for each $j$, the left-over money of its supporters is
    $\sum_{i : j \in A_i} (1 - u_i(W)/(u_i(W) + 1)) = \Delta_W^+(j)$.
    Intuitively, the payments $p'$ satisfy a variant of priceability with different per-unit prices $\Delta_W^+(j)$ for each $j$, whereas we need a common price $p$.
    
    We fix this defect by making voter $i$ spend some additional amount $\delta_{ij} \geq 0$ for each $j \in A_i$.
    We set $p \coloneqq \max_{j} \Delta_W^+(j)$, so that the constraint on left-over money is satisfied even without additional spending.
    In determining the $\delta_{ij}$, we ensure that (i) no voter's additional spending exceeds their left-over budget $1/(u_i(W) + 1)$ and that (ii) each alternative $j$ receives extra payments equal to $W(j) \cdot (p -\Delta_W^+(j))$.
    This way, $p_{ij} \coloneqq p_{ij}' + \delta_{ij}$ will certify priceability.

    Conditions (i) and (ii) can be written as constraints on a flow from voters to alternatives, with an edge between voters and their approved alternatives.
    Hence, our desired $\delta_{ij}$ exist iff a Hall condition is satisfied, which we derive in \cref{app:appor:omit}.
    Specifically, we must show that, for any set of parties $P' \subseteq P$,
    \begin{equation}\label{eq:hall}
    \sum_{j \in P'} W(j) \cdot (p - \Delta_W^+(j)) \leq \sum_{i : A_i \cap P' \neq \emptyset} \tfrac{1}{u_i(W) + 1}.
    \end{equation}

    Next, we make use of local optimality.
    If $T$ is a multiset of parties and $T(j) > 0$ for some $j$, set
    \[ \textstyle \Delta_T^-(j) \coloneqq \PAV(T) - \PAV(T - \{j\}) = \sum_{i : j \in A_i} 1/u_i(T). \]
    By local optimality, for parties $j$ and $j'$ such that $W(j)>0$,
    \[ 0 \geq \Delta_W^+(j') - \Delta_{W + \{j'\}}^-(j) \geq \Delta_W^+(j') - \Delta_{W}^-(j), \]
    where the inequality follows as $\Delta_T^-(j)$ is monotone nonincreasing in $T$.
    Hence, it holds that $\Delta_W^-(j) \geq \max_{j'} \Delta_W^+(j') = p$ for all such $j$, and that
        \begin{equation*}
        p\!-\!\Delta^+_W(j) \leq \Delta^-_W(j)\!-\!\Delta^+_W(j)= \sum_{i:j\in A_i} \tfrac{1}{u_i(W)}\!-\!\tfrac{1}{u_i(W)+1}.
    \end{equation*}
    
    By summing this inequality over $W$, we get
    \begin{align*}
       & \sum\nolimits_{j\in P': W(j) >0} W(j)\cdot (p- \Delta^+_W(j))\\
       \leq{} &\sum_{\mathclap{i: (A_i \cap P')\ne \emptyset, u_i(W) >0}}  \left(\tfrac{1}{u_i(W)} - \tfrac{1}{u_i(W)+1}\right)\cdot \left( \textstyle \sum_{j \in A_i \cap P'} W(j)\right)\\
        \leq{} &\sum_{\mathclap{i: (A_i \cap P')\ne \emptyset, u_i(W) >0}}  \left(\tfrac{1}{u_i(W)} - \tfrac{1}{u_i(W)+1}\right) \cdot u_i(W)\\
       \leq{} &\sum\nolimits_{i: A_i \cap P' \ne \emptyset} \frac{1}{u_i(W)+1}.
    \end{align*}
    This shows the Hall condition~\eqref{eq:hall}. Hence, $W$ is priceable.
\end{proof}
Note that a locally PAV optimal committee satisfies both Lindahl priceability and priceability, but it is not known whether it can be computed in polynomial time~\cite{KE24}. By contrast, a committee satisfying bounded PAV improvement (and so Lindahl priceability) can be found in polynomial time using local search~\cite{AEH+18}.

\section{Portioning}
\label{sec:portioning}
\subsection{Preliminaries}
In portioning, an instance is a triple $(N, P, A)$, just
as in apportionment but without a committee size.
A \emph{portioning} is a vector $r \in \mathbb{R}^P_{\geq 0}$ with $\norm{r}_1 = 1$, that is, it allocates seat shares $r(j)$ to each party $j$, adding up to one.
Denote $i$'s utility for a vector $t \in \mathbb{R}_{\geq 0}^P$ by $u_i(t) \coloneqq \sum_{j \in A_i} t(j)$.
A \emph{portioning method} maps each instance to a portioning. A prominent portioning method is \emph{Nash portioning}~\cite{BMS05}, which selects the portioning $r$ maximizing the Nash welfare $\operatorname{NW}(r) \coloneqq \prod_{i\in N} u_i(r)$.

\subsection{Connecting Apportionment and Portioning}
\label{sec:portioning:lift}
The portioning setting has often served as a source of inspiration in committee elections (e.g., in defining Lindahl priceability, \citealt{munagala2022auditing}) and apportionment (e.g., rounding majoritarian portioning, \citealt{paulapproval}), but this connection has typically stayed informal.
The one exception we know is a note by \textcite{Peters25}.
Fixing $N, P, A$, he shows that, when $W_1, W_2, \dots$ is an infinite sequence of PAV-optimal committees for committee sizes $k_1 < k_2 < \dots$, then every convergent subsequence of the normalized apportionments $W_\ell / k_\ell$ converges to a Nash-optimal portioning.

Building on this idea, we formally develop the connection between apportionment and portioning.
We do so by defining a \emph{lifting} operation that takes any logical predicate $X=X(N, P, A, k, W)$ over an apportionment instance and a committee and maps it to a corresponding predicate $\lift{X}=\lift{X}(N, P, A, r)$ over portioning instances and portionings.
Typically, these predicates will describe committee-level axioms ($X = \text{``$W$ satisfies EJR+ in $(N, P, A, k)$''}$), but we can also lift voting rules ($X = \text{``$W$ is the outcome of PAV in the instance''}$).
\begin{definition}\label{def:lifting}
Let \(X\) be a predicate over an apportionment instance and committee.
Its \emph{portioning analog} \(\lift{X}\)
is a predicate over a portioning instance \((N,P,A)\) and a portioning $r$, which is satisfied iff there exists a sequence of apportionment committees $(W_\ell)_{\ell\in \mathbb{N}}$ satisfying $X$ on the instance $(N, P, A, |W_\ell|)$, such that $\lim_{\ell\to \infty}|W_\ell| = \infty$ and $\lim_{\ell\to \infty}W_\ell/|W_\ell| = r$.
\end{definition}
\begin{table*}[tbh]
{\begin{center}
    \begin{tabular}{ll}
        \toprule
        Committee-Election / Apportionment Predicate & Lifted Portioning Analog \\
        \midrule
        Lindahl Priceability\textsuperscript{~\cite{munagala2022auditing}}  & Lindahl Equilibrium\textsuperscript{~\cite{Foley70}} \\
        Priceability\textsuperscript{~\cite{limitsofwelf}}  & Decomposability\textsuperscript{~\cite{BBG+22}} / Group Fair Share\textsuperscript{~\cite{BMS05}} \\
        Core Stability\textsuperscript{~\cite{JR}} & Weak Core\textsuperscript{~\cite{Aumann61}} \\
        Proportional Approval Voting\textsuperscript{~\cite{Thiele95}} & Nash Portioning\textsuperscript{~\cite{BMS05}} \\
        Method of Equal Shares\textsuperscript{\textdagger~\cite{limitsofwelf}} & Majoritarian Portioning\textsuperscript{\textdagger~\cite{Majoritarian-portioning}} \\
        \bottomrule
    \end{tabular}
\end{center}}\caption{Apportionment axioms and methods and their lifted analogs. (\textdagger: breaking ties according to fixed order over parties.)}
\label{tbl:analogs}
\end{table*}
To support the claim that this definition is ``the right one'', \cref{tbl:analogs} gives five examples of axioms and voting rules in committee elections that lift to key notions in the portioning literature.
These connections are also reassuring signs that the committee-election definitions actually track their continuous motivations.
For example, Lindahl priceability lifts to its inspiration \emph{Lindahl equilibrium}, despite missing one of its conditions, and we strengthen the one-sided containment of \textcite{Peters25} to show that the Nash optimal portionings are exactly the limits of PAV optimal apportionments.
Priceability lifts to \emph{decomposability} as defined by \textcite{BBG+22}, which \textcite{BBP+21} show to be equivalent to \emph{group fair share}~\cite{BMS05}. 
Other notions from apportionment translate into natural analogs in portioning that can be stated without reference to limits.
(See \cref{app:por:lift} for details and proofs.)

\subsection{Implications Between Properties}
We now consider the implications between properties in portioning.
We ask if the graph of implications in apportionment (all arrows, \cref{fig:abctoapp}) is more similar to that in committee elections (black arrows, \cref{fig:abctoapp}), or to the implications between the portioning analogs of the axioms (all arrows, \cref{fig:portioning}).

Another desirable property of our lifting is that it mechanically preserves implications and existence results from apportionment. (See \cref{app:por:thms} for a simple proof.)
As a result, all implications from apportionment induce valid arrows (in black) in \cref{fig:portioning}, and all axioms remain satisfiable.
\begin{proposition}\label{rem:carry}
Suppose that, on all apportionment instances, predicate $X$ implies predicate $Y$.
Then, $\lift{X}$ implies $\lift{Y}$ in all portioning instances.
Similarly, if for each apportionment instance, there exists a committee satisfying $X$, it must be true that, for each portioning instance, there exists a portioning satisfying $\lift{X}$.
\end{proposition}

As \cref{fig:portioning} shows, the portioning setting has a few additional implications relative to apportionment (proofs in \cref{app:por:thms}).
The more material one is that the analogs of local PAV optimality and PAV optimality coincide with those of bounded PAV improvement, which is already equivalent to Lindahl priceability in apportionment.
This collapse follows from the equivalence between Nash optimality and Lindahl equilibrium~\cite{FGM16} but can also be seen step-by-step: as $k \to \infty$, the approximate local optimality of bounded PAV improvement requires local optimality, and since the PAV objective converges to the concave Nash welfare, local optimality implies global optimality.
The other new implication is that the analog of core stability (weak core) implies the analog of priceability (decomposability).
This only incrementally strengthens the implication chain
\[ \lift{\text{Lindahl pr.}} \Leftrightarrow \lift{\text{local PAV opt.}} \Rightarrow \lift{\text{priceability}}.\]

More importantly, our \cref{thm:PAV-LP,thm:pavpriceable} relating to the PAV objective are natural from the perspective of portioning.
Indeed, our equivalence between bounded PAV improvement and Lindahl priceability recalls the equivalence of Lindahl equilibrium with maximum Nash welfare, and the prices we set in \cref{thm:PAV-LP} are discrete analogs of those proving the portioning equivalence via the KKT conditions ($i$'s value for $j$ divided by $i$'s current utility).
Similarly, the priceability of PAV is less surprising knowing that the Nash portioning is decomposable, which again follows from its first-order optimality conditions.
For these reasons, our impression is that the integrality of apportionment is less of a hurdle to axiomatic implications than the inability to choose candidates more than once in committee elections.

\section{Committee Elections}
\label{sec:abc}
\subsection{Preliminaries}
\label{sec:abc:prelims}
An instance in committee elections is defined as in apportionment, except that we write $C$ for the set of \emph{candidates} (replacing the parties).
Since candidates can be chosen at most once, $k \leq |C|$ and a committee $W$ is simply a subset of $C$ of size $k$.
We set $u_i(W) \coloneqq |A_i \cap W|$.

Fix a committee $W$. We restate key axioms, highlighting differences in color. We defer other definitions to \cref{app:abc:definitions}, and provide a more detailed discussion of bounded PAV improvement and frugal Lindahl priceability in \cref{app:com:new}. For FJR, we introduce a multiplicative relaxation by some $\alpha \geq 1$:
\definecolor{accentblue}{HTML}{7040AD}
\newcommand{\chg}[1]{\textcolor{accentblue}{#1}}
\begin{description}
    \item[EJR+.] $W$ satisfies EJR+ if, for every $\emptyset \neq S \subseteq N$ with $(\bigcap_{i \in S} A_i) \chg{\setminus W} \!\neq\! \emptyset$, there exists $i \in S$ with $u_i(W) \geq q(S)$.
    \item[\chg{$\alpha$-}FJR.] $W$ satisfies \chg{$\alpha$-}FJR if, for every $\emptyset \neq S \subseteq N$, every $T \chg{\subseteq C}$ of size $|T| \leq q(S)$ and every $\beta \in \mathbb{N}$ such that $u_i(T) \geq \chg{\alpha \cdot}\beta$ for all $i \in S$, there is $i \in S$ with $u_i(W) \geq \beta$.
    \item[Bounded PAV Improvement.] $W$ satisfies this axiom if $\Delta_W^+(j) < n/k$ for all $j \in C \chg{\setminus W}$.
\end{description}

\subsection{Transferring Results from Apportionment}
\label{sec:abc:approx}
Recall our motivation of using apportionment as a ``model organism'' for committee elections, to understand relations between concepts that are weaker than implications.
We now study if our three connections in apportionment carry over in terms of approximation results.

Concerning our first result, while EJR does not imply approximate notions of EJR+ and FJR as shown in \cref{app:abc:proofs}, we can adapt our proof to connect EJR+ and FJR:
\begin{theorem}\label{thm:2fjr}
    EJR+ implies 2-FJR, and PJR+ implies 2-FPJR.
\end{theorem}
\begin{proof}
    Retracing the proof of \cref{thm:efjr+}, we can already see that, if there were a ($1$-)FJR violation $S, T, \beta$ such that $T \subseteq C \setminus W$, there must be a group of size at least $\beta \cdot n / k$ who all like one element of $T$, so an EJR+ violation.

    We now assume that $W$ violates $2$-FJR and will deduce an EJR+ violation from that.
    We are given $S, T, \beta$ such that $u_i(T) \geq 2 \, \beta > 2 \, u_i(W)$ for all $i \in S$.
    Since $|A_i \cap W| < \beta$ and $|A_i \cap T| \geq 2 \, \beta$, $|A_i \cap (T \setminus W)| \geq \beta$.
    But then, $u_i(T \setminus W) \geq \beta > u_i(W)$ for all $i \in S$, so $S, T \setminus W, \beta$ witnesses a $1$-FJR violation with a candidate set $T \setminus W$ disjoint from $W$, so we obtain a contradiction with EJR+ as explained above.
    We give the PJR+ analog in \cref{app:abc:proofs}.
\end{proof}
Our argument recalls an observation by \textcite{paulapproval} that the proof for PAV's core stability in apportionment extends to committee elections with ``disjoint objections'', which they show implies a 2-approximation for core stability.\smallskip

In committee elections, Lindahl priceability and bounded PAV improvement are logically incomparable, seemingly disagreeing with the connection we found in apportionment.
We can, however, establish a bidirectional link if we slightly strengthen Lindahl priceability by adding a third condition (we also add an approximation by $\alpha \geq 1$):
\begin{description}
    \item[\chg{$\alpha$-Frugal} Lindahl Priceability.] $W$ is \chg{$\alpha$-frugal} Lindahl priceable if there exist prices $(p_{ij})_{i,j} \geq 0$ such that (i)~$\sum_{i}p_{ij} \leq n/k$ for all $j$; (ii)~whenever $\sum_{j \in T} p_{ij} \leq 1$ for some $i$ and $T \chg{\subseteq C}$, it holds that $u_i(T) \leq \chg{\alpha \cdot} u_i(W)$; and (iii)~\chg{$p_{ij} \leq p_{ij'}$ whenever $j \in A_i \cap W$ and $j' \in A_i \setminus W$}.
\end{description}
In apportionment, ($1$-)frugal Lindahl priceability coincides with Lindahl priceability, so both lift to Lindahl equilibrium.

\looseness=-1
Adding condition (iii) to Lindahl priceability is natural given its inspiration from Lindahl equilibrium since, in continuous settings, a voter maximizing their utility subject to a budget will buy a cheaper good before an expensive good that is equally preferred.
\mbox{($1$-)}frugal Lindahl priceability is implied by \emph{stable priceability}~\cite{PPS+21}, in which voters lexicographically maximize utility and then left-over budget.\footnote{Whereas stable priceability is unsatisfiable for some instances, this is open for frugal Lindahl priceability, see below.}

Frugal Lindahl priceability and bounded PAV improvement are related in both directions:
\begin{theorem}\label{thm:frugal}
    Frugal Lindahl priceability implies bounded PAV improvement, and bounded PAV improvement implies 2-frugal Lindahl priceability.\footnote{This strengthens the result that bounded PAV improvement implies 2-core stability~\cite{limitsofwelf}.}
\end{theorem}
\begin{proof}
Suppose first that \(W\) is frugal Lindahl priceable, with prices $(p_{ij})_{i,j}$.
Fix any candidate $j^* \in C \setminus W$.
We will show that $\Delta_W^+(j^*) < n/k$.
If no voter approves $j^*$, this is direct.
Else, let $i$ be any voter approving $j^*$ and $T \coloneqq (A_i \cap W) \cup \{j^*\}$.
Since $u_i(T) > u_i(W)$, we have that
\[ 1 < \sum_{j \in T} p_{ij} = p_{ij^*} + \sum_{j \in A_i \cap W} p_{ij} \leq (u_i(W) + 1) \, p_{ij^*}, \]
where the last inequality follows from condition (iii) of frugal Lindahl priceability.
Thus, $p_{ij^*} > 1/(u_i(W) + 1)$ and
\begin{align*}
     \Delta_W^+(j^*) = \sum_{i: j^* \in A_i} \frac{1}{u_i(W) + 1} < \sum_{i: j^* \in A_i} p_{ij^*} \leq \frac{n}{k},
\end{align*}
where the last inequality follows from condition (i).

As before, we prove the converse by contrapositive, assuming that $2$-frugal Lindahl priceability is violated and obtaining a $j \in C \setminus W$ with large PAV increase $\Delta_W^+(j)$.
This constraint ``$j \in C \setminus W$'' is what stops the old price system from working.
It would satisfy conditions (ii) and (iii), so the assumed violation of ($2$-)frugal Lindahl priceability would give us some $j$ with $\Delta_W^+(j) \geq n/k$.
But this does not contradict bounded PAV improvement in committee elections if $j \in W$.

Recalling the idea of ``disjoint objections'' leads us to the right pricing system.
Holding voter $i$ fixed, our goal is to ensure that no $T \subseteq C$ affordable within a budget of $1$ can violate condition (ii) by having $u_i(T) > 2 \, u_i(W)$.
The insight is that, for this to be the case, it must also hold that $u_i(T \setminus W) > u_i(W)$.
So, if we set a price slightly above $1/(u_i(W) + 1)$ only for approved alternatives outside of $W$, this is enough to avoid a violation of condition (ii).
Indeed, we set, for small enough $\epsilon > 0$,
\[
p_{ij}
\coloneqq
\begin{cases}
\tfrac{1}{u_i(W)+1} + \epsilon & \text{if $j\in A_i\setminus W$}\\
0 & \text{otherwise.}
\end{cases}
\]
This satisfies conditions (ii) and (iii) of $2$-frugal Lindahl priceability by construction.
Since, by assumption, this axiom does not hold, there must be a violation of condition (i), that is, some $j$ for which $\sum_{i} p_{ij} > n/k$.
Now, such a $j$ cannot be in $W$ because candidates in $W$ only have zero prices.
We conclude as in \cref{thm:PAV-LP} that, for some $j \in C \setminus W$,
\[ n/k < \sum_{i} p_{ij} = \sum_{i : j \in A_i} (\tfrac{1}{u_i(W) + 1} + \epsilon) \leq \Delta_W^+(j) + n \, \epsilon, \]
and hence that $\Delta_W^+(j) \geq n/k$ as claimed.
\end{proof}
\looseness=-1
Since we could not find any counterexamples to the satisfiability of frugal Lindahl priceability, we see it as a plausible target for the field's push to design a voting rule satisfying core stability.
The above result that frugal Lindahl priceability implies bounded PAV improvement may be useful in restricting the space of voting rules to explore;
while maximizing the PAV score can violate core stability, a frugal Lindahl priceable rule cannot leave out alternatives that would sufficiently increase the PAV score.
We also show in \cref{app:com:new} that bounded PAV improvement implies an optimal \emph{proportionality degree}~\cite{Skowron21} of $\ell - 1$, another attractive aspect of a hypothetical voting rule satisfying frugal Lindahl priceability. \smallskip

Our final result, the priceability of PAV, does not seem to carry over to committee elections. 
In \cref{app:abc:proofs}, we adapt a proof by \textcite{limitsofwelf} to show that a Pareto-optimal and welfarist voting rule like PAV cannot yield any constant approximation to priceability (where the approximation allows voters more unspent budget).

\section{Conclusion}
In this paper, we contributed to the understanding of approval-based apportionment by proving new axiomatic implications and formally investigating the setting's relationship with portioning.
Much more remains to be studied: we did not consider axioms like perfect representation and laminar proportionality that only constrain certain profiles, relational axioms like committee monotonicity, or the compatibility of axioms (e.g., committee monotonicity and core stability).
Conceptually, we would like to understand what kinds of axioms are easier to satisfy in apportionment than in general committee elections.

Our investigation also led us to interesting results for committee elections.
For example, we would not have thought to ask whether EJR+ implies approximate FJR had it not been for their equivalence in apportionment.
We are also curious whether the frugal strengthening of Lindahl priceability may prove useful in other works.
In hindsight, the simpler cohesiveness condition of EJR in apportionment foreshadowed the definition of EJR+, which, despite being stronger than EJR, has nicer computational properties and a natural greedy algorithm~\cite{EJR+}.
With luck, the apportionment setting might again point to an axiomatic variation with beneficial mathematical structure.

\printbibliography

\newpage 
\onecolumn
\appendix
\setcounter{section}{0}
\renewcommand{\thesection}{\Alph{section}}
\renewcommand{\thesubsection}{\thesection.\arabic{subsection}}
\setcounter{secnumdepth}{3}
\section*{Appendix}
\crefalias{section}{appendix}
\crefalias{subsection}{subappendix}      \crefalias{subsubsection}{subsubappendix} 

\section{Use of AI Tools}
We used conversations with ChatGPT and Claude to help develop initial versions of the proofs that PAV implies priceability, that majoritarian portioning is the lift of the Method of Equal Shares, and that Nash portioning is the lift of PAV.
These tools were also used for proofreading, including checks of both exposition and mathematical correctness.
All arguments were subsequently reviewed, revised, and verified by the authors.
In addition to our manual efforts, we searched for counterexamples to the satisfiability of frugal Lindahl priceability in committee elections using ChatGPT 5.6 Ultra, but did not find any.

\section{Apportionment}
An apportionment instance is a tuple $(N, P, A, k)$. By default, $i$ ranges over the set of voters $N$ and $j$ over the set of parties $P$.
A \emph{committee} is a multiset $W: P \to \mathbb{N}$ of size $|W| = \sum_{j} W(j) = k$.
We write $u_i(W) \coloneqq \sum_{j\in A_i} W(j)$ for voter $i$'s \emph{utility} for committee $W$. $q(S)\coloneqq \lfloor|S| \cdot k /n\rfloor$ is the (Hare) \emph{quota} of a set of voters $S\subseteq N$.
\subsection{Axiom Definitions}\label{app:appor:def}

For completeness, we first formalize the embedding of apportionment into committee elections introduced by \citet{paulapproval}, in a manner analogous to \cref{def:lifting}.

\newcommand{\ap}[1]{\operatorname{app}(#1)}

\begin{definition}[\citealt{paulapproval}]\label{def:paullift}
    Let $X = X(N, C, A, k, W)$ be a predicate over a committee election instance and a committee. Its \emph{apportionment analog} $X^{\mathrm{ap}} = X^{\mathrm{ap}} (N, P, A, k, W)$ is a predicate over an apportionment instance and committee defined as follows. Let $(N,P^{\mathrm{cl}},A^{\mathrm{cl}},k,W^{\mathrm{cl}})$ be the committee election instance obtained by replacing each party $j\in P$ with $k+1$ clones $j^{(1)},\ldots,j^{(k+1)}.$ Specifically,
    \[
    P^{\mathrm{cl}} = \cup_{j\in P} \{j^{(1)},\ldots,j^{(k+1)}\} \quad , \quad A^{\mathrm{cl}}_i=\cup_{j\in A_i} \{j^{(1)},\ldots,j^{(k+1)}\} \quad , \quad W^{\mathrm{cl}} = \cup_{j \in P} \{j^{(1)},\ldots,j^{(W(j))}\}.
    \]
 Then $X^{\mathrm{ap}}(N,P,A,k,W)$ holds if and only if $X(N,P^{\mathrm{cl}},A^{\mathrm{cl}},k,W^{\mathrm{cl}})$ holds.
\end{definition}

The following is the apportionment analog of \cref{rem:carry}, which is stated implicitly in \cite{paulapproval}.
\begin{proposition}\label{prop:liftimply}
Suppose that, on all committee election instances, predicate $X$ implies predicate $Y$.
Then, $X^{\mathrm{ap}}$ implies $Y^{\mathrm{ap}}$ on all apportionment instances.
Similarly, if for each committee election instance, there exists a committee satisfying $X$, it must be true that, for each apportionment instance, there exists a committee satisfying $X^{\mathrm{ap}}$.
\end{proposition}
\begin{proof}
    Suppose an apportionment committee $W$ satisfies $X^{\mathrm{ap}}$ on $(N, P, A, k)$. By definition, $W^{\mathrm{cl}}$ satisfies $X$, and therefore $Y$ on $(N, P^{\mathrm{cl}}, A^{\mathrm{cl}}, k)$. Again by definition, $W$ satisfies $Y^{\mathrm{ap}}$ on $(N, P, A, k)$.
    For the second part, fixing an apportionment instance $(N, P, A, k)$, there exists a committee $W^{\mathrm{cl}}$ satisfying $X$\footnote{To be precise, axiom $X$ must be invariant under relabeling of identical candidates.} on $(N, P^{\mathrm{cl}}, A^{\mathrm{cl}}, k)$. By definition, $W$ satisfies $X^{\mathrm{ap}}$ on $(N, P, A, k)$.
\end{proof}

We next define the apportionment versions of the axioms by applying the construction in \cref{def:paullift}. Fix an apportionment instance $(N, P, A, k)$ and a committee $W$. Then:
\begin{description}
    \item[JR.] $W$ satisfies \emph{justified representation} if, for every nonempty $S \subseteq N$ with $\bigcap_{i \in S} A_i \neq \emptyset$ and $q(S) \geq 1$, there exists a voter $i \in S$ with $u_i(W) \geq 1$.
    \item[EJR/EJR+.] $W$ satisfies \emph{extended justified representation (plus)} if, for every nonempty $S \subseteq N$ with $\bigcap_{i \in S} A_i \neq \emptyset$, there exists a voter $i \in S$ with $u_i(W) \geq q(S)$.
    \item[FJR.] $W$ satisfies \emph{full justified representation} if, for every nonempty $S \subseteq N$, every multiset $T: P \to \mathbb{N}$ of size $|T| \leq q(S)$, and every $\beta \in \mathbb{N}$ such that $u_i(T) \geq \beta$ for all $i \in S$, there is an $i \in S$ such that $u_i(W) \geq \beta$.
    \item[PJR/PJR+.] $W$ satisfies \emph{proportional justified representation (plus)}
    if, for every nonempty $S \subseteq N$ with $\bigcap_{i \in S} A_i \neq \emptyset$, it holds that $\sum_{j\in \bigcup_{i\in S} A_i} W(j)\geq q(S)$.
     \item[FPJR.] $W$ satisfies \emph{full proportional justified representation}
    if, for every nonempty $S \subseteq N$, every multiset $T: P \to \mathbb{N}$ of size $|T| \leq q(S)$, and every $\beta \in \mathbb{N}$ such that $u_i(T) \geq \beta$ for all $i \in S$, it holds that $\sum_{j\in \bigcup_{i\in S} A_i} W(j)\geq \beta$.
    \item[Core Stability.] $W$ satisfies \emph{core stability}
    if, for every $S\subseteq N$ and every nonempty multiset $T: P\to \mathbb{N}$ of size $|T| \leq q(S)$, there exists a voter $i\in S$ with $u_i(T) \leq u_i(W)$.
    \item[Lindahl Priceability.] $W$ is \emph{Lindahl priceable} if there exist nonnegative personalized prices $(p_{ij})_{i,j}$ such that $\sum_{i} p_{ij} \leq n/k$ for all $j$ and such that, for any voter $i$ and multiset $T: P \to \mathbb{N}$ with $\sum_{j} T(j) \, p_{ij} \leq 1$, it holds that $u_i(T) \leq u_i(W)$. (Conceptually, $i$ cannot afford any outcome they prefer over $W$ within a budget of $1$.)
    \item[Priceability.] $W$ is \emph{priceable} if there exist nonnegative payments $(p_{ij})_{i,j}$ and a per-unit price $p \geq 0$ such that $p_{ij}=0$ whenever $j \notin A_i$, $\sum_{j} p_{ij} \leq 1$ for all $i$, $\sum_{i} p_{ij} = W(j) \cdot p$ for all $j$, and $\sum_{i : j \in A_i} (1 -\sum_{j' \in P} p_{ij'}) \leq p$ for all $j$. (Conceptually, the unspent budget of $j$'s supporters is at most $p$.)
    \item[Bounded PAV Improvement.] $W$ satisfies bounded PAV improvement if $\Delta^+_W(j) < n/k$ for every $j \in P$.
    \item[Local PAV Optimality.] $W$ is locally PAV optimal if, $\PAV(W{-}\{j\}{+}\{j'\}) \leq \PAV(W)$ for all parties $j$ and $j'$ with $W(j) > 0$.
    \item[PAV Optimality.] $W$ is PAV optimal if, $\PAV(W') \leq \PAV(W)$ for every committee $W'$.
\end{description}

As the definitions of priceability and Lindahl priceability have gone through some simplification, we formally establish their correspondence with the committee election definitions. The original committee election definitions are included in \cref{app:abc:definitions}.

\begin{proposition}
    The apportionment definition of Lindahl priceability is precisely the lift, as defined in \cref{def:paullift}, of their committee election counterparts. 
\end{proposition}
\begin{proof}
    Given the apportionment $(N, P, A, k)$, committee $W$ is Lindahl priceable if and only if $W^{\mathrm{cl}}$ is  Lindahl priceable for $(N, P^{\mathrm{cl}}, A^{\mathrm{cl}}, k)$.
    
    Assume given $(N, P^{\mathrm{cl}}, A^{\mathrm{cl}}, k)$, $W^{\mathrm{cl}}$ is Lindahl priceable witnessed by $(p'_{ic})_{i, c\in P^{\mathrm{cl}}}$. Note that, by the utility-maximization condition of Lindahl priceability, we must have $\sum_{\ell}p'_{ij^{(\ell)}} > 1$ for every $i$ and $j\in A_i$ as $u_i(W^{\mathrm{cl}}) \leq k < u_i(\{j^{(1)}, j^{(2)}, \dots, j^{(k+1)}\})$. Set $p_{ij}$ to the average price of all copies of $j$, $p_{ij} = \sum_\ell p'_{ij^{(\ell)}}/(k+1)$.
    We get $\sum_i p_{ij} =  \sum_{\ell}\sum_i p'_{ij^{(\ell)}}/(k+1)\leq {n}/{k}$ for all $j$. 
    
Also, fixing $i$, consider a multiset $T: P \to \mathbb{N}$ with $\sum_{j} T(j)p_{ij} \leq 1$. 
    Note that for $j\in A_i$, as $p_{ij} > {1}/({k+1})$ and $\sum_{j} T(j)p_{ij} \leq 1$, we must have $T(j) <  k+1$. 
    Let $T'\subseteq P^{\mathrm{cl}}$ contain, for every $j\in A_i$, the $T(j)$ copies of $j$ with the lowest prices $p'_{ij^{(\ell)}}$, and no copies of any $j\notin A_i$. This gives $\sum_{j^{(\ell)}\in T'}p'_{ij^{(\ell)}} \leq \sum_{j} T(j)p_{ij}$ and $u_i(T) = u_i(T')$. Then we have $\sum_{j^{(\ell)}\in T'}p'_{ij^{(\ell)}} \leq \sum_{j} T(j)p_{ij} \leq 1$, implying $u_i(T) = u_i(T') \leq u_i(W^{\mathrm{cl}}) = u_i(W)$ where the inequality is by Lindahl priceability of $W^{\mathrm{cl}}$ for $(N, P^{\mathrm{cl}}, A^{\mathrm{cl}}, k)$.
    Therefore, $W$ is Lindahl priceable witnessed by $(p_{ij})_{i, j}$. 
    
    Conversely, assume given $(N, P, A, k)$, $W$ is Lindahl priceable witnessed by $(p_{ij})_{i,j}$. 
    Setting $p'_{ij^{(\ell)}} = p_{ij}$, we get $\sum_i p'_{ij^{(\ell)}} = \sum_i p_{ij}\leq n/k$ for every $j^{(\ell)} \in P^{\mathrm{cl}}$. 
    
    Also, for every $i$ and $T' \subseteq P^{\mathrm{cl}}$ with $\sum_{c\in T'} p'_{ic} \leq 1$, let $T:P\to \mathbb{N}$ be such that $T(j)$ is the number of clones of $j$ in $T'$. 
    Then we have  $\sum_j T(j)p_{ij} = \sum_{j^{(\ell)}\in T'} p'_{ij^{(\ell)}} \leq 1$, and $u_i(T') = u_i(T) \leq u_i(W) = u_i(W^{\mathrm{cl}})$ implying $W^{\mathrm{cl}}$ is Lindahl priceable witnessed by $(p'_{ic})_{i, c\in P^{\mathrm{cl}}}$.
\end{proof}

\begin{proposition}
    The apportionment definition of priceability is precisely the lift, as defined in \cref{def:paullift}, of their committee election counterparts. 
\end{proposition}
\begin{proof}
    Given the apportionment $(N, P, A, k)$, committee $W$ is priceable if and only if $W^{\mathrm{cl}}$ is priceable for $(N, P^{\mathrm{cl}}, A^{\mathrm{cl}}, k)$.
    Assume given $(N, P^{\mathrm{cl}}, A^{\mathrm{cl}}, k)$, $W^{\mathrm{cl}}$ is priceable witnessed by $(p', (p'_{ic})_{i, c\in P^{\mathrm{cl}}})$. Set $p_{ij}$ to the sum of prices of all copies of $j$: $p_{ij} = \sum_\ell p'_{ij^{(\ell)}}$. 
    Then if $j\notin A_i$, $j^{(\ell)} \notin A^{\mathrm{cl}}_i$ for every $\ell$ and $p_{ij} = \sum_\ell p'_{ij^{(\ell)}} = 0$. 
    Also, $\sum_{j} p_{ij} = \sum_{j^{(\ell)} \in P^{\mathrm{cl}}} p'_{ij^{(\ell)}} \leq 1$ for every $i$, and $\sum_{i} p_{ij} = \sum_{i}\sum_{\ell} p'_{ij^{(\ell)}} = W(j)\cdot p'$ for every $j$. 
    Finally, $\sum_{i: j\in A_i} (1-\sum_{j'\in P} p_{ij'}) = \sum_{i: j^{(1)}\in A^{\mathrm{cl}}_i} (1-\sum_{c\in P^{\mathrm{cl}}} p'_{ic}) \leq p'$ for every $j$ implying $W$ is priceable witnessed by $(p', (p_{ij})_{i, j})$.

     Conversely, assume given $(N, P, A, k)$, $W$ is priceable witnessed by $(p, (p_{ij})_{i,j})$. 
    Set $p'_{ij^{(\ell)}}=p_{ij}/W(j)$ if $j^{(\ell)} \in W^{\mathrm{cl}}$ and $p'_{ij^{(\ell)}}=0$ otherwise. 
    This way, if $j^{(\ell)}\notin A^{\mathrm{cl}}_i$, then $j\notin A_i$ and $p'_{ij^{(\ell)}} = 0$. Also, 
    $\sum_{j^{(\ell)}\in P^{\mathrm{cl}}} p'_{ij^{(\ell)}}=\sum_{j} \sum_{\ell=1}^{W(j)} p_{ij}/W(j) = \sum_{j} p_{ij} \leq 1$ for every $i$, and for every $j^{(\ell)} \in P^{\mathrm{cl}}$ it holds that $\sum_i p'_{ij^{(\ell)}} = 0$ if $j^{(\ell)} \notin W^{\mathrm{cl}}$ and $\sum_i p'_{ij^{(\ell)}} = \sum_i p_{ij}/W(j)=p$ if  $j^{(\ell)}\in W^{\mathrm{cl}}$. 
    Finally, $\sum_{i: j^{(\ell)}\in A^{\mathrm{cl}}_i} (1-\sum_{c\in P^{\mathrm{cl}}} p'_{ic}) = \sum_{i: j\in A_i} (1-\sum_{j'\in P} \sum_{l=1}^{W(j')}p'_{ij^{'(\ell)}}) \leq p$ for every $j^{(\ell)} \in P^{\mathrm{cl}}$. Therefore, $(p, p'_{ic})_{i,c\in P^{\mathrm{cl}}}$ witnesses the priceability of $W^{\mathrm{cl}}$ for $(N, P^{\mathrm{cl}}, A^{\mathrm{cl}}, k)$.
\end{proof}

\subsection{Omitted Proofs}\label{app:appor:omit}

\begin{proposition}\label{prop:FPJR+}
   In the apportionment setting, FPJR, PJR, and PJR+ coincide.
\end{proposition}
\begin{proof}
    The equivalence of PJR+ and PJR again follows immediately from the characterization of PJR in \citet{paulapproval}.
    
    Since FPJR already implies PJR in committee elections, it suffices to prove that PJR implies FPJR. Let the committee $W$ satisfy PJR.
Consider any \(S\subseteq N\) and \(T: P \to \mathbb{N}\) such that $|T|\leq q(S) \leq |S|\cdot k/n$ and $u_i(T) \geq \beta$ for all $i \in S$.
By averaging over $T$, some party $j \in T$ is approved by at least
$\beta \cdot |S|/|T| \geq \beta\cdot n/k$
voters in \(S\). By applying the PJR guarantee of \(W\) to this subset $S'$ of $S$, it must hold that $\sum_{j \in \bigcup_{i \in S'} A_i} W(j) \geq \beta$.
Since $\sum_{j \in \bigcup_{i \in S} A_i} W(j)$ is at least$\sum_{j \in \bigcup_{i \in S'} A_i} W(j)$, this implies FPJR.
\end{proof}

We next prove the following lemma, which we used in the proof of \cref{thm:pavpriceable}.
\begin{lemma}\label{lem:supp-dem}
Let \(G=(V\cup U,E)\) be a bipartite graph, with supply upper-bounds \(h_v\) for
\(v\in V\) and demands \(d_u\) for \(u\in U\). Then, there exist nonnegative weights \((x_{vu})_{\{v,u\}\in E}\) such that \(\sum_{u\in N(v)}x_{vu}\leq h_v\) for every \(v\in V\), and \(\sum_{v\in N(u)}x_{vu}=d_u\) for every \(u\in U\), if and only if
\(
\sum_{u\in U'}d_u
\leq
\sum_{v\in N(U')}h_v
\) 
for every $U'\subseteq U$ where \(N(U')\) is the set of neighbors of \(U'\).
\end{lemma}

\begin{proof}
Construct a flow network as follows. 

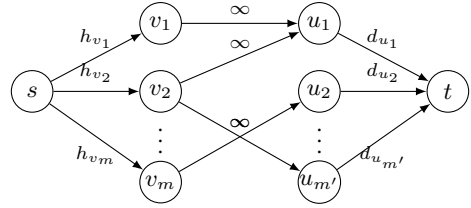
\begin{figure}[h]
\centering
\begin{tikzpicture}[
    >=latex,
    every node/.style={font=\small},
    vertex/.style={circle,draw,minimum size=5.5mm,inner sep=0pt},
    lab/.style={font=\scriptsize}
]

\node[vertex] (s) at (0,0) {$s$};

\node[vertex] (v1) at (1.7,0.9) {$v_1$};
\node[vertex] (v2) at (1.7,0) {$v_2$};
\node at (1.7,-0.55) {$\vdots$};
\node[vertex] (vn) at (1.7,-1.2) {$v_m$};

\node[vertex] (u1) at (3.8,0.9) {$u_1$};
\node[vertex] (u2) at (3.8,0) {$u_2$};
\node at (3.8,-0.55) {$\vdots$};
\node[vertex] (um) at (3.8,-1.2) {$u_{m'}$};

\node[vertex] (t) at (5.5,0) {$t$};

\draw[->] (s) -- node[above,lab] {$h_{v_1}$} (v1);
\draw[->] (s) -- node[above,lab] {$h_{v_2}$} (v2);
\draw[->] (s) -- node[below,lab] {$h_{v_{m}}$} (vn);

\draw[->] (v1) -- node[above,lab] {$\infty$} (u1);
\draw[->] (v2) -- node[above,lab] {$\infty$} (u1);
\draw[->] (v2) -- node[above,lab] {$\infty$} (um);
\draw[->] (vn) -- node[above,lab] {$\infty$} (u2);

\draw[->] (u1) -- node[above,lab] {$d_{u_1}$} (t);
\draw[->] (u2) -- node[above,lab] {$d_{u_2}$} (t);
\draw[->] (um) -- node[below,lab] {$d_{u_{m'}}$} (t);

\end{tikzpicture}
\end{figure}

Add a source $s$ and a sink $t$. For each
$v\in V$, add a directed edge $(s,v)$ of capacity $h_v$. For each edge $\{v,u\}\in E$, add
a directed edge $(v,u)$ of capacity $+\infty$. Finally, for each $u\in U$, add a directed  edge
$(u,t)$ of capacity $d_u$.
Observe that a feasible assignment $(x_{vu})_{\{v,u\}\in E}$ is equivalent to an
$s$--$t$ flow of value $\sum_{u\in U}d_u$: the flow on edge $(v,u)$ is
$x_{vu}$, the capacity constraints on $(s,v)$ enforce
$\sum_{u\in N(v)}x_{vu}\le h_v$, and saturating every edge $(u,t)$ is
equivalent to satisfying
$\sum_{v\in N(u)}x_{vu}=d_u$. On the other hand, the capacity of the cut $(V \cup U \cup \{s\}, \{t\})$ equals $\sum_{u\in U}d_u$, and the value of $s$--$t$ max-flow is at most $\sum_{u\in U}d_u$. Therefore, we need to prove the value of $s$--$t$ max-flow is exactly $\sum_{u\in U}d_u$. By the max-flow min-cut theorem, it suffices to prove that
$s$--$t$ min-cut has capacity exactly $\sum_{u\in U}d_u$. 

Consider any finite-capacity cut $(S,T)$ with $s\in S$ and $t\in T$. 
Let
\(
U'=T\cap U
\).
Then every neighbor of $U'$ lies in $T\cap V$; otherwise an infinite-capacity
edge would cross the cut. Hence
\(
N(U') \subseteq T\cap V
\), and the cut has capacity of at least
\(
\sum_{v\in N(U')}h_v+\sum_{u\in U\setminus U'}d_u
\). Moreover, this lower bound is attained exactly when $N(U') = T\cap V$.
Therefore, the capacity of min-cut is $\sum_{u\in U}d_u$ if and only if for every $U'\subseteq U$,
\[
\sum_{v\in N(U')}h_v+\sum_{u\in U\setminus U'}d_u
\ge
\sum_{u\in U}d_u,
\]
which is equivalent to
\(
\sum_{v\in N(U')}h_v
\ge
\sum_{u\in U'}d_u.
\)
\end{proof}

We next confirm that there are no other implication relations in \cref{fig:abctoapp}. For the reader's convenience, the figure is shown again below.

\begin{adjustbox}{max width=\linewidth,center}
            \begin{tikzpicture}[
     x=0.9cm,
    y=1.1cm,
    box/.style={
        draw,
        rounded corners,
        very thin,
        node font=\footnotesize,
        minimum height=14pt,
        inner xsep=2pt,
        inner ysep=1pt,
        align=center
    }
]
  \node[box] (node1) at (-3.57,3.96) {Lindahl Pr.\vphantom{Ty}};
  \node[box] (node2) at (-3.57,3.18) {Core Stability\vphantom{Ty}};
  \node[box] (node5) at (-3.56,2.22) {FJR\vphantom{Ty}};
  \node[box] (node12) at (-1.94,2.22) {EJR\vphantom{Ty}};
  \node[box] (node6) at (-0.18,2.22) {EJR+\vphantom{Ty}};

  \draw[arrows=-Stealth, thick] (node1.south) -- (node2.north);
  \draw[arrows=-Stealth, thick] (-3.19,2.18) -- (-2.32,2.15);
  \draw[arrows=-Stealth, thick] (-0.67,2.15) -- (-1.56,2.15);

  \node[box] (node7) at (-3.57,1.46) {FPJR\vphantom{Ty}};
  \node[box] (node11) at (-1.94,1.46) {PJR\vphantom{Ty}};
  \node[box] (node8) at (-0.18,1.45) {PJR+\vphantom{Ty}};
  \node[box] (node10) at (-1.91,0.38) {Priceability\vphantom{Ty}};

  \draw[arrows=-Stealth, thick] (node2.south) -- (node5.north);
  \draw[arrows=-Stealth, thick] (node5.south) -- (node7.north);
  \draw[decorate, thick, arrows=-Stealth, bend left=17]
      (node2.south) to[bend left=13] (node8.north);
  \draw[arrows=-Stealth, thick] (node12.south) -- (node11.north);
  \draw[arrows=-Stealth, thick] (node10.north) -- (node8.south);
  \draw[arrows=-Stealth, thick] (node10.north) -- (node7.south);
  \draw[arrows=-Stealth, thick] (node6.south) -- (node8.north);

  \node[box] (node3) at (-0.18,3.18) {Bounded PAV Imp.\vphantom{Ty}};
  \node[box] (node4) at (-0.18,3.96) {Local PAV Optimality\vphantom{Ty}};

  \draw[arrows=-Stealth, thick, line cap=projecting]
      (node4.south) -- (node3.north);
  \draw[thick, line cap=projecting, arrows=-Stealth, purple]
      (-1.56,2.28) -- (-0.68,2.27);
  \draw[arrows=-Stealth, thick]
      (-0.63,1.37) -- (-1.57,1.37);
  \draw[thick, line cap=projecting, arrows=-Stealth, purple]
      (-1.57,1.5) -- (-0.63,1.49);

  \draw[arrows=-Stealth, thick, line cap=projecting]
      (node3.south) -- (node6.north);

  \draw[thick, line cap=projecting, purple, arrows=Stealth-]
      (-3.19,2.3) -- (-2.34,2.28);

  \draw[arrows=-Stealth, thick]
      (-3.12,1.39) -- (-2.3,1.39);
  \draw[thick, line cap=projecting, purple, arrows=Stealth-]
      (-3.12,1.53) -- (-2.3,1.53);

  \draw[decorate, thick, arrows=Stealth-Stealth, purple, bend right=15]
      (-0.5,2.45) to[bend right=11] (-3.38,2.45);

  \draw[decorate, thick, arrows=Stealth-Stealth, purple, bend left=15]
      (-0.44,1.23) to[bend left=9] (-3.33,1.23);

  \draw[thick, purple, arrows=-Stealth]
      (-2.96,3.74) -- (-1.62,3.24);
  \draw[thick, purple, arrows=Stealth-]
      (-3.33,3.74) -- (-1.62,3.09);
  \draw[arrows=-Stealth, thick, purple, bend left=50]
      (node4.east) to[bend left=30] (node10.east);

  \node[box] (node9) at (-0.18,4.74) {PAV Optimality\vphantom{Ty}};

  \draw[arrows=-Stealth, thick, line cap=projecting]
      (node9.south) -- (node4.north);

\end{tikzpicture}
        \end{adjustbox}

\begin{proposition} [\citealt{paulapproval}]\label{app:appor:prop:PEJR}
    EJR does not imply core stability, and PJR does not imply EJR.
\end{proposition}
\begin{proof}
    Based on Appendix B of \citeauthor{paulapproval}, in the apportionment setting, seq-Phragmén satisfies PJR but not EJR, and the method of equal shares satisfies EJR but not core stability.
\end{proof} 

\begin{example}[\citealt{edge}]\label{ex:app:lp}
    Let $n=3$, $k=2$, and \(P=\{a, b, c, d\}\). The voters have approval sets 
\(
A_1=\{c, d\},
A_2=\{a, c, d\},
A_3=\{a, b, c, d\}
\), and the elected committee is \(W=\{a, b\}\).
\end{example}
\begin{proposition}
    Core stability does not imply Lindahl priceability.
\end{proposition}
\begin{proof}
    It is easy to verify that in \cref{ex:app:lp}, $W$ is core stable. \citet{edge} showed $W$ is not Lindahl priceable in the case of committee elections, and Lindahl priceability in apportionment is only more constrained.
\end{proof}

\begin{example}\label{ex:app:bpi}
    Let $n=10$, $k=5$, and $P{=}\{a, b, c\}$. The approval sets are $A_i = \{a\}$ for $i\in\{1,2 , 3\}$, $A_i = \{b\}$ for $i\in\{4, 5, 6\}$, $A_i = \{c\}$ for $i\in\{7, 8, 9, 10\}$. The elected committee is $W=\{a, b, c, c, c\}$.
\end{example}
\begin{example}\label{ex:app:pav}
    Let $n=4$, $k=2$, and $P{=}\{a, b, c, d\}$.  approval sets are $A_1 = \{a, c\}$, $A_2 = \{a, d\}$, $A_3 = \{b, c, d\}$, $A_4=\{b\}$, and $W=\{c, d\}$.
\end{example}

\begin{proposition}\label{app:appor:pav}
    Bounded PAV improvement does not imply local PAV optimality, and local PAV optimality does not imply PAV optimality.
\end{proposition}
\begin{proof}
    In \cref{ex:app:bpi} $W$ satisfies bounded PAV improvement as $\Delta^+_W(j) < {n}/{k} = 2$ for every $j$, but it is not locally optimal: $\PAV{}(W{+}\{a\} {-} \{c\}) > \PAV{}(W)$. In \cref{ex:app:pav}, $W$ is locally PAV optimal but not optimal as $\PAV{}(\{a, b\}) >\PAV{}(W)$.
\end{proof}

\begin{proposition}\label{app:appor:lindhprice}
    Lindahl priceability does not imply priceability.
\end{proposition}
\begin{proof}
   In \cref{ex:app:bpi} $W$ satisfies bounded PAV improvement, and by \cref{thm:PAV-LP} it is Lindahl priceable, but is not priceable: the first three voters have a total budget of $3$ and purchase one seat. 
   Their remaining budget is therefore $3-p$, which must be at most $p$. 
   Hence, \(p\geq \frac{3}{2}\). 
   On the other hand, the last four voters have a total budget of $4$ and purchase three seats, implying $3p\leq 4$, and \(p\leq \frac{4}{3}\).
\end{proof}
\begin{proposition}
    Priceability does not imply EJR.
\end{proposition}
\begin{proof}
    Seq-Phragmén satisfies priceability for committee elections~\cite{limitsofwelf} and thus for apportionment by \cref{prop:liftimply}, yet it can violate EJR in the apportionment setting~\cite{paulapproval}.
\end{proof}
\section{Portioning}
A portioning instance is a tuple $(N, P, A)$. By default, $i$ ranges over the set of voters $N$ and $j$ over the set of parties $P$.
A \emph{portioning} is a vector $r \in \mathbb{R}^P_{\geq 0}$ with $\norm{r}_1 = 1$.
$i$'s utility for a vector $t \in \mathbb{R}_{\geq 0}^P$ is denoted by $u_i(t) \coloneqq \sum_{j \in A_i} t(j)$.
\subsection{Preliminaries}
Below comes the definition of the portioning concepts we mentioned in the body.
\begin{description}
    \item[Lindahl Equilibrium.] Let $r\in \mathbb{R}^P_{\geq 0}$ be a portioning and $(p_{ij})_{i, j}$ be nonnegative prices. Then $(r, (p_{ij})_{i, j})$ is a \emph{Lindahl equilibrium} if  $\sum_{j} r(j)\,p_{ij} \leq 1$ for every $i$, $\sum_i p_{ij} \leq n$ for every $j$ and $\sum_i p_{ij} = n$ if $r(j) > 0$, and for every voter $i$ and $t\in \mathbb{R}^P_{\geq 0}$ such that $\sum_{j} t(j) \, p_{ij} \leq 1$,  $u_i(t) \leq u_i(r)$.
    \item [Decomposability.] A portioning $r$ is \emph{decomposable} if we can write $r=\sum_{i} \delta_i$ where $\delta_{i}: \mathbb{R}^P_{\geq 0}$ is such that $\delta_{i}(j)=0$ whenever $j \notin A_i$, and $\sum_{j} \delta_{i}(j) =1/n$ for all $i$.\footnote{The original definition by \citet{BBP+21} considers a more general setting which allows agents (voters) to have arbitrary contributions $c_i \geq 0$. We assume all voters contribute equally, with  $1/n$.}
    \item [Group Fair Share.] A portioning $r$ satisfies \emph{group fair share} if for every $S\subseteq N$, $\sum_{j \in \bigcup_{i \in S} A_i} r(j) \geq |S|/n$.
    \item [Weak Core.] A portioning $r$ is in \emph{weak core} if there does not exist a nonempty coalition $S\subseteq N$ and $t\in \mathbb{R}_{\geq 0}^P$ such that $\lVert t \rVert_1 \leq |S|/n$ and $u_i(t) > u_i(r)$ for every $i\in S$.
    \item[Nash Portioning.] The portioning method that selects a portioning $r$ maximizing the Nash welfare $\text{NW}(r)=\prod_{i} u_i(r)$.
    \item [Majoritarian portioning.] It proceeds in rounds $\ell=1,2, \dots$. Initially, all parties and voters are active. In iteration $\ell$, we select the active party $j_\ell$ that is approved by the highest number of active voters. Let $N_\ell$ be the set of active voters who approve $j_\ell$. Then, set $r(j_\ell)$ to $|N_\ell|/n$, and mark $j_\ell$ and all voters in $N_\ell$ as inactive. If active voters remain, the next iteration is started; else, $r$ is returned.
\end{description}
\subsection{Lifting Apportionment Predicates}\label{app:por:lift}
In this part, we derive the definitions for the lift of apportionment predicates mentioned in \cref{app:appor:def}. We will need Dirichlet's theorem on simultaneous Diophantine approximation, which is given below.

\begin{theorem}\label{thm:drich}
    Let \(\alpha\in\mathbb{R}^d\) be a real vector. Then there exist a sequence of integer vectors \((z^\ell)_{\ell \in \mathbb{N}}\) where $z^\ell \in \mathbb{Z}^d$ and a sequence of positive integers $(k_\ell)_{\ell\in \mathbb{N}}$ such that
\[
\lim_{\ell\to\infty} \lVert k_\ell\alpha-z^\ell\rVert_\infty = \lim_{\ell\to\infty} \max_{1\leq j\leq d} |k_\ell\alpha_j - z^\ell_j|=0
\]
\end{theorem}
\begin{corollary}\label{cor:lift}
    For every portioning \(r\in\mathbb{R}_{\geq 0}^P\), there exist integer-valued committees \(W_\ell: P \to \mathbb{N}\) for every $\ell\in \mathbb{N}$ such that $|W_\ell| \to \infty$, and $\max_{j}||W_\ell| r(j)-W_\ell(j)| \to 0$. In particular, we also get $W_\ell/|W_\ell| \to r$.
\end{corollary}
\begin{proof}
    For $|P|=1$, then necessarily $r=1$, and the desired sequence  $(W_\ell)_{\ell\in \mathbb{N}}$ trivially exists. So assume $|P| > 1$ and fix a $j^*\in P$ satisfying $r(j^*) > 0$. By \cref{thm:drich}, there exist $W_\ell: P\setminus\{j^*\} \to \mathbb{Z}$ and positive integers $k_{\ell}$ such that $\max_{j\ne j^*}|k_\ell r(j)-W_\ell(j)| \to 0$. Let $W_\ell(j^*) = k_\ell - \sum_{j \ne j^*} W_\ell(j)$. We obviously get $|W_\ell| = k_\ell$, and also

    {
    \footnotesize
    \begin{align*}
    ||W_\ell| r(j^*)-W_\ell(j^*)|
    =
    \left||W_\ell|\left(1-\sum_{j\neq j^*}r(j)\right)
    -\left(|W_\ell|-\sum_{j\neq j^*}W_\ell(j)\right)\right|
    =
    \left|\sum_{j\neq j^*}\bigl(W_\ell(j)-|W_\ell| r(j)\bigr)\right|
    \leq \sum_{j\neq j^*}
\left|W_\ell(j)-|W_\ell| r(j)\right|
\end{align*}
}
which converges to 0 as \(P\) is finite and
\(
\left|W_\ell(j)-|W_\ell| r(j)\right| \to 0
\) for every $j\ne j^*$.
    
    We now show that we may choose the sequence so that $|W_\ell| \to \infty$. If all coordinates $r(j)$ for $j\ne j^*$ are rational, let $D$ be a common denominator and $k_\ell = \ell D$. Then $W_\ell(j) = k_\ell r(j)$ for $j\ne j^*$ will be integral and satisfies $\max_{j\ne j^*}|k_\ell r(j)-W_\ell(j)|=0$. 
    Otherwise, if one of the coordinates $r(j)$ for $j\ne j^*$ is irrational, assume for the sake of contradiction that $|W_\ell|\not\to \infty$. 
    Then there exists some integer $k$ such that $|W_\ell|=k_\ell = k$ for infinitely many $\ell$. 
    Therefore, on this subsequence $\max_{j\ne j^*}|k r(j)-W_\ell(j)| \to 0$. As $W_\ell(j)$ is integral, $k r(j)$ must be integral for every $j\ne j^*$, which implies all coordinates $j\ne j^*$ are rational and leads to a contradiction. Therefore we showed $|W_\ell| \to \infty$.

    Also, as $|W_\ell| r(j) \geq 0$, $W_\ell(j)$ is integral, and $||W_\ell| r(j) - W_\ell(j)| \to 0$, we must have $W_\ell(j) \geq 0$ for sufficiently large $\ell$. Hence, after discarding finitely many initial terms, we may assume that $W_\ell \geq 0$ for every $\ell$.

    Finally, we have $\lim_{\ell\to \infty} \lVert r - W_\ell/|W_\ell|\rVert = \lim_{\ell\to \infty}\frac{1}{|W_\ell|}\lVert |W_\ell| r - W_\ell\rVert = 0$, implying $W_\ell/|W_\ell| \to r$.
\end{proof}

Observe that, for committees \(W_\ell\) obtained by applying \cref{cor:lift}, we have
\(
\varepsilon_\ell:=\max_{j}| |W_\ell| r(j)-W_\ell(j)|\to0.
\)
Therefore,  for every voter \(i\),
\begin{equation}\label{eq:lift1}
    |u_i(W_\ell)-|W_\ell| u_i(r)|
\leq |A_i|\varepsilon_\ell
< 1,
\end{equation}
and, for every \(S\subseteq N\),
\begin{equation}\label{eq:lift2}
    \left|\sum_{j\in \bigcup_{i\in S}A_i}W_\ell(j)-|W_\ell|\sum_{j\in \bigcup_{i\in S}A_i}r(j)\right|
\leq |\bigcup_{i\in S}A_i|\varepsilon_\ell
< 1
\end{equation}
for all sufficiently large $\ell$.

\begin{proposition}
    Given an instance $(N, P, A)$, every portioning \(r\) satisfies the lift of JR.
\end{proposition}
\begin{proof}
    Consider an arbitrary portioning $r$. For every $\ell\in\mathbb{N}$, define a committee $W_\ell$ by $W_\ell(j)=\lfloor r(j)\ell+1\rfloor$. 
        Note that $|W_\ell|=\sum_{j}\lfloor r(j)\ell+1\rfloor\geq\sum_{j}r(j)\ell=\ell$. Therefore $\lim_{\ell\to\infty}|W_\ell|=+\infty$. Since $W_\ell(j)\geq1$ for every $j$, it satisfies JR on the instance $(N,P,A,|W_\ell|)$. Moreover, $|W_\ell|\leq\sum_{j}(r(j)\ell+1)=\ell+|P|$, and hence, for a fixed $j$, we have $$\frac{\lfloor r(j)\ell+1\rfloor}{\ell+|P|}\leq\frac{W_\ell(j)}{|W_\ell|}\leq\frac{\lfloor r(j)\ell+1\rfloor}{\ell}.$$
        As $\ell\to\infty$, both the left- and right-hand sides converge to $r(j)$. Hence $\lim_{\ell\to\infty}W_\ell(j)/|W_\ell|=r(j)$ for every $j$, and as $P$ is finite, we have $\lim_{\ell\to\infty}W_\ell/|W_\ell|=r$.
\end{proof}

\begin{proposition}
    Given an instance $(N, P, A)$, portioning \(r\) satisfies the lift of \emph{EJR/EJR+} if, for every nonempty $S \subseteq N$ with $\bigcap_{i \in S} A_i \neq \emptyset$, there exists a voter $i \in S$ with $u_i(r) \geq |S|/n$.
\end{proposition}
\begin{proof}
    Assume we have a sequence of committees $(W_\ell)_{\ell\in\mathbb{N}}$ such that $W_\ell$ satisfies EJR on $(N,P,A,|W_\ell|)$, $|W_\ell|\to\infty$, and $W_\ell/|W_\ell|\to r$. 
        Consider an arbitrary nonempty set $S\subseteq N$ with $\bigcap_{i\in S}A_i\neq\emptyset$. 
        As every $W_\ell$ is EJR and $S$ is finite, there exists an $i^*\in S$ satisfying $u_{i^*}(W_\ell)\geq\lfloor |S||W_\ell|/n\rfloor$ for infinitely many $\ell\in\mathbb{N}$. 
        Passing to this subsequence, we have $$u_{i^*}(r)=\lim_{\ell\to\infty}u_{i^*}(W_\ell/|W_\ell|)\geq\lim_{\ell\to\infty}\lfloor |S||W_\ell|/n\rfloor/|W_\ell|=|S|/n.$$

        Conversely, assume that $r$ satisfies the proposed definition. Apply \cref{cor:lift} to find committees $W_\ell: P \to \mathbb{N}$. By \cref{eq:lift1} for all sufficiently large $\ell$, $|u_i(W_\ell)-|W_\ell| u_i(r)|<1$ for every $i$.
        Consider an arbitrary nonempty set $S\subseteq N$ with $\bigcap_{i\in S}A_i\neq\emptyset$, and let $i_S\in S$ be such that $u_{i_S}(r)\geq |S|/n$. Then $u_{i_S}(W_\ell)>|W_\ell| u_{i_S}(r)-1\geq |W_\ell||S|/n-1$ for all sufficiently large $\ell$. Since $u_{i_S}(W_\ell)$ is integral, this implies $u_{i_S}(W_\ell)\geq\lfloor |W_\ell||S|/n\rfloor$. As there are only finitely many sets $S\subseteq N$, $W_\ell$ satisfies EJR for all sufficiently large $\ell$.
\end{proof}

\begin{proposition}
    Given an instance $(N, P, A)$, portioning \(r\) satisfies the lift of \emph{FJR} if, for every nonempty $S \subseteq N$, every $t \in \mathbb{R}_{\geq 0}^P$ with $\lVert t\rVert \leq |S|/n$, and every $\beta \geq 0$ such that $u_i(t) \geq \beta$ for all $i \in S$, there is an $i \in S$ such that $u_i(r) \geq \beta$.
\end{proposition}
\begin{proof}
        Assume we have a sequence of committees $(W_\ell)_{\ell\in\mathbb{N}}$ such that $W_\ell$ satisfies FJR on $(N,P,A,|W_\ell|)$, $|W_\ell|\to\infty$, and $W_\ell/|W_\ell|\to r$. 
        Consider an arbitrary nonempty set $S\subseteq N$ and $t\in\mathbb{R}_{\geq0}^P$ such that $\lVert t\rVert_1\leq |S|/n$ and $u_i(t)\geq\beta$ for every $i\in S$. 
        For every $\ell\in\mathbb{N}$, let $T_\ell$ be the integer-valued function defined by $T_\ell(j)=\lfloor t(j)|W_\ell|\rfloor$. Note that $|T_\ell|\leq |W_\ell|\lVert t\rVert_1\leq |S||W_\ell|/n$, and as $T_\ell$ is integer-valued, $|T_\ell|\leq\lfloor |S||W_\ell|/n\rfloor$. Moreover, for every $i\in S$, we have $$u_i(T_\ell)\geq \sum_{j\in A_i} (t(j)|W_\ell| - 1) =   |W_\ell|u_i(t)-|A_i|\geq\beta|W_\ell|-\max_{h\in S}|A_h|.$$

        Since $W_\ell$ satisfies FJR, there is an $i_\ell\in S$ such that $u_{i_\ell}(W_\ell)\geq\beta|W_\ell|-\max_{h\in S}|A_h|$. As $S$ is finite, there exists an $i^*\in S$ such that $i_\ell=i^*$ for infinitely many $\ell$. Passing to this subsequence, we have 
        $$u_{i^*}(r)=\lim_{\ell\to\infty}u_{i^*}(W_\ell/|W_\ell|)\geq\lim_{\ell\to\infty}(\beta|W_\ell|-\max_{h\in S}|A_h|)/|W_\ell|=\beta.$$ As $S$, $t$, and $\beta$ were chosen arbitrarily, $r$ satisfies the proposed definition for the lift of FJR.

        Conversely, assume that $r$ satisfies the proposed definition for FJR. Apply \cref{cor:lift} to find committees $W_\ell: P \to \mathbb{N}$. By \cref{eq:lift1} for all sufficiently large $\ell$, $|u_i(W_\ell)-|W_\ell| u_i(r)|<1$ for every $i$.
        Fix such an $\ell$, a nonempty set $S\subseteq N$, and a function $T:P\to\mathbb{N}$ such that $|T|\leq\lfloor |W_\ell||S|/n\rfloor$. Let $\beta=\min_{i\in S}u_i(T)$ and $t=T/|W_\ell|$. Then $\lVert t\rVert_1\leq |S|/n$ and $u_i(t)\geq\beta/|W_\ell|$ for every $i\in S$. Since $r$ satisfies the proposed definition, there exists an $i\in S$ such that $u_i(r)\geq\beta/|W_\ell|$. Therefore $u_i(W_\ell)>|W_\ell| u_i(r)-1\geq\beta-1$. Since $u_i(W_\ell)$ and $\beta$ are integral, $u_i(W_\ell)\geq\beta$. Thus $W_\ell$ satisfies FJR for all sufficiently large $\ell$, completing the proof.
\end{proof}

\begin{proposition}
     Given an instance $(N, P, A)$, portioning \(r\) satisfies the lift of \emph{PJR/PJR+} if, for every nonempty $S \subseteq N$ with $\bigcap_{i \in S} A_i \neq \emptyset$, it holds that $\sum_{j \in \bigcup_{i\in S} A_i} r(j) \geq |S|/n$.
\end{proposition}
\begin{proof}
    Assume we have a sequence of committees $(W_\ell)_{\ell\in\mathbb{N}}$ such that $W_\ell$ satisfies PJR on $(N,P,A,|W_\ell|)$, $|W_\ell|\to\infty$, and $W_\ell/|W_\ell|\to r$. Consider an arbitrary nonempty set $S\subseteq N$ with $\cap_{i\in S}A_i\neq\emptyset$. As every $W_\ell$ is PJR, it holds that $\sum_{j\in\cup_{i\in S}A_i}W_\ell(j)\geq\lfloor |S||W_\ell|/n\rfloor$ for every $\ell\in\mathbb{N}$. Therefore $\sum_{j\in\cup_{i\in S}A_i}r(j)=\lim_{\ell\to\infty}\sum_{j\in\cup_{i\in S}A_i}W_\ell(j)/|W_\ell|\geq |S|/n$.

        Conversely, assume that $r$ satisfies the proposed definition. 
        Apply \cref{cor:lift} to find committees $W_\ell: P \to \mathbb{N}$. 
        By \cref{eq:lift2} for all sufficiently large $\ell$, $ |\sum_{j\in \cup_{i\in S}A_i}W_\ell(j)-|W_\ell|\sum_{j\in \cup_{i\in S}A_i}r(j)|<1$ for every $S \subseteq N$.

        Consider a nonempty set $S\subseteq N$ with $\cap_{i\in S}A_i\neq\emptyset$. Then $\sum_{j\in \cup_{i\in S}A_i}W_\ell(j)>|W_\ell|\sum_{j\in \cup_{i\in S}A_i}r(j)-1\geq |W_\ell||S|/n-1$. Since the left-hand side is integral, $\sum_{j\in \cup_{i\in S}A_i}W_\ell(j)\geq\lfloor |W_\ell||S|/n\rfloor$. Thus $W_\ell$ satisfies PJR for all sufficiently large $\ell$.
\end{proof}

\begin{proposition}
    Given an instance $(N, P, A)$, portioning \(r\) satisfies the lift of \emph{FPJR} if, for every nonempty $S \subseteq N$, every $t \in \mathbb{R}_{\geq 0}^P$ with $\lVert t\rVert_1 \leq |S|/n$, and every $\beta \geq 0$ such that $u_i(t) \geq \beta$ for all $i \in S$, it holds that $\sum_{j \in \bigcup_{i\in S} A_i} r(j) \geq \beta$.
\end{proposition}
\begin{proof}
    
        Assume we have a sequence of committees $(W_\ell)_{\ell\in\mathbb{N}}$ such that $W_\ell$ satisfies FPJR on $(N,P,A,|W_\ell|)$, $|W_\ell|\to\infty$, and $W_\ell/|W_\ell|\to r$. Consider an arbitrary nonempty set $S\subseteq N$ and $t\in\mathbb{R}_{\geq0}^P$ such that $\lVert t\rVert_1\leq |S|/n$ and $u_i(t)\geq\beta$ for every $i\in S$. For every $\ell\in\mathbb{N}$, let $T_\ell(j)=\lfloor t(j)|W_\ell|\rfloor$. As in the case of FJR, $|T_\ell|\leq\lfloor |S||W_\ell|/n\rfloor$ and $u_i(T_\ell)\geq\beta|W_\ell|-\max_{h\in S}|A_h|$ for every $i\in S$. Since $W_\ell$ satisfies FPJR, $\sum_{j\in\bigcup_{i\in S}A_i}W_\ell(j)\geq\beta|W_\ell|-\max_{h\in S}|A_h|$, and
        $$\sum_{j\in\bigcup_{i\in S}A_i}r(j) = \lim_{\ell\to \infty} \sum_{j\in\bigcup_{i\in S}A_i}W_\ell(j)/|W_\ell| \geq\lim_{\ell\to \infty}(\beta|W_\ell|-\max_{h\in S}|A_h|)/|W_\ell| = \beta$$
        Thus $r$ satisfies the proposed definition for the lift of FPJR.

        Conversely, assume that $r$ satisfies the proposed definition. 
        Apply \cref{cor:lift} to find committees $W_\ell: P \to \mathbb{N}$. By \cref{eq:lift2} for all sufficiently large $\ell$, $ |\sum_{j\in \bigcup_{i\in S}A_i}W_\ell(j)-|W_\ell|\sum_{j\in \bigcup_{i\in S}A_i}r(j)|<1$ for every $S \subseteq N$.

        Fix such an $\ell$, a nonempty set $S\subseteq N$, and an integer-valued function $T:P\to\mathbb{N}$ such that $|T|\leq\lfloor |W_\ell||S|/n\rfloor$. Let $\beta=\min_{i\in S}u_i(T)$ and $t=T/|W_\ell|$. Then $\lVert t\rVert_1\leq |S|/n$ and $u_i(t)\geq\beta/|W_\ell|$ for every $i\in S$. Since $r$ satisfies the proposed definition, $\sum_{j\in\bigcup_{i\in S}A_i}r(j)\geq\beta/|W_\ell|$. Therefore $\sum_{j\in\bigcup_{i\in S}A_i}W_\ell(j)>|W_\ell|\sum_{j\in\bigcup_{i\in S}A_i}r(j)-1\geq\beta-1$. Since the left-hand side and $\beta$ are integral, $\sum_{j\in\bigcup_{i\in S}A_i}W_\ell(j)\geq\beta$. Thus $W_\ell$ satisfies FPJR for all sufficiently large $\ell$, completing the proof.
\end{proof}

\begin{proposition}
    The lift of \emph{core stability} is weak core. That is, given an instance $(N, P, A)$, portioning \(r\) satisfies the lift of core stability if, for every nonempty $S \subseteq N$ and every $t \in \mathbb{R}_{\geq 0}^P$ such that $\lVert t\rVert_1 \leq |S|/n$, there exists a voter $i \in S$ with $u_i(t) \leq u_i(r)$.
\end{proposition}
\begin{proof}
    Assume we have a sequence of committees $(W_\ell)_{\ell\in\mathbb{N}}$ such that $W_\ell$ satisfies core stability on $(N,P,A,|W_\ell|)$, $|W_\ell|\to\infty$, and $W_\ell/|W_\ell|\to r$. Consider an arbitrary nonempty set $S\subseteq N$ and $t\in\mathbb{R}_{\geq0}^P$ such that $\lVert t\rVert_1\leq |S|/n$. For every $\ell\in\mathbb{N}$, let $T_\ell(j)=\lfloor t(j)|W_\ell|\rfloor$. Then $|T_\ell|\leq\lfloor |S||W_\ell|/n\rfloor$. Since $W_\ell$ satisfies core stability, there is an $i_\ell\in S$ such that $u_{i_\ell}(W_\ell)\geq u_{i_\ell}(T_\ell)$. As $S$ is finite, there exists an $i^*\in S$ such that $i_\ell=i^*$ for infinitely many $\ell$. Passing to this subsequence, we obtain $u_{i^*}(r)=\lim_{\ell\to\infty}u_{i^*}(W_\ell/|W_\ell|)\geq\lim_{\ell\to\infty}u_{i^*}(T_\ell)/|W_\ell|=u_{i^*}(t)$. Thus $r$ satisfies the proposed definition for the lift of core stability.

        Conversely, assume that $r$ satisfies the proposed definition for core stability. Apply \cref{cor:lift} to find committees $W_\ell: P \to \mathbb{N}$. By \cref{eq:lift1} for all sufficiently large $\ell$, $|u_i(W_\ell)-|W_\ell| u_i(r)|<1$ for every $i$.

        Fix such an $\ell$, a nonempty set $S\subseteq N$, and an integer-valued function $T:P\to\mathbb{N}$ such that $|T|\leq\lfloor |W_\ell||S|/n\rfloor$. Let $t=T/|W_\ell|$. Then $\lVert t\rVert_1\leq |S|/n$, so there exists an $i\in S$ such that $u_i(t)\leq u_i(r)$. Therefore $u_i(T)\leq |W_\ell| u_i(r)<u_i(W_\ell)+1$. Since $u_i(T)$ and $u_i(W_\ell)$ are integral, $u_i(T)\leq u_i(W_\ell)$. Thus $W_\ell$ satisfies core stability for all sufficiently large $\ell$, completing the proof.
\end{proof}

\begin{proposition}
    Given an instance $(N, P, A)$, portioning \(r\) satisfies the lift of \emph{priceability} if, there exist nonnegative payments $(p_{ij})_{i,j}$ and a per-unit price $p\geq 0$, such that $p_{ij}=0$ whenever $j \notin A_i$, $\sum_{j} p_{ij} =1$ for all $i$, and $\sum_{i} p_{ij} = n\cdot r(j)$ for every $j$. Therefore, the lift of \emph{priceability} is decomposability.
\end{proposition}
\begin{proof}
    Assume that we have a sequence of allocations $(W_\ell)_{\ell\in\mathbb{N}}$ such that $W_\ell$ is priceable on $(N,P,A,|W_\ell|)$, $|W_\ell|\to\infty$, and $W_\ell/|W_\ell|\to r$. For every $\ell\in\mathbb{N}$, let $(p^\ell_{ij})_{i,j}$ and $q_\ell>0$ be a priceability certificate for $W_\ell$. Thus, $\sum_jp^\ell_{ij}\leq1$ for every $i$, $p^\ell_{ij}=0$ whenever $j\notin A_i$, $\sum_ip^\ell_{ij}=q_\ell W_\ell(j)$ for every $j$, and $\sum_{i:j\in A_i}(1-\sum_{j'\in P}p^\ell_{ij'})\leq q_\ell$ for every $j$.

        Summing $\sum_ip^\ell_{ij}=q_\ell W_\ell(j)$ over all $j$, we obtain $q_\ell|W_\ell|=\sum_i\sum_jp^\ell_{ij}\leq n$. Hence $q_\ell\leq n/|W_\ell|$, and therefore $q_\ell\to0$. For every $i$, let $b_i^\ell=1-\sum_jp^\ell_{ij}$ denote voter $i$'s unspent budget. Fix some $j_i\in A_i$. By the last condition in the definition of priceability, $0\leq b_i^\ell\leq\sum_{h:j_i\in A_h}b_h^\ell\leq q_\ell$. It follows that $b_i^\ell\to0$ for every $i$. Consequently, $q_\ell|W_\ell|=\sum_i(1-b_i^\ell)\to n$.

        Since $N$ and $P$ are finite and $p^\ell\in[0,1]^{N\times P}$, $(p^\ell)_{\ell\in\mathbb{N}}$ has a converging subsequence. Thus, we may assume w.l.o.g. that $p^\ell_{ij}\to p_{ij}$. The limit payments are nonnegative and $p_{ij}=0$ whenever $j\notin A_i$. Moreover, $\sum_jp_{ij}=\lim_{\ell\to\infty}\sum_jp^\ell_{ij}=\lim_{\ell\to\infty}(1-b_i^\ell)=1$ for every $i$. Finally, for every $j$, $$\sum_ip_{ij}=\lim_{\ell\to\infty}\sum_ip^\ell_{ij}=\lim_{\ell\to\infty}q_\ell W_\ell(j)=\lim_{\ell\to\infty}(q_\ell|W_\ell|)W_\ell(j)/|W_\ell|=nr(j).$$
        Thus $r$ satisfies the proposed definition for the lift of priceability, witnessed by $(p_{ij})_{i,j}$.

         Conversely, suppose that \(r\) satisfies the proposed definition for
    priceability. Hence, there exist nonnegative payments
    \((p_{ij})_{i,j}\) such that \(p_{ij}=0\) whenever
    \(j\notin A_i\), \(\sum_{j}p_{ij}=1\) for every \(i\), and
    \(\sum_{i}p_{ij}=n\cdot r(j)\) for every \(j\).
    
    Consider the following polytope \(\mathcal{P}\), over the variables 
    \((x_j)_{j}\) and \((y_{ij})_{i,j}\):
    
    {
        \begin{align*}
            &\sum_{j}x_j =1\\
            &\sum_{j}y_{ij} =1
            &\forall i\\
            &\sum_{i}y_{ij} =nx_j
            &\forall j\\
            &y_{ij}=0
            &\forall i,\ j\notin A_i\\
            &x_j,\ y_{ij}\geq0
            &\forall i,\ j
        \end{align*}
    }
    
    All coefficients in the constraints of \(\mathcal{P}\) are rational, implying every vertex of the polytope is rational.
    As \((r,(p_{ij})_{i,j})\in\mathcal{P}\), we can therefore find a
    sequence of rational points
    \(
    \left(r_\ell,(p^\ell_{ij})_{i,j}\right)\in\mathcal{P}
    \)
    converging to \((r,(p_{ij})_{i,j})\). For every \(\ell\in\mathbb{N}\), choose \(I_\ell\in\mathbb{N}\) such
    that \(I_\ell\cdot  r_\ell\) is integral. Setting $W_{\ell} = \ell\cdot I_\ell \cdot r_\ell$ we have $\lim_{\ell\to \infty} |W_\ell|=+\infty$ and $\lim_{\ell\to \infty} {W_\ell}/{|W_\ell|}=r$. 
    
    Next, we show that \(W_\ell\) is priceable witnessed by payments $(p^\ell_{ij})_{i, j}$ and the per-unit price $q_\ell = {n}/{|W_\ell|}$. Each voter spends her entire budget, since
    \(
    \sum_{j}p^\ell_{ij}=1
    \) for every $i$, \(p_{ij}=0\) whenever
    \(j\notin A_i\). Moreover, for every \(j\), it holds that
    \(
    \sum_{i}p^\ell_{ij}
    =
    n\cdot r_\ell(j)
    =
    {n}/{|W_\ell|}W_\ell(j)
    =
    q_\ell W_\ell(j).
    \)
    Finally, since every voter spends her entire budget,
    \(
    \sum_{i:j\in A_i}
    (1-\sum_{j'\in P}p^\ell_{ij'})
    =
    0
    \leq q_\ell
    \) for every $j$.
    Therefore, \(W_\ell\) is priceable for every \(\ell\).

    To see this notion is equivalent to decomposability, simply set $\delta_i(j) = p_{ij}/n$.
\end{proof}

\begin{proposition}
    Given an instance $(N, P, A)$, portioning \(r\) satisfies the lift of \emph{Lindahl priceability} if, there exist nonnegative personalized prices $(p_{ij})_{i,j}$ such that $\sum_{i \in N} p_{ij} \leq n$ for all $j$ and such that, for any voter $i$ and $t\in \mathbb{R}^P_{\geq 0}$ such that $\sum_{j} t(j) \, p_{ij} \leq 1$, $u_i(t) \leq u_i(r)$.
\end{proposition}
\begin{proof}
    First, assume that we have a sequence of committees
$(W_\ell)_{\ell\in\mathbb{N}}$ such that $W_\ell$ is Lindahl priceable on
$(N,P,A,|W_\ell|)$, $|W_\ell|\to\infty$, and
$W_\ell/|W_\ell|\to r$. For every $\ell\in\mathbb{N}$, let
$(q^\ell_{ij})_{i,j}$ be personalized prices witnessing that $W_\ell$ is
Lindahl priceable. Thus, $\sum_{i}q^\ell_{ij}\leq n/|W_\ell|$ for
every $j$, and, for every voter $i$ and every function
$T:P\to\mathbb{N}$ such that $\sum_jT(j)q^\ell_{ij}\leq1$, it holds that
$u_i(T)\leq u_i(W_\ell)$. Let $p^\ell_{ij}=|W_\ell|\cdot q^\ell_{ij}$. For every $j$, we have
\(
    \sum_{i}p^\ell_{ij}
    =
    |W_\ell|\sum_{i}q^\ell_{ij}
    \leq n.
\)
In particular, $p^\ell_{ij}\in[0,n]$ for every $i$ and $j$. Since $N$
and $P$ are finite, $(p^\ell)_{\ell\in\mathbb{N}}$ has a converging
subsequence. Thus, we may assume w.l.o.g.\ that
$p^\ell_{ij}\to p_{ij}$ for every $i$ and $j$.

For every $i$ and $j\in A_i$, consider the allocation consisting of
$u_i(W_\ell)+1$ copies of $j$. Since this allocation gives voter $i$
strictly larger utility than $W_\ell$, it cannot be affordable. Therefore,
\(
    \bigl(u_i(W_\ell)+1\bigr)q^\ell_{ij}>1,
\)
and hence
 $q^\ell_{ij}>\frac{1}{u_i(W_\ell)+1}.$ Fixing a voter $i$ and a party $j\in A_i$, we get
\[
    p^\ell_{ij}=|W_\ell|\cdot q^\ell_{ij}
    >
    \frac{|W_\ell|}{u_i(W_\ell)+1}
    =
    \frac{1}{
        u_i(W_\ell)/|W_\ell|+1/|W_\ell|
    }.
\]
Since $p^\ell_{ij}\leq n$, this also implies
$u_i(W_\ell)/|W_\ell|+1/|W_\ell| > 1/n$.
Taking limits gives $u_i(r) = \lim_{\ell\to\infty} u_i(W_\ell)/|W_\ell|+1/|W_\ell| \geq 1/n > 0$. We may therefore take limits in
the previous inequality and obtain $p_{ij}\geq{1}/{u_i(r)}$ for every $i$ and $j\in A_i$.
Moreover, for every $j$, it holds that
\(
    \sum_{i}p_{ij}
    =
    \lim_{\ell\to\infty}
    \sum_{i}p^\ell_{ij}
    \leq n.
\)
Now consider an arbitrary voter $i$ and
$t\in\mathbb{R}_{\geq0}^P$ such that
$\sum_jt(j)p_{ij}\leq1$. As $p_{ij} \geq 1/u_i(r)$ for $j\in A_i$, we have
\[
    1
    \geq
    \sum_jt(j)p_{ij}
    \geq
    \sum_{j\in A_i}t(j)p_{ij}
    \geq
    \frac{1}{u_i(r)}
    \sum_{j\in A_i}t(j)
    =
    \frac{u_i(t)}{u_i(r)}.
\]
Therefore, $u_i(t)\leq u_i(r)$. Thus, $(p_{ij})_{i,j}$ witnesses that
$r$ satisfies the proposed definition for the lift of Lindahl
priceability.

Conversely, suppose that $r$ satisfies the proposed definition for
Lindahl priceability, witnessed by nonnegative personalized prices
$(p_{ij})_{i,j}$. First, note that $u_i(r)>0$ for every voter $i$.
Indeed, if $u_i(r)=0$, then, for any $j\in A_i$, voter $i$ could afford
a sufficiently small positive amount of $j$ and obtain strictly positive
utility, which is a contradiction.

Moreover, for every $i$ and $j\in A_i$, we have $p_{ij}\geq{1}/{u_i(r)}$. Otherwise, suppose that $p_{ij}<1/u_i(r)$ for some $j\in A_i$. If
$p_{ij}>0$, consider the portioning $t$ where $t(j)=1/p_{ij}$ and
$t(j')=0$ for every $j'\neq j$. Then
$\sum_{j'}t(j')p_{ij'}=1$, while
$u_i(t)=1/p_{ij}>u_i(r)$, which is a contradiction. If $p_{ij}=0$,
voter $i$ can obtain an arbitrary positive amount of $j$ at zero cost,
which also gives a contradiction.

Apply \cref{cor:lift} to find committees $W_\ell: P \to \mathbb{N}$. By \cref{eq:lift1} for all sufficiently large $\ell$, $|u_i(W_\ell)-|W_\ell| u_i(r)|<1$ for every $i$, implying $u_i(W_\ell)+1>|W_\ell| u_i(r)$.
For every sufficiently large $\ell$, define
$q^\ell_{ij}=p_{ij}/|W_\ell|$, which implies $\sum_{i}q^\ell_{ij}
    =
    \frac{1}{|W_\ell|}
    \sum_{i}p_{ij}
    \leq
    \frac{n}{|W_\ell|}$ for every $j$, and for every $j\in A_i$,
\begin{equation}\label{eq:lindahl-discrete-price}
    q^\ell_{ij}
    =
    \frac{p_{ij}}{|W_\ell|}
    \geq
    \frac{1}{|W_\ell| u_i(r)}
    >
    \frac{1}{u_i(W_\ell)+1}.
\end{equation}

Now consider an arbitrary voter $i$ and a  function
$T:P\to\mathbb{N}$ such that $u_i(T)>u_i(W_\ell)$. Since both utilities
are integral, $u_i(T)\geq u_i(W_\ell)+1$. Using
\cref{eq:lindahl-discrete-price}, we obtain
\(
    \sum_jT(j)q^\ell_{ij}
    \geq
    \sum_{j\in A_i}T(j)q^\ell_{ij}
    >
    \frac{u_i(T)}{u_i(W_\ell)+1}
    \geq1.
\)
Implying that
$(q^\ell_{ij})_{i,j}$ witnesses that $W_\ell$ is Lindahl priceable.
\end{proof}

\begin{proposition}\label{prop:lindahleq}
    Lindahl equilibrium is equivalent to the lift of Lindahl priceability.
\end{proposition}
\begin{proof}
If $(r,(p_{ij})_{i,j})$ is a Lindahl equilibrium, then $r$ is Lindahl priceable, as the definition of Lindahl equilibrium only imposes additional conditions.

Conversely, assume that a portioning $r$ satisfies the lift of Lindahl priceability, witnessed by $(p_{ij})_{i,j}$. We first show that $\sum_j r(j)p_{ij}\geq 1$ for every voter $i$. Suppose otherwise that $\sum_j r(j)p_{ij}<1$. Fix some $j'\in A_i$. Then we can choose $\varepsilon>0$ sufficiently small such that $\sum_j r(j)p_{ij}+\varepsilon p_{ij'}\leq1$. Let $t=r+\varepsilon e_{j'}$. Then $t$ is affordable for voter $i$, while $u_i(t)=u_i(r)+\varepsilon>u_i(r)$, contradicting that $r$ maximizes voter $i$'s utility among all affordable portionings.
On the other hand, since $\sum_i p_{ij}\leq n$ for every $j$, $r\geq0$, and $\lVert r\rVert_1=1$, we have $n\geq\sum_j r(j)\left(\sum_i p_{ij}\right)=\sum_i\sum_j r(j)p_{ij}\geq n$ where the last inequality follows from $\sum_j r(j)p_{ij}\geq 1$. Therefore, all inequalities hold with equality. In particular, $\sum_j r(j)p_{ij}=1$ for every voter $i$, and $\sum_i p_{ij}=n$ whenever $r(j)>0$. Hence $(r,(p_{ij})_{i,j})$ is a Lindahl equilibrium.
\end{proof}

\begin{proposition}\label{por:bpav}
    Given an instance $(N, P, A)$, portioning \(r\) satisfies the lift of \emph{bounded PAV improvement} if, $\frac{\partial \log\text{NW}(r)}{\partial r(j)} {=} \sum_{i: j \in A_i} \frac{1}{u_i(r)} \leq n$ for every $j$.
\end{proposition}
\begin{proof}
    Assume that we have a sequence $(W_\ell)_{\ell\in\mathbb{N}}$ such that $\Delta^+_{W_\ell}(j)<n/|W_\ell|$ for every $j$, $|W_\ell|\to\infty$, and $W_\ell/|W_\ell|\to r$. 
        Therefore, for a fixed $j$, $|W_\ell|\Delta^+_{W_\ell}(j)=\sum_{i:j\in A_i}|W_\ell|/(u_i(W_\ell)+1)<n$. If $u_i(r)=0$ for some $i$ with $j\in A_i$, then $u_i(W_\ell)/|W_\ell|\to0$ and $$\lim_{\ell\to\infty}|W_\ell|/(u_i(W_\ell)+1)=\lim_{\ell\to\infty}1/(u_i(W_\ell)/|W_\ell|+1/|W_\ell|)=+\infty,$$ which contradicts $|W_\ell|\Delta^+_{W_\ell}(j)=\sum_{i:j\in A_i}|W_\ell|/(u_i(W_\ell)+1)<n$. Hence $u_i(r)>0$ for every $i$ that $j\in A_i$. 
        Taking limits gives 
        $$n \geq \lim_{\ell\to\infty}\sum_{i:j\in A_i}1/(u_i(W_\ell)/|W_\ell|+1/|W_\ell|)= \sum_{i:j\in A_i}1/u_i(r).$$
        Thus $r$ satisfies the proposed definition of bounded PAV improvement.

        Conversely, suppose that $r$ satisfies $\sum_{i:j\in A_i}1/u_i(r)\leq n$ for every $j$. Apply \cref{cor:lift} to find committees $W_\ell: P \to \mathbb{N}$. By \cref{eq:lift1} for all sufficiently large $\ell$, $|u_i(W_\ell)-|W_\ell| u_i(r)|<1$ for every $i$. Hence, for all sufficiently large $\ell$, $u_i(W_\ell)+1>|W_\ell| u_i(r)$ for every $i$. Therefore, for every $j$, $|W_\ell|\Delta^+_{W_\ell}(j)=\sum_{i:j\in A_i}|W_\ell|/(u_i(W_\ell)+1)<\sum_{i:j\in A_i}1/u_i(r)\leq n$. Thus $\Delta^+_{W_\ell}(j)<n/|W_\ell|$ for every $j$, and $W_\ell$ satisfies bounded PAV improvement for all sufficiently large $\ell$.
\end{proof}

\begin{proposition}\label{prop:lift}
The lift of proportional approval voting (PAV) is Nash portioning. That is, given an instance $(N, P, A)$, portioning \(r\) satisfies the lift of \emph{PAV optimality} if, $\text{NW}(t) \leq \text{NW}(r)$ for every portioning $t$.
\end{proposition}
\begin{proof}
        Assume we have a sequence of PAV-optimal committees $(W_\ell)_{\ell\in\mathbb{N}}$ such that $|W_\ell|\to\infty$ and $W_\ell/|W_\ell|\to r$. 
        First note that, as $W_\ell$ is PAV-optimal, it also satisfies bounded PAV improvement, and we showed in the proof of \cref{por:bpav} that $u_i(r) > 0$, meaning $\log u_i(r)$ is well-defined. 
        As $W_\ell/|W_\ell| \to r$,  $\log u_i(W_\ell/|W_\ell|) = \log u_i(r) + \varepsilon'_{\ell i}$ where $\lim_{\ell \to \infty}\varepsilon'_{\ell i} = 0$. Also, recall that $H_m=\log m+\gamma+\varepsilon_m$ where $\gamma$ is Euler's constant and $\lim_{m\to \infty}\varepsilon_m = 0$. Thus:
        \begin{align*}
            \PAV(W_\ell) = \sum_{i} H_{u_i(W_\ell)} &= \sum_{i} \log u_i(W_\ell) +\varepsilon_{u_i(W_\ell)} + n\gamma \\
            &= \sum_{i} \log u_i(r) + n\log(|W_\ell|) + n\gamma + \bar{\varepsilon}_\ell \tag{as $\log u_i(W_\ell) =\log u_i(r) + \log |W_\ell| + \varepsilon'_{\ell i} $}\\
            &= \log \text{NW}(r) + n\log(|W_\ell|) + n\gamma + \bar{\varepsilon}_\ell
        \end{align*}

        Let $t$ be an arbitrary portioning. For every $\ell$, one can round $|W_\ell| t$ to a committee $T_\ell$ of size $|W_\ell|$ such that $T_\ell/|W_\ell|\to t$.
        If $u_i(t) = 0$ for some $i$, then the inequality $\text{NW}(t) \leq \text{NW}(r)$ trivially holds as $\text{NW}(t) = 0$. 
        So, suppose $u_i(t) > 0$ for every $i$. Then using the same argument as above, $\PAV(T_\ell) = \log\text{NW}(t) + n\log(|W_\ell|) + n\gamma + \hat{\varepsilon}_\ell$, and as $W_\ell$ is PAV optimal, $\PAV(W_\ell)\geq\PAV(T_\ell)$ implying $\text{NW}(r) \geq \text{NW}(t)$.

        Conversely, suppose that a portioning $r$ maximizes Nash welfare. First note that if $r'$ is another Nash-optimal portioning, we have $u_i(r) = u_i(r')$ for every $i$. If not, as $u_i(r), u_i(r') > 0$ for every $i$ and by strict concavity of $\sum_i\log u_i$ it holds that
        \[\log\text{NW}((r+r')/2)=\sum_i \log u_i((r+r')/2) > 1/2 \sum_i \log u_i(r) + 1/2\sum_i \log u_i(r')=\log\text{NW}(r),\]
        contradicting the optimality of $r$.
        
        Let matrix $B \in \{0,1\}^{(n+1) \times |P|}$ be such that $Bx=(u_1(x),\ldots,u_n(x),\sum_{j} x_j)$ for every portioning $x$. Basically, the first $n$ rows of $B$ are indicator vectors of approval sets $A_1, \dots, A_n$, and the last row is $\mathds{1}$.   
        For every $k\in\mathbb{N}$, choose a PAV-optimal committee $W_k$ for the instance $(N, P, A, k)$. By the first direction, every convergent subsequence of $(W_k/k)_{k\in\mathbb{N}}$ converges to a Nash-optimal portioning. Since $W_k/k\in [0,1]^P$ and $[0,1]^P$ is compact, at least one convergent subsequence exists. Fix a convergent subsequence, denoted $(W_\ell)_{\ell\in\mathbb{N}}$, with $W_\ell/|W_\ell|\to r'$ for some portioning $r'$. Then $r'$ is Nash optimal, and hence $Br'=Br$.

        Let $J=\{j \in P: r(j) + r'(j) > 0\}$ and $d = r-r'$. Then $Bd = Br - Br' = 0$ and as $d(j) = 0$ for every $j\notin J$ it holds that $B_Jd_J=0$, where $B_J$ is the restriction of $B$ to the columns in $J$, and similarly $d_J$ is the restriction of $d$ to coordinates in $J$. 
        Applying Gaussian elimination to $B_Jx=0$, and by integrality of $B_J$, we obtain rational vectors spanning all of its real solutions. Clearing denominators, we obtain integer vectors $h^1,\ldots,h^s\in \mathbb{Z}^{|J|}$ such that $\text{span}\{h^1, h^2, \dots, h^s\} = \{x: B_Jx = 0\}$. Therefore, we can write $d_J=\sum_{a=1}^s\alpha_ah^a$ for some real numbers $\alpha_1,\ldots,\alpha_s$.

        Fix $\delta\in(0,1)$. For every $\ell$ define $q_\ell=\sum_{a=1}^s\lfloor |W_\ell|(1-\delta)\alpha_a \rfloor h^a$, extended by zero outside $J$. Note that $Bq_\ell = 0$ and $\lim_{\ell \to \infty} q_\ell/|W_\ell| = (1-\delta)d$.
        Set $Z_\ell^\delta=W_\ell+q_\ell$. Since $W_\ell$ and $q_\ell$ are integral, $Z_\ell^\delta$ is integral and as $Bq_\ell=0$, it holds that $BZ_\ell^\delta = BW_\ell$ implying $|Z_\ell^\delta| = |W_\ell|$ and $u_i(Z_\ell^\delta)=u_i(W_\ell)$ for every $i$. 
        Thus, 
        \[
        \lim_{\ell \to \infty}  \frac{Z_\ell^\delta}{|Z_\ell^\delta|} = \lim_{\ell \to \infty} \frac{W_\ell}{|W_\ell|} + \lim_{\ell \to \infty}\frac{q_\ell}{|W_\ell|} = r' + (1-\delta) d = (1-\delta)r + \delta r'
        \]

         For $j \notin J$ we have  $q_\ell(j) = 0$ implying $Z_\ell^\delta(j) = W_\ell(j)\geq 0$, and for $j\in J$, the limiting point $(1-\delta)r(j) + \delta r'(j)$ is strictly positive. Hence $Z_\ell^\delta$ is nonnegative for all sufficiently large $\ell$, implying it is a valid PAV-optimal committee. Finally, let $0<\delta_m < 1/m$ and choose a PAV-optimal committee  $Z_m$ from $(Z_\ell^{\delta_m})_{\ell\in \mathbb{N}, \ell \geq m}$ such that $\lVert Z_m/|Z_m|-((1-\delta_m)r+\delta_m r')\rVert_\infty<1/m$. 
         Then $Z_m/|Z_m|\to r$ and $|Z_m| \to \infty$ implying $r$ is in the lift of PAV optimality.
\end{proof}

For the next result, we first describe the method of equal shares in the apportionment setting. 
\begin{description}
    \item[Method of Equal Shares.] It proceeds in rounds $t=1,2, \dots$. Initially, each voter $i$ has a budget $b_i(1) = 1$, the committee is $W = \emptyset$, and the price of a seat is $p = n/k$. Let $b_i(t)$ be the amount of leftover budget of $i$ just before iteration $t$. In iteration $t$, a party $j$ is said to be $q$-affordable if $\sum_{i: j\in A_i} \min(q, b_i(t)) \geq p$. If no party is $q$-affordable for any $q > 0$, the rule stops and returns $W$. Otherwise, we select a party $j$ that is $q$-affordable for a minimum $q$, and set $W=W+\{j\}$. Then, set the budget of each voter $i$ that $j\in A_i$ to $b_i(t + 1) = b_i(t) - \min(q, b_i(t))$, and set $b_i(t + 1) = b_i(t)$ for the rest of voters.
\end{description}

\begin{lemma}\label{lem:mes-stage}
    Consider a round $\ell$ of majoritarian portioning in which ties are broken according to a fixed priority order over parties. Let $R_\ell$ be
    the set of active voters immediately before this round, let
    $j_\ell$ be the selected party, and let
    \(N_\ell=\{i\in R_\ell:j_\ell\in A_i\}\).

    Suppose that there is a constant $C_\ell$, independent of
    $k$, such that for sufficiently large $k$, at some point in the execution of MES on the instance $(N, P, A, k)$, every voter
    \(i\in R_\ell\) has remaining budget at least \(1-C_\ell p_k\), while
    the voters outside $R_\ell$ have total remaining budget at most
    \(C_\ell p_k\) where $p_k=n/k$. 
 Starting from this point, define \emph{stage} $\ell$ as the consecutive sequence of MES iterations ending with the first iteration after which some voter in $N_\ell$ has remaining budget less than  \(p_k/|N_\ell|\). 
    Then there exist constants $E_\ell,D_\ell,B_\ell$, independent of
    $k$, such that:
    \begin{enumerate}
        \item at most $E_\ell$ copies of parties other than $j_\ell$ are
        selected during stage $\ell$;
        \item every voter in
        \(R_{\ell+1}=R_\ell\setminus N_\ell\) has remaining budget at
        least \(1-(C_\ell+E_\ell)p_k\);
        \item at the end of the stage, the voters in $N_\ell$ have total
        remaining budget at most \(D_\ell p_k\); and
        \item let $n^k_\ell$
    denote the number of copies of $j_\ell$ selected during this stage, then
        \(
            \left|
                n_\ell^k-kr(j_\ell)
            \right|
            \leq B_\ell.
        \)
    \end{enumerate}
\end{lemma}

\begin{proof}
    For all sufficiently large $k$, we have
    \(1-C_\ell p_k>p_k/|N_\ell|\). Thus, at the beginning of stage $\ell$,
    every voter in $N_\ell$ has more than \(p_k/|N_\ell|\) remaining, meaning the stage is ``well-defined''.
    Throughout the stage, every voter in $N_\ell$ has a budget of at least
    \(p_k/|N_\ell|\), and since these contributions sum to $p_k$, the party
    $j_\ell$ is \(p_k/|N_\ell|\)-affordable. Consequently, the affordability
    value $q$ of every party selected during the stage satisfies
    \(q\leq p_k/|N_\ell|\).
    
Call an iteration \emph{exceptional} if MES selects a party
    $j\neq j_\ell$, and let
    \(
        a_j:=|\{i\in R_\ell:j\in A_i\}|.
    \)
    By the choice of $j_\ell$, we have \(a_j\leq |N_\ell|\). 
    If \(a_j<|N_\ell|\), by integrality,
    \(a_j\leq |N_\ell|-1\), and the voters in $R_\ell$ contribute at most
    \[
        a_jq
        \leq
        (|N_\ell|-1)\frac{p_k}{|N_\ell|}
        =
        p_k-\frac{p_k}{|N_\ell|}.
    \]
    Therefore, the voters outside $R_\ell$ contribute at least
    \(p_k/|N_\ell|\). Since their total remaining budget is at most
    \(C_\ell p_k\), there can be at most \(C_\ell p_k / (p_k/|N_\ell|)= C_\ell |N_\ell|\) exceptional
    iterations of this type. Otherwise, \(a_j=|N_\ell|\), and therefore $j_\ell$ has priority over $j$.
    So MES can select $j$ only if its affordability value is strictly
    smaller than that of $j_\ell$. In particular,
    \(q<p_k/|N_\ell|\). The voters in $R_\ell$ therefore contribute
    strictly less than $q|N_\ell| < p_k$, so some voter outside $R_\ell$ must
    contribute.
    Moreover, the affordability value of every selected party satisfies
    \(q\geq p_k/n\), since
    \(
        p_k
        =
        \sum_{i:j\in A_i}\min\{q,b_i\}
        \leq nq,
    \)
    where $b_i$ denotes voter $i$'s remaining budget. In each such
    exceptional iteration, either some voter outside $R_\ell$ exhausts
    her budget, which occurs at most $n$ times, or no such voter is exhausted. In the latter case, every contributing voter
    outside $R_\ell$ pays $q$, so these voters spend at least
    \(q\geq p_k/n\). Since their total remaining budget is at most
    \(C_\ell p_k\), this case can occur at most \(C_\ell p_k /(p_k/n)  = nC_\ell\) times.
    Thus, the total number of exceptional iterations is at most
    \(
        E_\ell:=C_\ell |N_\ell|+n+nC_\ell.
    \)

    In each exceptional iteration, any voter spends at most
    \(q\leq p_k/|N_\ell| \leq p_k\). Hence every voter spends at most \(E_\ell p_k\) over
    all exceptional iterations. The voters in
    \(R_{\ell+1}=R_\ell\setminus N_\ell\) do not approve $j_\ell$, so
    they spend only during exceptional iterations. Each such voter
    therefore has remaining budget at least
    \(1-(C_\ell+E_\ell)p_k\) at the end of the stage.

    At the beginning of the stage, the budget of every voter in $N_\ell$ is at least $1-C_\ell p_k$ and at most $1$. Thus, the budgets of any two voters in
    $N_\ell$ differ by at most \(C_\ell p_k\). Whenever MES selects
    $j_\ell$, all voters in $N_\ell$ pay the same amount. Exceptional
    iterations can increase the difference between their budgets by at
    most \(E_\ell p_k\). Thus, throughout the stage, the budgets of any
    two voters in $N_\ell$ differ by at most
    \((C_\ell+E_\ell)p_k\).
    When the stage ends, some voter in $N_\ell$ has remaining budget less
    than \(p_k/|N_\ell|\). Therefore, every voter in $N_\ell$ has remaining
    budget at most
    \(
        \frac{p_k}{|N_\ell|}+(C_\ell+E_\ell)p_k.
    \)
    Hence their total remaining budget is at most \(D_\ell p_k\), where
    \(
        D_\ell:=1+|N_\ell|(C_\ell+E_\ell).
    \)

    Finally, at the beginning of the stage, the
    voters in $N_\ell$ have total remaining budget of at least
    \(|N_\ell|(1- C_\ell p_k)\). They spend at most \(E_\ell p_k\) in total on
    exceptional selections, and retain at most \(D_\ell p_k\) at
    the end of the stage. Thus, they spend at least
    \(
        |N_\ell|-(|N_\ell| C_\ell+E_\ell+D_\ell)p_k
    \)
    on copies of $j_\ell$.
    No voter in \(R_\ell\setminus N_\ell\) approves $j_\ell$, the voters in $N_\ell$ have total budget at most $N_\ell$, while the
    voters outside $R_\ell$ have total remaining budget at most
    \(C_\ell p_k\). Therefore, the total spent on copies of $j_\ell$ is at most $|N_\ell| + C_\ell p_k$, implying
    \(
        |N_\ell|-(|N_\ell| C_\ell+E_\ell+D_\ell)p_k
        \leq n_\ell^k p_k
        \leq |N_\ell|+C_\ell p_k.
    \)
    Dividing by \(p_k=n/k\) gives
    \[
        \frac{k |N_\ell|}{n}
        -(|N_\ell| C_\ell+E_\ell+D_\ell)
        \leq n_\ell^k
        \leq
        \frac{k |N_\ell|}{n}+C_\ell.
    \]
    Therefore, setting
    \(
        B_\ell
        :=
        \max\{|N_\ell| C_\ell+E_\ell+D_\ell,C_\ell\},
    \)
    we obtain
    \(
        \left|
            n_\ell^k-{k|N_\ell|}/{n}
        \right| = \left|
            n_\ell^k-{kr(j_\ell)}
        \right|
        \leq B_\ell.
    \)
\end{proof}

\begin{proposition}
    Fix a priority order over $P$, and use this order to break ties in both
    majoritarian portioning and the method of equal shares (MES). Then
    majoritarian portioning is the lift of the MES.
\end{proposition}

\begin{proof}
    Let $r$ be the portioning returned by majoritarian portioning, and
    suppose that the method terminates after $L$ rounds. For every
    $\ell\in\{1,\dots,L\}$, let $R_\ell$, $j_\ell$, and $N_\ell$ be as in
    \cref{lem:mes-stage}. For every $k\in\mathbb{N}$, let $W_k$ be the committee returned by MES
    on $(N,P,A,k)$.

    We apply \cref{lem:mes-stage} inductively. At the beginning of the
    execution of MES, every voter has budget $1$. Thus, the assumptions of
    the lemma hold for stage $1$ with $C_1=0$. Suppose that the assumptions of the lemma hold at the beginning of
    stage $\ell$ for some constant $C_\ell$. By
    \cref{lem:mes-stage}, at the end of this stage, every voter in
    $R_{\ell+1}=R_\ell\setminus N_\ell$ has remaining budget at least
    \(
        1-(C_\ell+E_\ell)p_k.
    \)
    Moreover, the voters outside $R_\ell$ have total remaining budget at
    most $C_\ell p_k$, while the voters in $N_\ell$ have total remaining
    budget at most $D_\ell p_k$. Hence, the voters outside
    $R_{\ell+1}$ have total remaining budget at most
    \(
        (C_\ell+D_\ell)p_k.
    \)
    Therefore, the assumptions of the lemma hold for stage $\ell+1$ with
    \(
        C_{\ell+1}
        :=
        C_\ell+\max\{E_\ell,D_\ell\}.
    \)

    After stage $L$, we have $R_{L+1}=\emptyset$, and the
    preceding induction shows that the voters have total remaining budget
    at most $C_{L+1}p_k$. Since every additional copy costs $p_k$, MES can
    select at most $C_{L+1}$ further copies after stage $L$.
     Also, the total number of exceptional iterations is bounded by $\sum_{\ell=1}^L E_\ell$, implying that \(0
        \leq
        W_k(j_\ell)-n_\ell^k
        \leq \sum_{\ell=1}^L E_\ell+C_{L+1}\) for every $\ell\in\{1, \dots, L\}$.
    Combining with \(\left|
            n_\ell^k-{kr(j_\ell)}
        \right|
        \leq B_\ell\),
    we obtain
    \(
        \left|W_k(j_\ell)-kr(j_\ell)\right|
        \leq
         \sum_{\ell=1}^L E_\ell+C_{L+1}+B_\ell
    \)
    for all sufficiently large $k$. Similarly, if $j \notin \{j_1, \dots, j_L\}$, we have $r(j) = 0$, and 
    \(
    0\leq
        W_k(j)
        \leq \sum_{\ell=1}^L E_\ell+C_{L+1}.
    \)
    Therefore,
    \(
        \left|W_k(j)-kr(j)\right|
        \leq
         \sum_{\ell=1}^L E_\ell+C_{L+1}+B_\ell,
    \) for every $j\in P$ and ${W_k(j)}/{k}\to r(j)$.
    Note that MES does not necessarily return a committee of size $k$, and $|W_k| \leq k$. However, since $\sum_{j\in P}r(j)=1$, it follows that
    \(
        {|W_k|}/{k}
        =
        \sum_{j\in P}{W_k(j)}/{k}
        \to
        \sum_{j\in P}r(j)
        =
        1.
    \)
    In particular, $|W_k|\to\infty$. Finally, for every $j\in P$,
    \[
        \frac{W_k(j)}{|W_k|}
        =
        \frac{W_k(j)/k}{|W_k|/k}
        \to
        r(j),
    \]
    and
    \(
        {W_k}/{|W_k|}\to r.
    \)
    
    It remains to verify that the above convergence characterizes the lift of MES according to \cref{def:lifting}\footnote{The subtlety is that $W_k$ is the outcome of MES on $(N, P, A, k)$, and not on $(N, P, A, |W_k|)$ as required by \cref{def:lifting}.}. First, consider committee sizes of the form \(k=mn\), where \(m\in\mathbb{N}\). In this case, \(p_k=n/k=1/m\), and MES exactly follows the rounds of majoritarian portioning. Indeed, inductively, at the beginning of round \(\ell\), every voter in \(R_\ell\) has budget \(1\), while every voter outside \(R_\ell\) has budget \(0\). The party \(j_\ell\) is \(1/(m|N_\ell|)\)-affordable, and no party has a smaller affordability value because \(j_\ell\) has the largest number of active supporters; ties are resolved in favor of \(j_\ell\) by the fixed priority order. Thus, MES selects exactly \(m|N_\ell|\) copies of \(j_\ell\), exhausting the budgets of the voters in \(N_\ell\), before proceeding to the next round. Consequently,
\(
    W_{mn}(j_\ell)
    =
    m|N_\ell|
    =
    mn\,r(j_\ell)
\)
for every \(\ell\), and hence \(|W_{mn}|=mn\) and \(W_{mn}/|W_{mn}|=r\). Therefore, \((W_{mn})_{m\in\mathbb{N}}\) witnesses that \(r\) satisfies the lift of MES.

Conversely, suppose that a portioning \(t\) satisfies the lift of MES, witnessed by a sequence \((V_m)_{m\in\mathbb{N}}\). Since ties are broken according to a fixed priority order, MES has a unique outcome for every committee size, and hence \(V_m=W_{|V_m|}\). As \(|V_m|\to\infty\), the convergence established above implies
\(
    {V_m}/{|V_m|}
    =
    {W_{|V_m|}}/{|W_{|V_m|}|}
    \to r.
\)
Since \(V_m/|V_m|\to t\) by assumption, we obtain \(t=r\). Hence, majoritarian portioning is precisely the lift of MES.
\end{proof}

\subsection{Implication Relations} \label{app:por:thms}
\begin{proposition}
Suppose that, on all apportionment instances, predicate $X$ implies predicate $Y$.
Then, $\lift{X}$ implies $\lift{Y}$ in all portioning instances.
Similarly, if for each apportionment instance, there exists a committee satisfying $X$, it must be true that, for each portioning instance, there exists a portioning satisfying $\lift{X}$.
\end{proposition}
\begin{proof}
Assume apportionment predicate \(X\) implies \(Y\).
Let the portioning \(r\) satisfy $\text{lift}(X)$ witnessed by committees \((W_\ell)_{\ell}\)
such that \(|W_\ell|\to\infty\) and \(W_\ell/|W_\ell|\to r\).
Each \(W_\ell\) satisfies \(X\); hence, by the assumption, each \(W_\ell\) also satisfies \(Y\).
Therefore \(r\) satisfies $\text{lift}(Y)$ witnessed by   \((W_\ell)_{\ell}\).

For existence, consider committees $W_k$ satisfying condition $X$ for each committee size $k$.
Consider the normalized sequence $(W_k / k) \in [0,1]^P$ for $k = 1, 2, \dots$.
Since $[0,1]^P$ is compact, there exists a convergent subsequence, whose limit satisfies $\lift{X}$ by definition.
\end{proof}

\begin{proposition}\label{prop:por:eq}
   The lifts of the following axioms are equivalent: bounded PAV improvement, local PAV optimality, PAV optimality, and Lindahl priceability.
\end{proposition}
\begin{proof}
The following implications are inherited from the apportionment setting \[ \lift{\text{PAV opt.}} \Rightarrow \lift{\text{local PAV opt.}}\Rightarrow \lift{\text{bounded PAV improvement}}\Leftrightarrow \lift{\text{Lindahl pr.}}\] 

By \cref{prop:lindahleq}, the lift of Lindahl priceability is Lindahl equilibrium. 
\citet{FGM16} showed that Lindahl equilibrium is equivalent to Nash optimality, which is the lift of PAV optimality by \cref{prop:lift}. This completes the proof.
\end{proof}

\begin{theorem}
    Every core stable portioning $r$ satisfies priceability.
\end{theorem}
\begin{proof}
    Consider a core stable portioning $r$. First, note that letting for every voter $i$, letting $|S|=\{i\}$ and $t=e_j/n$ for some $j\in A_i$, core stability implies $u_i(r) \geq 1/n > 0$.
    
    Applying \cref{lem:supp-dem} to the graph \(G=(N\cup P,E)\) where \((i,j)\in E\) if and only if \(j\in A_i\), it suffices to verify that for every $P' \subseteq P$,
    \[
         n\cdot \sum_{j \in P'}r(j) \leq \sum_{i: (A_i \cap P')\ne\emptyset}  1 = 
         |\{i: (A_i \cap P') \ne \emptyset\}|
    \]
    Fix a subset $P' \subseteq P$. 
    If $\sum_{j\notin P'} r(j)=0$, as $u_i(r) > 0$ for every $i$,
    we get $|\{i: (A_i \cap P') \ne \emptyset\}| = n$ and the inequality holds.
    
    Otherwise, let $S=\{i: A_i \subseteq P\setminus P'\}$, and define $t$ as follows
    \begin{align*}
        t(j) =  
        \begin{cases}
            \frac{r(j)}{\sum_{j'\notin P'} r(j')}\cdot \frac{|S|}{n} & j \notin P'\\
            0 & j \in P'
        \end{cases}
    \end{align*}
    So we get $\lVert t \rVert_1 = {|S|}/{n}$, and by core stability there exists $i\in S$ such that
    $
        u_i(r) \geq u_i(t) ={u_i(r)}/({\sum_{j\notin P'} r(j)}
        )\cdot {|S|}/{n}
    $
    where the equality holds because $i\in S$ and $A_i \subseteq P\setminus P'$. As $u_i(r) > 0$, by rearranging the terms we get  ${|S|}/{n}\leq \sum_{j \notin P'}r(j)$. Therefore $1-{|S|}/{n}\geq 1-\sum_{j \notin P'}r(j)$ implying ${|N\setminus S|}/{n}\geq \sum_{j \in P'}r(j)$  where $N\setminus S =\{i: A_i \cap P' \ne \emptyset\}$.
\end{proof}

We next confirm that there are no other implication relations in \cref{fig:portioning}. For the reader's convenience, the figure is shown again below. Each axiom label refers to the lift of the named axiom.

        \begin{adjustbox}{max width=\linewidth,center}
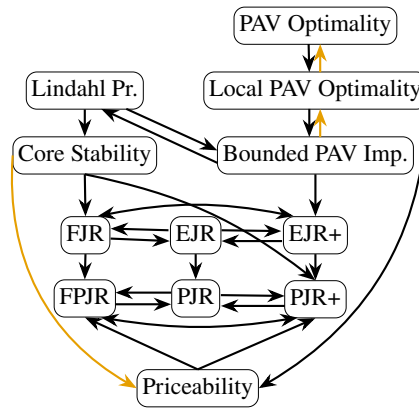
            \begin{tikzpicture}[
     x=0.9cm,
    y=1.1cm,
    box/.style={
        draw,
        rounded corners,
        very thin,
        node font=\footnotesize,
        minimum height=14pt,
        inner xsep=2pt,
        inner ysep=1pt,
        align=center
    }
]
  \node[box] (node1) at (-3.57,3.96) {Lindahl Pr.\vphantom{Ty}};
  \node[box] (node2) at (-3.57,3.18) {Core Stability\vphantom{Ty}};
  \node[box] (node5) at (-3.56,2.22) {FJR\vphantom{Ty}};
  \node[box] (node12) at (-1.94,2.22) {EJR\vphantom{Ty}};
  \node[box] (node6) at (-0.18,2.22) {EJR+\vphantom{Ty}};

  \draw[arrows=-Stealth, thick] (node1.south) -- (node2.north);
  \draw[arrows=-Stealth, thick] (-3.19,2.18) -- (-2.32,2.15);
  \draw[arrows=-Stealth, thick] (-0.67,2.15) -- (-1.56,2.15);

  \node[box] (node7) at (-3.57,1.46) {FPJR\vphantom{Ty}};
  \node[box] (node11) at (-1.94,1.46) {PJR\vphantom{Ty}};
  \node[box] (node8) at (-0.18,1.45) {PJR+\vphantom{Ty}};
  \node[box] (node10) at (-1.91,0.38) {Priceability\vphantom{Ty}};

  \draw[arrows=-Stealth, thick] (node2.south) -- (node5.north);
  \draw[arrows=-Stealth, thick] (node5.south) -- (node7.north);
  \draw[decorate, thick, arrows=-Stealth, bend left=17]
      (node2.south) to[bend left=13] (node8.north);
  \draw[arrows=-Stealth, thick] (node12.south) -- (node11.north);
  \draw[arrows=-Stealth, thick] (node10.north) -- (node8.south);
  \draw[arrows=-Stealth, thick] (node10.north) -- (node7.south);
  \draw[arrows=-Stealth, thick] (node6.south) -- (node8.north);

  \node[box] (node3) at (-0.18,3.18) {Bounded PAV Imp.\vphantom{Ty}};
  \node[box] (node4) at (-0.18,3.96) {Local PAV Optimality\vphantom{Ty}};

  \draw[arrows=-Stealth, thick, line cap=projecting]
      (-0.27,3.74) -- (-0.27,3.4);
  \draw[thick, line cap=projecting, arrows=-Stealth]
      (-1.56,2.28) -- (-0.68,2.27);
  \draw[arrows=-Stealth, thick]
      (-0.63,1.37) -- (-1.57,1.37);
  \draw[thick, line cap=projecting, arrows=-Stealth]
      (-1.57,1.5) -- (-0.63,1.49);

  \draw[arrows=-Stealth, thick, line cap=projecting]
      (node3.south) -- (node6.north);

  \draw[thick, line cap=projecting, arrows=Stealth-]
      (-3.19,2.3) -- (-2.34,2.28);

  \draw[arrows=-Stealth, thick]
      (-3.12,1.39) -- (-2.3,1.39);
  \draw[thick, line cap=projecting, arrows=Stealth-]
      (-3.12,1.53) -- (-2.3,1.53);

  \draw[decorate, thick, arrows=Stealth-Stealth, bend right=15]
      (-0.5,2.45) to[bend right=11] (-3.38,2.45);

  \draw[decorate, thick, arrows=Stealth-Stealth, bend left=15]
      (-0.44,1.23) to[bend left=9] (-3.33,1.23);

  \draw[thick, arrows=-Stealth]
      (-2.96,3.74) -- (-1.62,3.24);
  \draw[thick, arrows=Stealth-]
      (-3.33,3.74) -- (-1.62,3.09);
  \draw[arrows=-Stealth, thick, bend left=50]
      (node4.east) to[bend left=30] (node10.east);

  \node[box] (node9) at (-0.18,4.74) {PAV Optimality\vphantom{Ty}};

  \draw[arrows=-Stealth, thick, line cap=projecting]
      (-0.28,4.52) -- (-0.27,4.18);

  \draw[arrows=-Stealth, thick, bend left=50, portioningcolor]
        (node2.west) to[bend right=32] (node10.west);

  \draw[arrows=-Stealth, thick, portioningcolor] (-0.11,3.4) -- (-0.11,3.74);

  \draw[arrows=-Stealth, thick, portioningcolor] (-0.11,4.18) -- (-0.11,4.52);
\end{tikzpicture}
        \end{adjustbox}

\begin{example}\label{ex:por:lp}
    Let $n=3$ and \(P=\{a, b\}\). The voters have approval sets 
\(
A_1=\{a\},
A_2=\{b\},
A_3=\{a, b\}
\), and the selected portioning is $r=(1/3, 2/3)$.
\end{example}

\begin{proposition}
    Core stability does not imply Lindahl priceability.
\end{proposition}
\begin{proof}
    In \cref{ex:por:lp}, portioning $r$ satisfies core stability, but it does not satisfy  (lift of) bounded PAV improvement as \[\frac{\partial \log\text{NW}(r)}{\partial r(a)} = \sum_{i: a \in A_i} \frac{1}{u_i(r)}=4 > 3 =n.\] 
    By \cref{prop:por:eq}, $r$ is not Lindahl priceable.  
\end{proof}

\begin{example}\label{ex:ejr-core}
    Let $n=4$ and $P=\{a, b, c\}$. The voters have approval sets 
\(
A_1=\{a\},
A_2=\{b\},
A_3=\{c\},
A_4=\{b, c\}
\), and the selected portioning is $r=(1/3, 1/3, 1/3)$.
\end{example}

\begin{example}\label{app:por:ex:pjr}
    Let $n=4$ and $P=\{a, b, c, d, e\}$. The voters have approval sets
    \(
A_1 = \{a\},
A_2=\{b, d\},
A_3=\{c, e\},
A_4=\{d, e\}
\), and the selected portioning is $r=(1/4, 1/4, 1/4, 1/8, 1/8)$.
\end{example}

\begin{proposition}
    Priceability and EJR are incomparable.
\end{proposition}
\begin{proof}
    In \cref{ex:ejr-core}, portioning $r$ satisfies EJR, but is not priceable: voter 1 only approves $a$, implying $p_{1a}=1$, and $a$ is only approved by voter 1, implying $p_{1a} = n\cdot r(a) =4/3$.
    In \cref{app:por:ex:pjr}, portioning $r$ is priceable witnessed by \(
    p_{1a} = 1,
    p_{2b}=1, 
    p_{3c}=1, 
    p_{4d} = p_{4e} = 1/2,
    \) and all other prices equal to zero, but $S=\{3, 4\}$ blocks EJR.
\end{proof}
The above proposition implies that PJR, which is weaker than priceability, does not imply EJR either.

\begin{proposition}
    EJR does not imply core stability.
\end{proposition}
\begin{proof}
    In \cref{ex:ejr-core}, portioning $r$ satisfies EJR, but $S=\{2,3,4\}$ and $t=\{0, 3/8, 3/8\}$ blocks core stability.
\end{proof}

\section{Committee Elections}
\label{app:abc}
A committee election instance is a tuple $(N, C, A, k)$. By default, $i$ ranges over the set of voters $N$ and $j$ over the set of candidates $C$.
A \emph{committee} is a subset of candidates $W\subseteq C$ of size $|W| = k$.
We write $u_i(W) \coloneqq |W\cap A_i|$ for voter $i$'s \emph{utility} for committee $W$. $q(S)\coloneqq \lfloor|S| \cdot k /n\rfloor$ is the (Hare) \emph{quota} of a set of voters $S\subseteq N$.
\subsection{Definitions}\label{app:abc:definitions}
Fix an instance $(N, C, A, k)$ and a committee $W$. We say that a nonempty group $S\subseteq N$ is $\ell$-cohesive, for some
$\ell\in\mathbb{N}$, if $q(S)\geq\ell$ and
$\left|\bigcap_{i\in S}A_i\right|\geq\ell$.
Then:
\begin{description}
    \item[JR.] $W$ satisfies \emph{justified representation} if, for every 1-cohesive group $S\subseteq N$, there exists a voter $i \in S$ with $u_i(W) \geq 1$.
    \item[EJR.] $W$ satisfies \emph{extended justified representation} if, for every $\ell \in \mathbb{N}$ and every $\ell$-cohesive group $S \subseteq N$, there exists a voter $i \in S$ with $u_i(W) \geq \ell$.
    \item[\chg{$\alpha$-}EJR+.] Given $\chg{\alpha}\geq 1$, $W$ satisfies \emph{extended justified representation plus} if, for every nonempty $S \subseteq N$ with $(\bigcap_{i \in S} A_i)\setminus W \ne \emptyset$, there exists a voter $i \in S$ with $u_i(W) \geq q(S)\chg{/\alpha}$.
    \item[\chg{$\alpha$-}FJR.] Given $\chg{\alpha}\geq 1$, $W$ satisfies \chg{$\alpha$-}\emph{full justified representation} if, for every $\emptyset \neq S \subseteq N$, every $T \subseteq C$ of size $|T| \leq q(S)$ and every $\beta \in \mathbb{N}$ such that $u_i(T) \geq \chg{\alpha \cdot}\beta$ for all $i \in S$, there is $i \in S$ with $u_i(W) \geq \beta$.
    \item[PJR.] $W$ satisfies \emph{proportional justified representation}
    if, for every $\ell \in \mathbb{N}$ and every $\ell$-cohesive group $S \subseteq N$, it holds that $\left| W \cap \bigcup_{i\in S} A_i\right| \geq \ell$.
    \item[\chg{$\alpha$-}PJR+.] Given $\chg{\alpha}\geq 1$, $W$ satisfies \emph{proportional justified representation plus}
    if, for every nonempty $S \subseteq N$ with $(\bigcap_{i \in S} A_i)\setminus W \neq \emptyset$, it holds that $\left| W \cap \bigcup_{i\in S} A_i\right| \geq q(S)\chg{/\alpha}$.
     \item[\chg{$\alpha$-}FPJR.] Given $\chg{\alpha}\geq 1$, $W$ satisfies \chg{$\alpha$-}\emph{full proportional justified representation}
    if, for every nonempty $S \subseteq N$, every subset of candidates $T\subseteq C$ of size $|T| \leq q(S)$, and every $\beta \in \mathbb{N}$ such that $u_i(T) \geq \chg{\alpha \cdot} \beta$ for all $i \in S$, it holds that $\left| W \cap \bigcup_{i\in S} A_i\right| \geq \beta$.
    \item[Core Stability.] $W$ satisfies \emph{core stability}
    if, for every nonempty $S\subseteq N$ and every subset of candidates $T\subseteq C$ of size $|T| \leq q(S)$, there exists a voter $i\in S$ with $u_i(T) \leq u_i(W)$.
    \item[Lindahl Priceability.] $W$ is \emph{Lindahl priceable} if there exist nonnegative personalized prices $(p_{ij})_{i,j}$ such that $\sum_{i} p_{ij} \leq n/k$ for all $j$ and such that, for any voter $i$ and subset of candidates  $T\subseteq C$ with $\sum_{j'\in T} p_{ij'} \leq 1$, it holds that $u_i(T) \leq u_i(W)$. (Conceptually, $i$ cannot afford any outcome they prefer over $W$ within a budget of $1$.)

    \item[\chg{$\alpha$-Frugal} Lindahl Priceability.] Given $\chg{\alpha}\geq 1$, $W$ is \chg{$\alpha$-frugal} Lindahl priceable if there exist prices $(p_{ij})_{i,j} \geq 0$ such that (i)~$\sum_{i}p_{ij} \leq n/k$ for all $j$; (ii)~whenever $\sum_{j \in T} p_{ij} \leq 1$ for some $i$ and $T \subseteq C$, it holds that $u_i(T) \leq \chg{\alpha \cdot} u_i(W)$; and (iii)~\chg{$p_{ij} \leq p_{ij'}$ whenever $j \in A_i \cap W$ and $j' \in A_i \setminus W$}.
    \item[\chg{$\alpha$-}Priceability.] $W$ is \chg{$\alpha$-}priceable for $\alpha \geq 1$ if there exist nonnegative payments $(p_{ij})_{i,j}$ and a per-unit price $p \geq 0$ such that $p_{ij}=0$ whenever $j \notin A_i$, $\sum_{j} p_{ij} \leq 1$ for all $i$, $\sum_{i} p_{ij} = p$ for all $j\in W$, $\sum_{i} p_{ij} = 0$ for all $j\in C\setminus W$, and $\sum_{i : j \in A_i} (1 -\sum_{j' \in C} p_{ij'}) \leq \chg{\alpha\cdot} p$ for all $j \in C \setminus W$. (Conceptually, the unspent budget of $j$'s supporters is at most $\chg{\alpha\cdot} p$.)
    
    \item[Bounded PAV Improvement.] $W$ satisfies bounded PAV improvement if, $\Delta^+_W(j) < n/k$ for every $j \in C \setminus W$.
    \item[Local PAV Optimality.] $W$ is locally PAV optimal if, $\PAV(W{-}\{j\}{+}\{j'\}) \leq \PAV(W)$ for all candidates $j\in W$ and $j'\in C\setminus W$.
    \item[PAV Optimality.] $W$ is PAV optimal if, $\PAV(W') \leq \PAV(W)$ for every committee $W'$.
\end{description}

\subsection{Newly Introduced Axioms}\label{app:com:new}
\paragraph{Bounded PAV Improvement} Since we formally treat bounded PAV improvement as an axiom, we include, for completeness, the proofs that bounded PAV improvement implies EJR+ and is implied by local PAV optimality. The following was established implicitly in the proof that PAV satisfies EJR, before EJR+ was introduced in \citet{EJR+}.

\begin{proposition} [\cite{JR}]
    Bounded PAV improvement implies EJR+.
\end{proposition}
\begin{proof}
    We will prove the contrapositive. Consider a committee $W$ which does not satisfy EJR+, that is, there exists some $S \subseteq N$ and $c\in \cap_{i\in S} A_i \setminus W$ such that $u_i(W) < q(S)$ for every $i\in S$. Therefore, 
    \[
    \Delta^+_{W}(c) = \PAV(W \cup\{c\}) - \PAV(W) \geq \sum_{i\in S} \frac{1}{u_i(W) + 1} \geq \sum_{i\in S} \frac{1}{q(S)} = \frac{|S|}{\lfloor|S| \cdot k /n\rfloor} \geq \frac{n}{k},
    \]
    implying $W$ does not satisfy bounded PAV improvement.
\end{proof}

The following observation has been frequently used in the proofs related to PAV optimality \cite{JR, EJR+, paulapproval}.
\begin{observation}\label{app:com:prop:marginal}
    Given an instance $(N, C, A, k)$ and a committee $W$, the sum of marginal contributions of candidates in $W$ to its PAV score, $\sum_{j\in W} \Delta^-_W(j)$, is at most $n$. Consequently, there exists some candidate $c\in W$ such that 
    \[\Delta^-_W(c) = \PAV(W) - \PAV(W\setminus \{c\})\leq n/k.\]
\end{observation}
\begin{proof}
\begin{align*}
    \sum_{j\in W} \Delta^-_W(j)=\sum_{j\in W} \PAV(W) - \PAV(W\setminus \{j\})=\sum_{j\in W} \sum_{i: j\in A_i} \frac{1}{u_i(W)}
    = \sum_{i: u_i(W) > 0} \sum_{j\in A_i} \frac{1}{|A_i\cap W|}
    = \sum_{i: u_i(W) > 0} 1
    \leq n
\end{align*}
An averaging argument directly implies the existence of a candidate $c$.
\end{proof}

\begin{proposition}\label{app:com:prop:local}
    Local PAV optimality strictly implies bounded PAV improvement.
\end{proposition}
\begin{proof}
    We first prove the contrapositive of (weak) implication. Consider a committee $W$ which does not satisfy bounded PAV improvement, meaning there exists $c\in C\setminus W$ such that $\Delta^+_W(c)\geq n/k$. By \cref{app:com:prop:marginal}, there exists $c' \in W \cup \{c\}$, satisfying $\Delta^-_W(c')\leq n/(k+1)$. This implies
\begin{align*}
    \PAV(W\cup \{c\}\setminus \{c'\}) - \PAV(W)
    =\  \Delta^+_W(c)-\Delta^-_{W\cup \{c\}}(c')
    \geq\  n/k-n/(k+1) > 0,
\end{align*}
which means $W$ is not locally PAV optimal.

If the implication were not strict, then bounded PAV improvement would imply local PAV optimality. By \cref{prop:liftimply}, this implication would carry over to the apportionment setting, contradicting \cref{app:appor:pav}.
\end{proof}

From the above proof, we can observe that local PAV optimality implies $\Delta^+_W(c) \leq n/(k+1)$ for every \(c \in C\setminus W\). We nevertheless use the weaker bound \(<n/k\) in the definition of bounded PAV improvement for two reasons. First, so that a committee satisfying the axiom can be found in polynomial time. Second,  in the apportionment setting, this condition is equivalent to Lindahl priceability. Moreover, the following example shows that the bound in bounded PAV improvement is tight for guaranteeing EJR+. In particular, if the strict inequality is replaced by the weak inequality $\Delta_W^+(j)\leq n/k$ for every $j\in C\setminus W$, the resulting condition no longer implies EJR+.

\begin{example}
    Let $n=4$, $k=2$, and $C=\{a, b, c\}$. Approval sets are
    \(
    A_1= \{a\},
    A_2 = \{b\},
    A_3=A_4=\{c\}
    \), and the elected committee is $W=\{a, b\}$. Here $S=\{3, 4\}$ blocks EJR+, and we also have $\Delta^+_W(c) = 2 = n/k$.
\end{example}

The next result about proportionality degree provides another reason why bounded PAV improvement is an interesting axiom. For a more detailed discussion of proportionality degree, see \citet{Skowron21}.
\begin{proposition} If a committee \(W\) satisfies bounded PAV improvement, then \(W\) has optimal\footnote{\citet{AEH+18} showed a proportionality degree of larger than $l-1$ is not achievable.} proportionality degree of \(\ell-1\). That is, for every $\ell\in \mathbb{N}$, and every group \(S\subseteq N\) with size $|S| \geq \ell \cdot n/k$, it holds that if $|\bigcap_{i\in S} A_i| \geq \ell - 1$, then \( \frac{1}{|S|}\sum_{i\in S}u_i(W)\geq \ell-1. \)
\end{proposition} 
\begin{proof}
    Let $W$ be a committee  satisfying bounded PAV improvement,
    meaning that $\Delta_W^+(j)=\sum_{i:j\in A_i}1/(u_i(W)+1)< n/k$
    for every $j\in C\setminus W$. For $\ell=1$, the claim follows immediately from the nonnegativity of
    utilities. Thus, suppose $\ell\geq2$ and toward a contradiction,
    assume there exists a set $S$ satisfying $|S| \geq \ell\cdot n/k$, $|\bigcap_{i\in S} A_i| \geq \ell-1$, and $\frac{1}{|S|}\sum_{i\in S}u_i(W)<\ell-1$. If every candidate in $\bigcap_{i\in S} A_i$ was in $W$, we must have $u_i(W) \geq \ell-1$ for every $i\in S$, which is not true. So, there exists a candidate
    $c\in\bigcap_{i\in S}A_i\setminus W$. It holds that
    \begin{align*}
        \Delta_W^+(c)&\geq\sum_{i\in S}\frac{1}{u_i(W)+1}\tag{as every $i\in S$ approves $c$}\\
        &\geq  \frac{|S|^2}{\sum_{i\in S}(u_i(W)+1)} \tag{$\left(\sum_{i\in S}\frac{1}{u_i(W)+1}\right)\left(\sum_{i\in S}{(u_i(W)+1)}\right)\geq \left(\sum_{i\in S} 1\right)^2$ by Cauchy-Schwarz}\\
        &>  \frac{|S|^2}{|S| (\ell - 1) + |S|} \tag{as $\sum_{i\in S} u_i(W) < |S| (\ell - 1)$} \\
        &=  \frac{|S|}{\ell} 
        \geq \frac{n}{k}
    \end{align*}
    This
    contradicts bounded PAV improvement, and completes the proof.
\end{proof}

\begin{example}\label{ex:bpiii}
    Let $n=6$, $k=3$, and $C=\{a, b, c, d\}$. Approval sets are $A_1 = A_2 = A_3=\{a, c\}$, $A_4 = A_5 = \{b, c\}$ and $A_6 = \{d\}$, and the elected committee is $W = \{a,b,d\}$.
\end{example}

PAV is known to violate FPJR in some instances~\cite{FPJR}. Since PAV optimality is stronger than bounded PAV improvement by \cref{app:com:prop:local}, it follows that bounded PAV improvement implies neither FPJR nor any axioms stronger than FPJR, like Lindahl priceability. The following proposition implies that Lindahl priceability and bounded PAV improvement are incomparable, even though the two axioms coincide in the apportionment setting by \cref{thm:PAV-LP}.
\begin{proposition}\label{app:com:prop:lindhalpav}
Lindahl priceability does not imply bounded PAV improvement.
\end{proposition}
\begin{proof}
    In \cref{ex:bpiii} $W$ is Lindahl priceable witnessed by $p_{1a}=p_{2a} = p_{3a} = 2/3$, $p_{4b}=p_{5b}= 2/3$, $p_{1c}=p_{2c}=p_{3c}=p_{4c}=p_{5c} = 2/5$, and zero otherwise. 
    But it does not satisfy bounded PAV improvement as $\Delta^+_W(c) = 2.5 \geq n/k = 2$
\end{proof}
Most of the axioms in \cref{fig:abctoapp} are implied by Lindahl priceability and therefore, by \cref{app:com:prop:lindhalpav}, do not imply bounded PAV improvement. Moreover, priceability does not imply EJR+ and hence cannot imply bounded PAV improvement. It remains to show that EJR+, which is not implied by Lindahl priceability, also does not imply bounded PAV improvement.\\

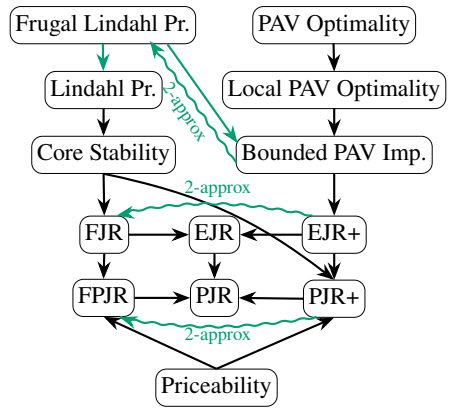
\begin{adjustbox}{max width=\linewidth,center}
            \begin{tikzpicture}[
     x=0.9cm,
    y=1.1cm,
    box/.style={
        draw,
        rounded corners,
        very thin,
        node font=\footnotesize,
        minimum height=14pt,
        inner xsep=2pt,
        inner ysep=1pt,
        align=center
    }
]
  \node[box] (node1) at (-3.57,3.96) {Lindahl Pr.\vphantom{Ty}};
  \node[box] (node13) at (-3.58,4.74) {Frugal Lindahl Pr.\vphantom{Ty}};
  \node[box] (node2) at (-3.57,3.18) {Core Stability\vphantom{Ty}};
  \node[draw=none, node font=\scriptsize, text=committeecolor, fill opacity=1, draw opacity=1] at (-1.89,2.77) {2-approx};
  \node[box] (node5) at (-3.56,2.22) {FJR\vphantom{Ty}};
  \node[box] (node12) at (-1.94,2.22) {EJR\vphantom{Ty}};
  \node[box] (node6) at (-0.18,2.22) {EJR+\vphantom{Ty}};

  \draw[arrows=-Stealth, thick] (node1.south) -- (node2.north);
  \draw[arrows=-Stealth, thick] (node5.east) -- (node12.west);
  \draw[arrows=-Stealth, thick] (node6.west) -- (node12.east);

  \node[box] (node7) at (-3.57,1.46) {FPJR\vphantom{Ty}};
  \node[box] (node11) at (-1.94,1.46) {PJR\vphantom{Ty}};
  \node[box] (node8) at (-0.18,1.45) {PJR+\vphantom{Ty}};
  \node[box] (node10) at (-1.91,0.38) {Priceability\vphantom{Ty}};

  \draw[arrows=-Stealth, thick] (node2.south) -- (node5.north);
  \draw[arrows=-Stealth, thick] (node5.south) -- (node7.north);
  \draw[decorate, thick, arrows=-Stealth, bend left=17]
      (node2.south) to[bend left=13] (node8.north);
  \draw[arrows=-Stealth, thick] (node12.south) -- (node11.north);
  \draw[arrows=-Stealth, thick] (node10.north) -- (node8.south);
  \draw[arrows=-Stealth, thick] (node10.north) -- (node7.south);
  \draw[arrows=-Stealth, thick] (node6.south) -- (node8.north);

  \node[box] (node3) at (-0.18,3.18) {Bounded PAV Imp.\vphantom{Ty}};
  \node[box] (node4) at (-0.18,3.96) {Local PAV Optimality\vphantom{Ty}};

  \draw[arrows=-Stealth, thick, line cap=projecting]
      (-0.18,3.74) -- (-0.18,3.4);
  \draw[arrows=-Stealth, thick]
      (node8.west) -- (node11.east);

  \draw[arrows=-Stealth, thick, line cap=projecting]
      (node3.south) -- (node6.north);

  \draw[arrows=-Stealth, thick]
      (node7.east) -- (node11.west);

  \draw[decorate, thick, arrows=-Stealth, bend right=15, committeecolor, decorate, decoration={snake}, /pgf/decoration/segment length=7pt, /pgf/decoration/amplitude=0.35pt]
      (-0.5,2.45) to[bend right=11] (-3.38,2.45);

  \draw[decorate, thick, arrows=-Stealth, bend left=15, committeecolor, decorate, decoration={snake}, /pgf/decoration/segment length=7.0pt, /pgf/decoration/amplitude=0.35pt]
      (-0.44,1.23) to[bend left=9] (-3.33,1.23);

  \draw[thick, arrows=-Stealth, committeecolor]
      (-2.63,4.51) -- (-1.60,3.32);
  \draw[thick, arrows=Stealth-, committeecolor, decorate, decoration={snake}, /pgf/decoration/segment length=6.6pt, /pgf/decoration/amplitude=0.4pt]
      (-2.89,4.55) -- (-1.62,3.09);

  \node[box] (node9) at (-0.18,4.74) {PAV Optimality\vphantom{Ty}};

  \draw[arrows=-Stealth, thick, line cap=projecting]
      (-0.19,4.52) -- (-0.18,4.18);

  \draw[arrows=-Stealth, thick, committeecolor] (node13.south) -- (node1);

  \node[draw=none, node font=\scriptsize, text=committeecolor, fill opacity=1, draw opacity=1, rotate=-56] at (-2.39,3.69) {2-approx};

  \node[draw=none, node font=\scriptsize, text=committeecolor, fill opacity=1, draw opacity=1] at (-1.89,1) {2-approx};
\end{tikzpicture}
        \end{adjustbox}

\begin{proposition}\label{app:com:prop:ejr+}
EJR+ does not imply bounded PAV improvement.
\end{proposition}
\begin{proof}
    Suppose, for contradiction, that EJR+ implies bounded PAV improvement. By \cref{prop:liftimply}, this implication would carry over to the apportionment setting. By \cref{thm:efjr+} EJR+ coincides with EJR in apportionment, and therefore, from \cref{app:appor:prop:PEJR} we conclude EJR+ does not imply core stability and therefore the stronger notion of Lindahl priceability, in the apportionment setting. By \cref{thm:PAV-LP}, bounded PAV improvement coincides with Lindahl priceability, a contradiction.
\end{proof}

Finally, for convenience of the reader, the polynomial procedure for finding a committee satisfying bounded PAV improvement is given in \cref{alg:bounded-pav}. The polynomial running time of this procedure was first established by \citet{localpav}. By \cref{app:com:prop:marginal}, each iteration increases the PAV score by at least $n/k - n/(k+1)$. Since PAV score is bounded by $n\cdot H(k)$, the algorithm terminates after $O(k^2\log k)$ iterations.

\begin{algorithm}[t]
\caption{Satisfying Bounded PAV Improvement}
\label{alg:bounded-pav}
\textbf{Input}: A committee election $(N,C,A,k)$\\
\textbf{Output}: A committee $W\subseteq C$
\begin{algorithmic}[1]
\STATE Choose an arbitrary committee $W\subseteq C$ of size $|W|=k$
\WHILE{there exist $c\in C\setminus W$ such that
\(
\Delta^+_W(c)
\geq \frac{n}{k}
\)}
    \STATE Find $c' \in W$ maximizing $\PAV(W \cup \{c\} \setminus \{c'\})$
    \STATE $W\gets W \cup \{c\} \setminus \{c'\}$
\ENDWHILE
\RETURN $W$
\end{algorithmic}
\end{algorithm}

\paragraph{Frugal Lindahl Priceability} We showed in \cref{thm:frugal} that (1-)frugal Lindahl priceability implies bounded PAV improvement, whereas \cref{app:com:prop:lindhalpav} establishes that Lindahl priceability alone does not. Thus, the frugality requirement provides a genuine and meaningful strengthening. We next show that this strengthening is nevertheless mild: in the apportionment setting, (1-)frugal Lindahl priceability and Lindahl priceability coincide. 

We first derive the definition of frugal Lindahl priceability in the apportionment setting using the embedding from \cref{def:paullift}.
\begin{proposition}
    Following \cref{def:paullift}, the definition of frugal Lindahl priceability in apportionment is as follows:
    \begin{description}
        \item[ ] Given an apportionment instance $(N, P, A, k)$, a committee $W:P\to \mathbb{N}$ is frugal Lindahl priceable if there exist prices $(p_{ij})_{i,j} \geq 0$ such that (i)~$\sum_{i}p_{ij} \leq n/k$ for all $j$; (ii)~whenever $\sum_{j} T(j)p_{ij} \leq 1$ for some $i$ and $T: P \to \mathbb{N}$, it holds that $u_i(T) \leq u_i(W)$; and (iii)~$p_{ij} \leq p_{ij'}$ for every $j, j'\in A_i$ such that $W(j) > W(j')$.
    \end{description}
\end{proposition}
\begin{proof}
    Following \cref{def:paullift}, given the apportionment instance $(N,P,A,k)$, committee $W$ is frugal Lindahl priceable if and only if $W^{\mathrm{cl}}$ is frugal Lindahl priceable for $(N,P^{\mathrm{cl}},A^{\mathrm{cl}},k)$.
    
    Assume that $W^{\mathrm{cl}}$ is frugal Lindahl priceable, witnessed by $(p'_{ic})_{i,c\in P^{\mathrm{cl}}}$. Fix a voter $i$, and let $\tau_i=\min\{p'_{ij^{(\ell)}}:j\in A_i,\ \ell>W(j)\}$ be the minimum price of an approved unelected clone. Such a clone always exists because every party has $k+1$ clones and $W(j)\leq k$. By condition (iii), every approved elected clone has price at most $\tau_i$, and hence $\sum_{c\in A_i^{\mathrm{cl}}\cap W^{\mathrm{cl}}}p'_{ic}\leq u_i(W)\tau_i$. Moreover, the set consisting of all approved elected clones together with an approved unelected clone of price $\tau_i$ gives voter $i$ utility $u_i(W)+1$. The utility-maximization condition therefore implies $\sum_{c\in A_i^{\mathrm{cl}}\cap W^{\mathrm{cl}}}p'_{ic}+\tau_i>1$. Consequently, $(u_i(W)+1)\tau_i\geq \sum_{c\in A_i^{\mathrm{cl}}\cap W^{\mathrm{cl}}}p'_{ic}+\tau_i>1$, and thus $\tau_i>1/(u_i(W)+1)$.
    For every $j\in A_i$ define
    \[
    p_{ij}=\min\left\{\frac{1}{k+1-W(j')}\sum_{\ell=W(j')+1}^{k+1}p'_{ij'^{(\ell)}}:j'\in A_i,\ W(j')\leq W(j)\right\},
    \]
    and set $p_{ij}=0$ for every $j\notin A_i$.
    
    For every $j\in P$, we have $p_{ij}\leq\sum_{\ell=W(j)+1}^{k+1}p'_{ij^{(\ell)}}$ whenever $j\in A_i$, while $p_{ij}=0$ otherwise. Therefore,
    \[
    \sum_i p_{ij}
    \leq
    \frac{1}{k+1-W(j)}
    \sum_{\ell=W(j)+1}^{k+1}\sum_i p'_{ij^{(\ell)}}
    \leq \frac{n}{k}.
    \]
    Thus, condition~(i) holds. Also, $p_{ij}\geq \tau_i>1/(u_i(W)+1)$ for every $j\in A_i$. Hence, if $T:P\to\mathbb{N}$ satisfies $u_i(T)>u_i(W)$, then
    \(
    \sum_jT(j)p_{ij}
    \geq
    u_i(T)\min_{j\in A_i}p_{ij}
    \geq
    \bigl(u_i(W)+1\bigr)\min_{j\in A_i}p_{ij}
    >1,
    \)
    establishing condition~(ii). 
    Finally, suppose that $j,j'\in A_i$ and $W(j)>W(j')$. Then $\{h\in A_i:W(h)\leq W(j')\}\subseteq\{h\in A_i:W(h)\leq W(j)\}$, and hence $p_{ij}\leq p_{ij'}$. Thus, condition~(iii) holds, and $W$ satisfies the stated definition.

    Conversely, assume that $W$ satisfies conditions~(i)--(iii) of the stated definition, witnessed by $(p_{ij})_{i,j}$. For every $i$ and $j\in A_i$, we must have $p_{ij}>1/(u_i(W)+1)$. Indeed, otherwise the multiset consisting of $u_i(W)+1$ copies of $j$ would cost at most $1$ and give voter $i$ utility $u_i(W)+1$, contradicting condition~(ii). Therefore it holds that $\min_{j'\in A_i}p_{ij'} > 1/(u_i(W)+1)$.
    Define prices for the cloned instance by setting
    \[
    p'_{ij^{(\ell)}}=
    \begin{cases}
    \min_{j'\in A_i}p_{ij'}, & \text{if $j\in A_i$ and $\ell\leq W(j)$},\\
    p_{ij}, & \text{if $j\in A_i$ and $\ell>W(j)$},\\
    0, & \text{if $j\notin A_i$}.
    \end{cases}
    \]
    For an elected copy $j^{(\ell)}$, we have $\sum_i p'_{ij^{(\ell)}}= \sum_i \min_{j'\in A_i}p_{ij'} \leq \sum_i p_{ij}\leq n/k$. For an unelected copy $j^{(\ell)}$, we similarly have $\sum_i p'_{ij^{(\ell)}}\leq\sum_i p_{ij}\leq n/k$. Next, for every $i$ and $c\in A_i^{\mathrm{cl}}$, it holds that $p'_{ic} \geq \min_{j'\in A_i}p_{ij'} > 1/(u_i(W)+1)$. Thus, any $T\subseteq P^{\mathrm{cl}}$ satisfying $u_i(T) > u_i(W)$ has total price at least $(u_i(W)+1)\min_{j'\in A_i}p_{ij'}>1$. 
    Finally, as $\min_{j'\in A_i}p_{ij'} \leq p_{ij}$ it holds that $p'_{ic}\leq p'_{ic'}$ whenever $c\in A_i^{\mathrm{cl}}\cap W^{\mathrm{cl}}$ and $c'\in A_i^{\mathrm{cl}}\setminus W^{\mathrm{cl}}$. Therefore, $W^{\mathrm{cl}}$ is frugal Lindahl priceable.
\end{proof}

\begin{proposition}\label{app:com:prop:lindh}
    In the apportionment setting, frugal Lindahl priceability and Lindahl priceability are equivalent.
\end{proposition}
\begin{proof}
    By \cref{prop:liftimply}, frugal Lindahl priceability implies Lindahl priceability in the apportionment setting. As bounded PAV improvement is equivalent to Lindahl priceability in apportionment, it suffices to prove that bounded PAV improvement implies frugal Lindahl priceability. This directly follows from the second direction of the proof of \cref{thm:PAV-LP}. Just notice that the suggested prices satisfy condition (iii) of frugal Lindahl priceability.
\end{proof}

We next show (1-)frugal Lindahl priceability is implied by stable priceability~\cite{PPS+21}. This strengthens the result by ~\cite{PPS+21} that stable priceability implies core stability.\footnote{To the best of our knowledge, it has not previously been observed that
stable priceability implies Lindahl priceability.} We start by the definition of stable priceability.

A committee $W$ is \emph{stable priceable} if there exist nonnegative payments $(p_{ij})_{i,j}$ and a per-unit price $p \geq 0$ such that (1) $p_{ij}=0$ whenever $j \notin A_i$, (2) $\sum_{j} p_{ij} \leq 1$ for all $i$, (3) $\sum_{i} p_{ij} = p$ for all $j\in W$, (4) $\sum_{i} p_{ij} = 0$ for all $j\in C\setminus W$, and (5) there exist no nonempty set of voters $S\subseteq N$, no
    nonempty set of candidates $T\subseteq C\setminus W$, no nonnegative
    payments $(p'_{ij})_{i,j\in T}$, and no sets
    $(R_i)_{i\in S}$ with $R_i\subseteq W$, such that
    \begin{enumerate}[label=(\alph*)]
        \item $\sum_{i\in S}p'_{ij}>p$ for every $j\in T$;
        \item $\sum_{j\in W\setminus R_i} p_{ij}+\sum_{j\in T}p'_{ij}\leq 1$ for
        every $i\in S$; and
        \item for every $i\in S$, either
        \(
            u_i\bigl((W\setminus R_i)\cup T\bigr)>u_i(W),
        \)
        or
        \(
            u_i\bigl((W\setminus R_i)\cup T\bigr)=u_i(W)
        \)
        and
        \(
            \sum_{j\in W\setminus R_i}p_{ij}+\sum_{j\in T}p'_{ij}<\sum_{j\in W}p_{ij}.
        \)
    \end{enumerate}
\citet{PPS+21} showed that condition (5) can equivalently be represented as 
\begin{equation}\label{eq:stableprice}
    \sum_{i:c\in A_i}\max\left\{1-\sum_{j}p_{ij},\max_{j}p_{ij}\right\}\leq p \qquad \text{for every $c\in C\setminus W$}
\end{equation}
\begin{proposition}
    Stable priceability implies (1-)frugal Lindahl priceability.
\end{proposition}
\begin{proof}
    Consider a stable priceable committee $W$ witnessed by payments $(p_{ij})_{i, j}$ and per-unit price $p$. First, note that summing the payments over all candidates gives $k\cdot p=\sum_{j}\sum_i p_{ij}
        =\sum_i\sum_{j}p_{ij}\leq n$, 
    implying $p\leq n/k$.  For every voter $i$ let
    $b_i=1-\sum_{j}p_{ij}$ be their remaining budget, and 
    \(
        M_i=\max\left\{b_i,\max_{j}p_{ij}\right\} > 0.
    \)
    
    If $p<n/k$, choose $\varepsilon>0$ sufficiently small such that
    $p+n\varepsilon\leq n/k$. If $p=n/k$, let $\varepsilon=0$. Define
    personalized prices $(q_{ij})_{i,j}$ by
    \[
        q_{ij}=
        \begin{cases}
            p_{ij} & j\in W\\
            M_i+\varepsilon & j\notin W\text{ and }j\in A_i\\
            0 & j\notin W\text{ and }j\notin A_i.
        \end{cases}
    \]

    We now verify $(q_{ij})_{i,j}$ is a certificate of (1-)frugal Lindahl priceability. For every
    $j\in W$, it holds that $\sum_i q_{ij}=\sum_i p_{ij}=p\leq n/k$.
    For every $j\notin W$ using \cref{eq:stableprice} it holds that 
    $\sum_i q_{ij}
    =\sum_{i:j\in A_i}(M_i+\varepsilon)
    \leq p+n\varepsilon\leq n/k$, meaning condition (i) of frugal Lindahl priceability holds.
    To verify condition (ii), fix a voter $i$ and a set $T\subseteq C$ such that $u_i(T)>u_i(W)$. As $|A_i \cap T| = u_i(T) > u_i(W) = |A_i \cap W|$, we have 
    \begin{equation}\label{eq:setminus}
        |(A_i\cap T)\setminus W| \geq |(A_i\cap W)\setminus T| + 1.
    \end{equation}
    
    It holds that
    \begin{align*}
        \sum_{j\in T}q_{ij}
        &\geq
        \sum_{j\in (A_i\cap T)\cap W}p_{ij} + \sum_{j\in (A_i\cap T)\setminus W} (M_i + \varepsilon)\\
        &\geq\sum_{j\in (A_i\cap W)\cap T}p_{ij} + \sum_{j\in (A_i\cap W)\setminus T} (M_i + \varepsilon) + (M_i + \varepsilon)\tag{by \cref{eq:setminus}}\\ 
        &\geq \sum_{j\in A_i\cap W}p_{ij} + (b_i + \varepsilon) \tag{as $M_i \geq p_{ij}$ for every $j$ and $M_i \geq b_i$}\\
        &= 1+\varepsilon  \tag{as $\sum_{j\in A_i\cap W}p_{ij} + b_i = \sum_{j}p_{ij} + b_i =1$}\\
        &\geq 1
    \end{align*}
    If $\epsilon > 0$, the last inequality is strict, and if $\epsilon=0$, we have $p=n/k$, and the total spent budget is $n$, meaning $b_i=0$ for every $i$. Therefore, $M_i > b_i$ and the third inequality above is strict. So, $\sum_{j\in T}q_{ij} > 1$ and condition (ii)  holds as well.

    Finally, consider any voter $i$, and candidates $c\in A_i\cap W$ and
    $c'\in A_i\setminus W$. By construction,
    $q_{ic}=p_{ic}\leq M_i\leq M_i+\varepsilon=q_{ic'}$. Hence condition
    (iii) is also satisfied.
\end{proof}

Stable priceability was introduced to strengthen the notion of priceability. Unfortunately, priceability is not implied by frugal Lindahl priceability.
\begin{proposition}\label{app:com:prop:flinhprice}
    (1-)frugal Lindahl priceability does not imply priceability.
\end{proposition}
\begin{proof}
    Suppose, for contradiction, that frugal Lindahl priceability implies priceability. By \cref{prop:liftimply}, this implication would carry over to the apportionment setting. However, by \cref{app:com:prop:lindh}, frugal Lindahl priceability coincides with Lindahl priceability in apportionment, while \cref{app:appor:lindhprice} shows that Lindahl priceability does not imply priceability in this setting.
\end{proof}

\begin{proposition}
    (1-)frugal Lindahl priceability does not imply local PAV optimality.
\end{proposition}
\begin{proof}
    The reasoning is analogous to the proof of \cref{app:com:prop:flinhprice}, except we need the equivalence of Lindahl priceability and bounded PAV improvement (\cref{thm:PAV-LP}) and that bounded PAV improvement is strictly weaker than local PAV optimality in the apportionment setting (\cref{app:appor:pav}).
\end{proof}

\subsection{Deferred Proofs}
\label{app:abc:proofs}

\begin{proposition}
    In committee elections, PJR+ implies 2-FPJR.
\end{proposition}
    \begin{proof}
    The proof is identical to the proof for the first part of \cref{thm:2fjr}. Recreating the proof of \cref{prop:FPJR+}, we can already see that, if there were a ($1$-)FPJR violation $S, T, \beta$ such that $T \subseteq C \setminus W$, there must be a group of size at least $\beta \cdot n / k$ who all like one element of $T \in C\setminus W$, so a PJR+ violation.

    We now assume that $W$ violates $2$-FPJR and will deduce a PJR+ violation from that.
    We are given $S, T, \beta$ such that $u_i(T) \geq 2 \, \beta > 2 \, u_i(W)$ for all $i \in S$.
    Since $|A_i \cap W| < \beta$ and $|A_i \cap T| \geq 2 \, \beta$, $|A_i \cap (T \setminus W)| \geq \beta$.
    But then, $u_i(T \setminus W) \geq \beta > u_i(W)$ for all $i \in S$, so $S, T \setminus W, \beta$ witnesses a $1$-FPJR violation with a candidate set $T \setminus W$ disjoint from $W$, so we obtain a contradiction with PJR+ as explained above.
\end{proof}

We next show that none of the remaining implications established in the apportionment setting hold approximately in committee elections.

Given a committee election \((N, C, A, k)\), we define the \emph{utility vector} of a committee \(W\) as
\(
\mathbf{u}_A(W)
=
\bigl(\lvert A_1\cap W\rvert,\ldots,\lvert A_n\cap W\rvert\bigr).
\)
A possibly irresolute\footnote{Meaning that the voting rule maps the given instance to a set of committees and not necessarily a single committee.} voting rule $f$ is \emph{welfarist} if, for every \(k\), there exists a function \(g_k\) mapping utility vectors to real values such that,
\[
f(N, C, A, k)
=
\operatorname*{arg\,max}_{W\subseteq C:\,\lvert W\rvert=k}
g_k\bigl(\mathbf{u}_A(W)\bigr).
\]
We say that a welfarist rule is Pareto optimal if the functions $g_k$ are strictly increasing, and it is easy
to see that the corresponding committee rule is then Pareto optimal in the usual sense. 

Unlike the apportionment setting, where PAV satisfies priceability as shown in \cref{thm:pavpriceable}, \citet{limitsofwelf} showed no Pareto optimal and welfarist voting rule satisfies priceability in committee elections. In fact, the following more general result also holds.

\begin{proposition}
    For any constant $\alpha \geq 1$, there exists no Pareto optimal welfarist rule that would always return an $\alpha$-priceable committee.
\end{proposition}
\begin{proof}
\begin{figure*}[t]
\centering

\newcommand{\sixcells}[3]{\foreach \r in {0,...,5}{
        \filldraw[fill=#3,draw=black]
        ({#1+0.8*\r},#2) rectangle ++(0.8,0.5);
    }
}

\begin{subfigure}[t]{0.48\textwidth}
    \centering
    \begin{tikzpicture}[scale=0.85,transform shape]

\filldraw[fill=blue!10!white,draw=black]
            (0,0) rectangle (4.8,0.5);
        \node at (2.4,0.25) {$c_1$};

        \filldraw[fill=blue!10!white,draw=black]
            (0,0.5) rectangle (4.8,1);
        \node at (2.4,0.75) {$c_2$};

        \filldraw[fill=blue!10!white,draw=black]
            (0,1) rectangle (4.8,1.5);
        \node at (2.4,1.25) {$c_3$};

\sixcells{0}{1.5}{blue!10!white}
        \sixcells{0}{2}{blue!10!white}
        \sixcells{0}{2.5}{blue!10!white}
        \sixcells{0}{3}{blue!10!white}
        \sixcells{0}{3.5}{white}

        \node at (2.4,4.25) {$\cdots$};

        \sixcells{0}{4.5}{white}

\sixcells{4.8}{0}{blue!10!white}
        \sixcells{4.8}{0.5}{blue!10!white}
        \sixcells{4.8}{1}{blue!10!white}
        \sixcells{4.8}{1.5}{blue!10!white}
        \sixcells{4.8}{2}{blue!10!white}
        \sixcells{4.8}{2.5}{white}

        \node at (7.2,3.25) {$\cdots$};

        \sixcells{4.8}{3.5}{white}

\foreach \i/\x in {
            1/0.4,2/1.2,3/2.0,4/2.8,5/3.6,6/4.4,
            7/5.2,8/6.0,9/6.8,10/7.6,11/8.4,12/9.2
        }{
            \node at (\x,-0.35) {$v_{\i}$};
        }

    \end{tikzpicture}
    \caption{Instance 1}
    \label{fig:alpha-welfarism-1}
\end{subfigure}
\hfill
\begin{subfigure}[t]{0.48\textwidth}
    \centering
    \begin{tikzpicture}[scale=0.85,transform shape]

\filldraw[fill=blue!10!white,draw=black]
            (4.8,0) rectangle (9.6,0.5);
        \node at (7.2,0.25) {$c_1$};

        \filldraw[fill=blue!10!white,draw=black]
            (4.8,0.5) rectangle (9.6,1);
        \node at (7.2,0.75) {$c_2$};

        \filldraw[fill=blue!10!white,draw=black]
            (4.8,1) rectangle (9.6,1.5);
        \node at (7.2,1.25) {$c_3$};

\sixcells{0}{0}{blue!10!white}
        \sixcells{0}{0.5}{blue!10!white}
        \sixcells{0}{1}{blue!10!white}
        \sixcells{0}{1.5}{blue!10!white}
        \sixcells{0}{2}{blue!10!white}
        \sixcells{0}{2.5}{white}

        \node at (2.4,3.25) {$\cdots$};

        \sixcells{0}{3.5}{white}

\sixcells{4.8}{1.5}{blue!10!white}
        \sixcells{4.8}{2}{blue!10!white}
        \sixcells{4.8}{2.5}{blue!10!white}
        \sixcells{4.8}{3}{blue!10!white}
        \sixcells{4.8}{3.5}{white}

        \node at (7.2,4.25) {$\cdots$};

        \sixcells{4.8}{4.5}{white}

\foreach \i/\x in {
            1/0.4,2/1.2,3/2.0,4/2.8,5/3.6,6/4.4,
            7/5.2,8/6.0,9/6.8,10/7.6,11/8.4,12/9.2
        }{
            \node at (\x,-0.35) {$v_{\i}$};
        }

    \end{tikzpicture}
    \caption{Instance 2}
    \label{fig:alpha-welfarism-2}
\end{subfigure}

\caption{Each rectangle represents a block of $t$ identical candidates. Each voter approves more than $k=57t$ candidates, and $k$ candidates are colored to illustrate a committee. The figure is reproduced from \citet{limitsofwelf}.}
\label{fig:alpha-welfarism}
\end{figure*}
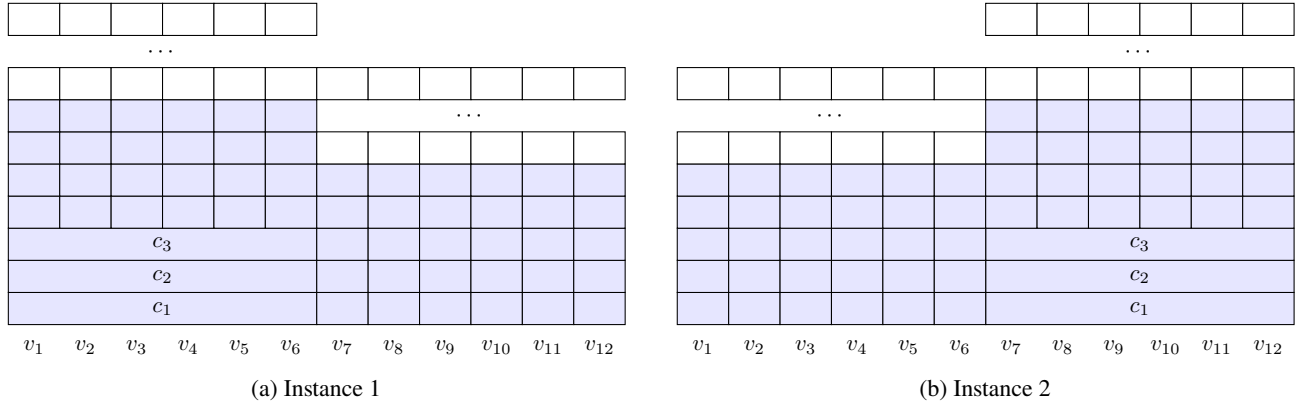

    Fixing $\alpha \geq 1$, the proof by \citet{limitsofwelf} goes through by creating $t$ clones of each candidate and setting $k=57t$ for some integer $t \geq \alpha$. For completeness, we include the complete proof below.

Consider Instance~1 shown in \cref{fig:alpha-welfarism}. There are $12$ voters, the committee size is $k=57t$, and each voter approves at least $57t$ candidates. Suppose, for contradiction, that there exists a Pareto optimal welfarist rule $f$ that always returns $\alpha$-priceable committees. Let $g_k$ be as in the definition of welfarist rules, and $W$ be an $\alpha$-priceable outcome of $f$ on instance~1, witnessed by $(p,(p_{ij})_{i, j})$. 

We first show that each of the last six voters has strictly fewer than $6t$ representatives in $W$. Suppose, toward a contradiction, that one of these voters, say $v_7$, has at least $6t$ representatives. 
Since no other voter approves the candidates approved by $v_7$, this voter must pay for them alone, and hence $p\leq 1/(6t)$. 
As $t \geq \alpha$, this implies $\alpha p \leq 1/6$. Each of the other five voters must have at least $5t$ representatives. Indeed, a voter with fewer than $5t$ representatives would have a remaining budget greater than $1-5tp\geq 1/6\geq \alpha p$, contradicting $\alpha$-priceability.
Thus, at least $6t+5\times 5t=31t$ seats are filled by candidates approved by voters $v_7,\dots,v_{12}$, leaving at most $57t-31t=26t$ seats for the remaining candidates. Hence, the total spending of voters $v_1,\dots,v_6$ is at most $26tp\leq 26/6$. Their average remaining budget is therefore at least
\(
{(6-26/6)}/{6}>\frac{1}{6}\geq\alpha p.
\)
Consequently, one of these voters has a remaining budget greater than $\alpha p$ and can afford an unelected private candidate, again contradicting $\alpha$-priceability.

Second, we show that each of the last six voters has strictly more than $3t$ representatives. Suppose, toward a contradiction, that one of these voters, say $v_7$, has at most $3t$ representatives. Since $v_7$ approves an unelected private candidate, $\alpha$-priceability restricts their remaining budget to be at most $\alpha p$, that is, $1-3tp\leq\alpha p$, and therefore
\(
p\geq{1}/{(3t+\alpha)}\geq{1}/{(4t)}.
\)
It follows that each of the other five voters has at most $4t$ representatives. Hence, at most $3t+5\times4t=23t$ candidates of $A_7\cup\cdots\cup A_{12}$ belong to $W$. This leaves at least $57t-23t=34t$ seats for candidates approved by voters $v_1,\dots,v_6$. At most $3t$ of these are clones of $c_1,c_2,c_3$, so $W$ contains at least $31t$ private candidates approved by the first six voters. Thus, one of these voters approves at least $31t/6>5t$ private candidates in $W$. Since no other voter approves these candidates, we obtain $p<1/(5t)$, contradicting $p\geq1/(4t)$.

Summarizing, each of the last six voters has strictly between $3t$ and $6t$ representatives in every $\alpha$-priceable committee. Since fewer than $36t$ members of $W$ are approved by voters $v_7,\dots,v_{12}$, more than $21t$ members are approved by voters $v_1,\dots,v_6$. Since $f$ is Pareto optimal, all $3t$ clones of $c_1,c_2,c_3$ must belong to $W$. Indeed, if one of these clones were not selected, it could replace a selected private candidate approved by one of the first six voters, making the other five voters strictly better off without making anyone worse off.
Therefore, $W$ contains more than $18t$ private candidates approved by voters $v_1,\dots,v_6$. Together with the $3t$ common candidates, these voters have total utility greater than $18t+6\times 3t=36t$, and hence average utility greater than $6t$. Let $\mathbf{u}_1$ denote the welfare vector of $W$.

Consider Instance~2, which is obtained from Instance~1 by exchanging the roles of the first and last six voters. By the same reasoning, in every committee selected by $f$, each of the first six voters has strictly between $3t$ and $6t$ representatives, all $3t$ clones of $c_1,c_2,c_3$ are selected, and the last six voters have average utility greater than $6t$. Let $\mathbf{u}_2$ denote the welfare vector of such a committee.

Observe that both $\mathbf{u}_1$ and $\mathbf{u}_2$ are achievable in both instances. To preserve a utility vector when moving from one instance to the other, select all $3t$ common candidates, remove $3t$ private representatives from each voter who gains the common candidates, and add $3t$ private representatives for each voter who loses them. This is possible because every coordinate of $\mathbf{u}_1$ and $\mathbf{u}_2$ is greater than $3t$, and it preserves the total committee size. In Instance~1, the vector $\mathbf{u}_2$ cannot be selected, since each of its first six coordinates is smaller than $6t$, whereas every selected committee gives these voters average utility greater than $6t$. Therefore, $g_k(\mathbf{u}_1)>g_k(\mathbf{u}_2)$. By the symmetric argument applied to Instance~2, we obtain $g_k(\mathbf{u}_2)>g_k(\mathbf{u}_1)$, a contradiction.

 \end{proof}

\begin{proposition}\label{app:com:prop:approxejr}
    EJR does not imply any constant approximation of FPJR or PJR+.
\end{proposition}
\begin{proof}
    Fix a prime \(q\), let \(\mathbb{Z}_q=\{0,1,\dots,q-1\}\), and \(t=\left\lceil\log_2(q^2+1)\right\rceil.\) Consider the instance \((N,C,A,k)\) with \(n=k=q^2\) and
\(
N=\mathbb{Z}_q^2=\{(a,b):a,b\in\mathbb{Z}_q\}.
\) Since $q^2 \leq 2^t-1$, we can assign each voter $(a, b)$ a distinct nonempty subset $R_{(a,b)} \subseteq R=\{r_1, r_2, \dots, r_t\}$.
The candidate set is \(C=R\cup I\cup D\), where
\(
R=\{r_1, \dots, r_t\}
\)
is the set of \emph{regular candidates},
\(
I=\{c_{u,v}:u,v\in\mathbb{Z}_q\}
\)
is the set of \emph{irregular candidates}, and \(
D=\{d_1, \dots, d_{k-t}\}
\) is the set of \emph{dummy candidates}. Voter \((a,b)\) approves the following candidate:
\[
A_{(a,b)}
=
R_{(a,b)} \cup \{c_{u,(au+b)\bmod q}:u\in\mathbb{Z}_q\}.
\]

Since \(q\) is prime, every irregular candidate is approved by exactly \(q\) voters, and any two distinct voters have at most one approved irregular candidate in common. Indeed, if \(a\neq c\), the equations \(au+b=cu+d\pmod q\) have a unique solution \(u\in\mathbb{Z}_q\); if \(a=c\) and \(b\neq d\), they have no solution.

Consider the committee \(W=R\cup D\), which consists of all regular and dummy candidates. Consider any $\ell$-cohesive group $S$ with $|S| \geq 2$. We have $|\cap_{(a,b)\in S} A_{(a, b)}| \geq \ell$, and at most one of the common candidates is irregular. Thus, $|\cap_{(a,b)\in S} R_{(a, b)}| \geq \ell - 1$. On the other hand, because all $R_{(a, b)}$ are distinct, for at least one voter $(a, b)\in S$ we have $R_{(a, b)} \not\subseteq \cap_{(c,d)\in S} R_{(c, d)}$, implying $\ell \leq |R_{(a, b)}|= |A_{(a, b)} \cap W|$ and $S$ does not block EJR. If $|S| = 1$, then as $n=k$, every voter deserves one candidate, so $S$ must be $1$-cohesive. However, every voter approves at least one regular candidate in $W$, and a $1$-cohesive group cannot block EJR. 

On the other hand, every irregular candidate \(c_{u, v}\) is approved by \(q\) voters, and under PJR+ these voters must collectively approve at least $q$ candidates in $W$. But the voters approve at most $|R| = t$ candidates, implying \(W\) violates \(\alpha\)-PJR+ for every \(\alpha<q/t\).
Also, consider \(S=N\) and \(T=I\). We have $q(S)=q^2=|T|$, every voter \((a,b)\in S\) approves exactly \(q\) candidates in \(T\), and they collectively approve at most $|R|=t$ voters in in $W$, implying \(W\) violates \(\alpha\)-FPJR for every \(\alpha<q/t\).
Finally, note that $q/t = q/ \left\lceil\log_2(q^2+1)\right\rceil \to \infty$ as $q\to \infty$ completing the proof.
\end{proof}

As PJR is implied by EJR, $\alpha$-EJR+ implies $\alpha$-PJR+, and $\alpha$-FJR implies $\alpha$-FPJR, we can directly conclude the following from \cref{app:com:prop:approxejr}.
\begin{corollary}
    EJR does not imply any constant approximation of FJR or EJR+, and PJR does not imply any constant approximation of FPJR or PJR+.
\end{corollary} 
\end{document}